%% file: main.tex
\documentclass{article}
\usepackage{fullpage}
\usepackage{amsmath,amsthm}
\usepackage{authblk}

 \usepackage{libertine}

\usepackage{microtype}
\usepackage{amsfonts,amssymb}
\newtheorem{theorem}{Theorem}
\newtheorem{proposition}{Proposition}

\newtheorem{definition}{Definition}
\newtheorem{claim}{Claim}
\newtheorem{lemma}{Lemma}
\newtheorem{corollary}{Corollary}
\newtheorem{bound}{Bound}
\newtheorem{fact}{Fact}
\newtheorem*{fact*}{Fact}

\theoremstyle{definition}

\newtheorem{observation}{Observation}

\usepackage{tikz}
\usetikzlibrary{automata,arrows,positioning,calc}
\usepackage{pgfplots}
\usepackage{subfiles}
\usepackage[numbers, sectionbib]{natbib}
\usepackage[inline]{enumitem}
\usepackage{hhline}
\usepackage{multirow}
\usepackage{mathtools}

\usepackage{hyperref}

\newcommand{\Ignore}[1]{}

\newcommand{\gf}[1]{\mathsf{#1}}
\newcommand{\Z}{\mathbb{Z}}
\newcommand{\ZZ}{\mathbb{Z}}
\newcommand{\R}{\mathbb{R}}
\newcommand{\RR}{\mathbb{R}}
\newcommand{\NN}{\mathbb{N}}
\newcommand{\hdepth}{\mathbf{d}}
\DeclareMathOperator{\length}{length}
\DeclareMathOperator{\depth}{depth}
\DeclareMathOperator{\height}{height}
\DeclareMathOperator{\gap}{gap}
\DeclareMathOperator{\reach}{reach}
\DeclareMathOperator{\reserve}{reserve}

\newcommand{\Slot}{\textsl{sl}}
\newcommand{\slot}{\Slot}
\newcommand{\fprefix}{\sqsubseteq}
\newcommand{\Adversary}{\mathcal{A}}
\newcommand{\Challenger}{\mathcal{C}}
\newcommand{\Distribution}{\mathcal{D}}
\newcommand{\dominatedby}{\preceq}

\newcommand{\SuchThat}{:}
\newcommand{\Given}{\mid}
\newcommand{\Union}{\cup}

\newcommand{\Fork}{\vdash}
\newcommand{\ForkPrefix}{\sqsubseteq}

\newcommand{\PrefixEq}{\preceq}
\newcommand{\Prefix}{\prec}
\newcommand{\DominatedBy}{\preceq}
\newcommand{\Chain}{\mathcal{C}}
\newcommand{\Intersect}{\cap}
\newcommand{\SlotCP}{\mathrm{CP}^{\mathsf{slot}}}
\newcommand{\kSlotCP}[1][k]{{#1}\text{-}\SlotCP}
\newcommand{\CP}{\mathrm{CP}}
\newcommand{\kCP}[1][k]{{#1}\text{-}\CP}
\newcommand{\defeq}{\triangleq}
\newcommand{\TrimSlot}[1]{^{\lfloor {#1}}}
\newcommand{\SlotDivergence}{\mathrm{div}_{\mathsf{slot}}}
\newcommand{\pinch}[2]{{#2}^{\vartriangleright {#1} \vartriangleleft} }
\newcommand{\cut}[2]{{#2}^{{#1} \vartriangleleft} }
\newcommand{\h}{\mathtt{h}}
\renewcommand{\H}{\mathtt{H}}
\newcommand{\A}{\mathtt{A}}
\newcommand{\Hheavy}{\h\H\text{-heavy}}
\newcommand{\Aheavy}{\A\text{-heavy}}

\newcommand{\Reduce}{\rho_\Delta}
\newcommand{\DeltaFork}{\Fork_\Delta}

\makeatletter
\renewcommand\AB@affilsepx{\quad\protect\Affilfont} 
\makeatother

\author[1,3]{Aggelos Kiayias}
\author[2]{Saad Quader}
\author[2,3]{Alexander Russell}

\affil[1]{University of Edinburgh}
\affil[2]{University of Connecticut}
\affil[3]{IOHK}

\begin{document}

\title{
Consistency of 
Proof-of-Stake Blockchains
with Concurrent Honest Slot Leaders
}

\maketitle

\begin{abstract}
  We improve the fundamental security threshold of 
  eventual consensus Proof-of-Stake
  (PoS) blockchain protocols under longest-chain rule, 
  reflecting for the first time the
  positive effect of rounds with concurrent honest leaders.  Current
  analyses of these protocols 
  reduce consistency to the dynamics of an abstract, round-based block
  creation process that is determined by three probabilities:
  \begin{itemize}
  \item $p_\A$, the probability that a round has at least one adversarial leader;
  \item $p_\h$, the probability that a round has a single honest leader; and
  \item $p_\H$, the probability that a round has multiple, but honest, leaders.
  \end{itemize}
  We present a 
  consistency analysis that achieves the optimal threshold $p_\h + p_\H > p_\A$. 
  This is a first in the literature and can be applied
  to both the simple synchronous setting and the setting with bounded
  delays. 
  Moreover, we achieve the optimal consistency error $e^{-\Theta(k)}$ where $k$ is the confirmation time. 
  We also provide an efficient algorithm to explicitly calculate these error probabilities in the synchronous setting.

  All existing consistency analyses either incur a penalty
  for rounds with concurrent honest leaders, 
  or treat them neutrally. 
  Specifically, the consistency analyses in Ouroboros Praos (Eurocrypt 2018) and
  Genesis (CCS 2018) assume that the probability of a uniquely honest
  round exceeds that of the other two events combined 
  (i.e., $p_\h - p_\H > p_\A$); the analyses in
  Sleepy Consensus (Asiacrypt 2017) and Snow White (Fin. Crypto 2019)
  assume that a uniquely honest round is more likely than an
  adversarial round 
  (i.e., $p_\h > p_\A$). 
  In addition,
  previous analyses completely break down when uniquely honest rounds
  become less frequent, i.e., $p_\h < p_\A$. 
  These
  thresholds determine the critical trade-off between the honest majority,
  network delays, and consistency error.

  Our new results can be directly applied to improve the 
  consistency guarantees of the existing protocols. 
  We complement these results with a consistency analysis in the setting
  where uniquely honest slots are rare, even letting $p_\h = 0$, under
  the added assumption that honest players adopt a consistent chain
  selection rule.

\end{abstract}


\section{Introduction}
\input{intro}

\section{The model and our main theorems}\label{sec:model}
\input{model}

\section{Unique Vertex Property via Catalan slots}\label{sec:definitions}
\input{definitions}
\input{fork}
\subsection{Catalan slots and the UVP}\label{sec:catalan}

\input{catalan-cp}

\section{Main theorems via tail bounds for Catalan slots}\label{sec:bounds-main-proofs}
\input{bounds}

\section{Proofs of Bounds~\ref{bound:unique-honest-catalan} and~\ref{bound:two-catalans}}\label{sec:estimates}

\input{estimates}

\section{An optimal online adversary against slot settlement}
\label{sec:recursion}
\input{recursive-formulation}

\subsection{An algorithm to compute exact settlement probabilities}
\label{sec:exact-prob}
\input{exact-probabilities}

\section{Proofs of Theorem~\ref{thm:relative-margin} and Theorem~\ref{thm:opt-adversary-canonical}}
\label{sec:margin-proof}

\input{margin-proof}

\section{The semi-synchronous setting}\label{sec:async}\label{sec:async-model}
\input{async}

\section{The common prefix property}\label{sec:cp}\label{sec:cp-model}
\input{cp}

\section*{Acknowledgments}
We thank Peter Ga\v{z}i (IOHK) for finding a bug in Fact~\ref{fact:fork-structure} in a previous version of this paper.

\bibliography{forks,abbrev0,crypto_crossref}

\appendix


\break
\newpage
\section{CP violations and balanced forks with concurrent honest leaders}
\label{sec:cp-forks}
\input{cp-forks}



\end{document}

%% file: intro.tex
Proof-of-Stake (PoS) blockchain protocols have emerged as a viable
alternative to resource-intensive Proof-of-Work (PoW) blockchain
protocols such as Bitcoin and Ethereum. These PoS protocols are
organized in rounds (which we call \emph{slots} in this paper); their
most critical algorithmic component is a leader election procedure
which determines---for each slot---a subset of participants with the
authority to add a block to the blockchain. Existing security analyses
of these protocols are logically divided into two components: the
first reasons about the properties of the leader election process, the
second reasons about the combinatorial properties of the blockchains
that can be produced by an \emph{idealized} leader schedule in the
face of adaptive adversarial control of some participants. An
attractive side effect of this structure is that the combinatorial
considerations can be treated independently of other aspects of the
protocol. A recent article of Blum et al.~\cite{LinearConsistency}
gave an axiomatic treatment of this combinatorial portion of the
analysis which we extend in this paper.

These common combinatorial arguments can be formulated with very
little information about the leader election process.  Specifically,
current analyses focus on three parameters:
\begin{itemize}
\item $p_\h$, the probability that a slot is \emph{uniquely honest}, having a single honest leader;
\item $p_\H$, the probability that a slot is \emph{multiply honest}, having multiple, but honest,
  leaders; and
\item $p_\A$, the probability that a slot has at least one adversarial leader.
\end{itemize}
Our major contribution is a generic, rigorous guarantee of consistency
under the most desirable assumption\footnote{ 
  Consistency is unachievable in the case $p_\h + p_\H < p_\A$.
  See \cite{GK18} for a detailed discussion of the honest majority
  assumption. } $p_\h + p_\H > p_\A$ that achieves optimal consistency
error $\exp(-\Theta(k))$ as a function of confirmation time $k$. Our
analysis can be directly applied to existing protocols to improve
their consistency guarantees.

To contrast this with existing literature, the analysis of Ouroboros
Praos~\cite{Praos} and Ouroboros Genesis~\cite{Genesis}
require the threshold assumption $p_\h - p_\H > p_\A$ to achieve the
optimal consistency error of $e^{-\Theta(k)}$. Note how multiply
honest slots actually \emph{detract} from security, appearing
negatively in the basic security threshold. The consistency analyses
in Snow White~\cite{SnowWhite} and Sleepy Consensus~\cite{Sleepy}
assume an improved threshold $p_\h > p_\A$; however, they only
establish a consistency error bound of $e^{-\Theta(\sqrt{k})}$. Note
here that multiply honest slots appear neutrally. All existing
analyses break down if $p_\h < p_\A$, i.e., when the uniquely
honest slots are less probable than the adversarial slots.


Multiply honest slots may arise by design, e.g., when each player
checks privately whether he is a leader.  They may also occur
naturally in the non-synchronous setting when the time between the
broadcast of two blocks is exceeded by network delay---in this case
the party issuing the later block may not be aware of the earlier
block which can result the two blocks sharing the same chain history,
a de facto incidence of multiple honest leaders. The role of these
slots is rather delicate: while it is good for the system to have many
honest blocks, \emph{concurrent} blocks can help the adversary in
creating two long, diverging blockchains that might jeopardize the
consistency property. Our new analysis shows that this second effect
can be mitigated, achieving consistency error bound of
$e^{-\Theta(k)}$ under the (tight) assumption $p_\h + p_\H > p_\A$.

\paragraph{Our results and contributions.} 
As described above, we show for the first time that PoS blockchain
protocols using the longest-chain rule can achieve a consistency error
of $e^{-\Theta(k)}$ under the desirable condition
$p_\h + p_\H > p_\A$.  This improves the security guarantee of all
``longest chain rule''  PoS protocols such as Praos~\cite{Praos},
Genesis~\cite{Genesis}, and Snow White~\cite{SnowWhite}
(we remark that other PoS protocols such as Algorand~\cite{DBLP:journals/corr/Micali16} 
operate in a different setting where explicit participation bounds are assumed
and forks can be prevented).
We discuss
our results in more detail before turning to the model and proofs.

Our analysis in the simple synchronous model achieves the same
asymptotic error bound as in~\cite{LinearConsistencySODA}---the
tightest result in the literature---under a much weaker assumption,
namely $p_\h + p_\H > p_\A$.  Thus PoS protocols can in fact achieve
consistency with $p_\h < p_\A$, a regime beyond reach of all previous analyses. 
When uniquely honest slots are rare 
(i.e., when $p_\h$ is very small), 
our bound has the desired dependence on $p_\h$. 
Moreover, when $p_\H = 0$ (i.e., all honest slots are in fact
uniquely honest), we exactly recover the bound
in~\cite{LinearConsistencySODA}. 
We also give an algorithm to explicitly compute the probability 
that a given slot encounters a consistency violation 
under the idealized leader election mechanism. 
The time and space required by this algorithm is cubic 
in the length of the protocol execution.

Next, we consider a variant model where the honest players use a
consistent tie-breaking rule when selecting the longest chain.  (I.e.,
when a fixed set of blockchains of equal length are presented to a
collection of honest players, they all select the same chain.
In previous models, the adversary had the right to break such ties by influencing
network delivery.)
Assuming $p_\h + p_\H > p_\A$, we prove that the consistency error
bound in this model is identical to the $e^{-\Theta(k)}$ bound
in~\cite{LinearConsistencySODA} \emph{even when $p_\h =
  0$}. No existing analysis survives in this regime.


\paragraph{$\Delta$-synchronous setting.}
In the $\Delta$-synchronous
communication setting, all messages are delivered with at most
a $\Delta$ delay. Our results mentioned above can be transferred to
this setting using the \emph{$\Delta$-synchronous to synchronous reduction
approach} used in the Ouroboros Praos analysis~\cite{Praos}. Thus, we
can achieve a consistency error probability of $e^{-\Theta(k)}$ in this
setting as well. 
This analysis is presented in 
Section~\ref{sec:async}.

\paragraph{A technical overview.}
We initially work in the synchronous communication model and extend
the synchronous combinatorial framework
of~\cite{LinearConsistency} to accommodate multiply honest
slots. 

First, our analysis focuses on a combinatorial event called a ``Catalan
slot.''\footnote{The name is a nod to the \emph{Catalan number} in
  combinatorics: The $n$th Catalan number $C_n$ is the number of
  strings $w \in \{0, 1\}^{2n}$ so that every prefix $x$ of $w$
  satisfies $\#_0(x) \geq \#_1(x)$.} Catalan slots are honest slots
$c$ with the property that any interval containing $c$ possesses
strictly more honest slots---with any number of honest leaders---than
adversarial ones. The analysis of~\cite{SnowWhite} and ~\cite{Sleepy}
introduced this basic concept, though they counted only uniquely
honest slots. In comparison with their analysis, then, our treatment
has two important advantages: first of all, we let multiply honest
slots count in the analysis and, additionally, we achieve strikingly
stronger error bounds: specifically, we achieve optimal settlement
error of $\exp(-\theta(k))$ rather than $\exp(-\theta(\sqrt{k}))$.

A Catalan slot $c$ acts as a barrier for the adversary in that if an
honest blockchain from a slot $h < c$ is padded with adversarial
blocks and presented to an honest observer at slot $c + 1$, the
observer will never adopt this blockchain.  As a result, the chains
adopted by this honest observer must contain \emph{some} block from
slot $c$.  Note that this is true \emph{even if $c$ is
  multiply honest}.  A critical observation is that \emph{a slot is
  Catalan if and only if all competitive blockchains in future slots
  contain at least one block from this slot}.  Thus, if a Catalan slot
$c$ is uniquely honest, all blockchains that are eligible to be
adopted by future honest players must contain the (only) honest block
issued from slot $c$.  We call this the ``Unique Vertex Property''
(UVP).  Note how the UVP is reminiscent of the ``Common Prefix
Property'' (CP) in the literature. Thus, together, the UVP and 
Catalan slots act as a conduit between consistency
violations and the underlying stochastic process. 

Our major technical challenge is to bound the probability that Catalan
slots are infrequent. Here we break away entirely from the analysis
of~\cite{SnowWhite} and approach the question using the theory of
generating functions and stochastic dominance. We find an exact
generating function for a related event and use this, by dominance, to
control the undesirable event that a long window of slots is devoid of
Catalan slots. This yields
asymptotically optimal settlement bounds.

Finally, it follows from the discussion above that if two consecutive
slots are Catalan then any subsequent honest block must contain, in
its prefix, a block from each of these slots.  In a setting where all
honest players use a consistent longest-chain selection rule,
we show that both slots have UVP as well.  Since Catalan slots can be
multiply honest, PoS protocols can achieve a consistency error bound
of $e^{-\Theta(k)}$ in this model even if $p_\h = 0$.

In a separate line of reasoning, in Section~\ref{sec:recursion}, 
we generalize the fork-theoretic framework of~\citet{LinearConsistency} for the multi-leader setting. 
Here, we characterize the UVP 
in terms of the so-called ``relative margin,'' 
a combinatorial property of a given slot. 
We describe an adversary who optimally attacks the UVP 
of all slots, simultaneously. 
Next, we prove a recurrence relation for relative margin. 
Suppose each slot is 
independently and identically chosen 
(by the leader election mechanism) 
to be either uniquely honest, multiply honest, or adversarial. 
The recurrence relation mentioned above then 
leads to an algorithm to explicitly compute 
the probability that 
a given slot encounters a consistency violation; 
see Section~\ref{sec:exact-prob}. 
In contrast, the Catalan slot-centric characterization of the UVP 
gives us only an asymptotic bound on this probability. 
It can be concluded that the fork-framework, after all, 
is expressive enough to capture consistency violations 
in the multi-leader setting.

\paragraph{Outline.}
We specify our model in Section~\ref{sec:model} and focus on a
specific consistency property called ``$k$-settlement.''  This section
also contains our main theorems; the proofs are deferred to
Section~\ref{sec:bounds-main-proofs}.  In
Section~\ref{sec:definitions}, we describe amplifications
to the fork framework of~\cite{LinearConsistency} in order to
explore the relationship between Catalan slots and the UVP. 
In Section~\ref{sec:bounds-main-proofs},
we present two bounds on the stochastic events of interest, e.g., the
rarity of a Catalan slot; these bounds lead to short proofs of the
main theorems.  The proofs of these bounds are presented next in
Section~\ref{sec:estimates} which contains all of our probabilistic
arguments.  

Section~\ref{sec:recursion} contains an alternative treatment 
of the UVP via fork-theoretic notions of~\cite{LinearConsistency}. 
Along the way, it describes an optimal adversary who simultaneously attacks the consistency of all slots. 
It also describes an algorithm to compute explicit values 
for the probability of consistency violations. 
The proofs of two important theorems from this section 
are presented subsequently in Section~\ref{sec:margin-proof}.

Our treatment of the $\Delta$-synchronous setting is
presented in Section~\ref{sec:async}.  In Section~\ref{sec:cp}, we
treat the traditional Common Prefix (CP) violations using our bounds
on the UVP.  

In Appendix~\ref{sec:cp-forks}, 
we characterize common prefix violations 
in the presence of multiply honest slots 
using ``balanced forks'' from~\cite{LinearConsistency} (and, importantly, 
without using Catalan slots).


%% file: model.tex
We study the behavior of the elementary \emph{longest-chain rule}
algorithm, carried out by a collection of participants:
\begin{itemize}
  \item In each round,
each participant collects all valid blockchains from the network; if a
participant is a leader in the round, he adds a block to the longest
chain and broadcasts the result.
\end{itemize}
Here, ``valid'' indicates that any block appearing in the
chain was indeed issued by a leader from the associated slot; in the
PoS setting, this property is guaranteed with digital
signatures.

We begin by studying this algorithm in the simple, synchronous model
posited by Blum et. al~\cite{LinearConsistency}. The model adopts a
synchronous communication network in the presence of a \emph{rushing}
adversary: in particular,
\begin{enumerate}[label={\textbf{A\arabic*}}., ref={\textbf{A\arabic*}}, series=axiom, start = 0]
  \item\label{axiom:message-delivery} 
  Any message broadcast by an honest participant at the beginning of a
  particular slot is received by the adversary first, who may decide
  strategically and individually for each recipient in the network
  whether to inject additional messages and in which order all messages
  are to be delivered prior to the conclusion of the slot. 
\end{enumerate}
See the comments prior to Section~\ref{sec:model-settlement} for
further discussion of this network assumption.  A variant of this
adversarial message-ordering is presented in
Section~\ref{sec:lcr-model}.  The $\Delta$-synchronous communication
model is handled in Section~\ref{sec:async}.

Given this, it is easy to describe the behavior of the longest-chain
rule when carried out by a group of honest participants with the extra
guarantee that exactly one is elected as leader in a slot: Assuming
that the system is initialized with a common ``genesis block''
corresponding to $\slot_0$, the players observe a common, linearly
growing blockchain:
\begin{center}
  \begin{tikzpicture}[>=stealth', auto, semithick,
    flat/.style={circle,draw=black,thick,text=black,font=\small}]
    \node[flat]    at (0,0)  (base) {$0$};
    \node[flat]    at (1,0)  (n1) {$1$};
    \node[flat]    at (2,0)  (n2) {$2$};
    \node[flat,white]    at (3,0)  (n3) {$\ \ \  $};
    \node at (3,0) {$\ldots$};
    \draw[thick,->] (base) to (n1);
    \draw[thick,->] (n1) to (n2);
    \draw[thick,->] (n2) to (n3);
  \end{tikzpicture}
\end{center}
\noindent
Here node $i$ represents the block broadcast by the leader of slot $i$
and the arrows represent the direction of increasing time.

\paragraph{The blockchain axioms: Informal discussion.}
The introduction of adversarial participants or multiple slot leaders
complicates the family of possible blockchains that could emerge from
this process. To explore this in the context of our protocols, we work
with an abstract notion of a blockchain which
ignores all internal structure. We consider a fixed assignment of
leaders to time slots, and assume that the blockchain uses a proof
mechanism to ensure that any block labeled with slot $\slot_t$ was
indeed produced by a leader of slot $\slot_t$; this is guaranteed in
practice by appropriate use of a secure digital signature scheme.

Specifically, we treat a \emph{blockchain} as
a sequence of abstract blocks, each labeled with a slot number, so
that:
\begin{enumerate}[label={\textbf{A\arabic*}}., ref={\textbf{A\arabic*}}, resume=axiom]
  \item\label{axiom:root} 
  The blockchain begins with a fixed ``genesis'' block, assigned to slot $\slot_0$.
  
  \item\label{axiom:labels} 
  The (slot) labels of the blocks are in strictly increasing order.
\end{enumerate}
It is further convenient to introduce the structure of a directed
graph on our presentation, where each block is treated as a vertex; in
light of the first two axioms above, a blockchain is a path beginning
with a special ``genesis'' vertex, labeled $0$, followed by vertices
with strictly increasing labels that indicate which slot is associated
with the block. 
\begin{center}
  \begin{tikzpicture}[>=stealth', auto, semithick,
    flat/.style={circle,draw=black,thick,text=black,font=\small}]
    \node[flat]    at (0,0)  (base) {$0$};
    \node[flat]    at (1,0)  (n1) {$2$};
    \node[flat] at (2,0)  (n2) {$4$};
    \node[flat] at (3,0)  (n3) {$5$};
    \node[flat] at (4,0)  (n4) {$7$};
    \node[flat] at (5,0)  (n5) {$9$};
    \draw[thick,->] (base) to (n1);
    \draw[thick,->] (n1) to (n2);
    \draw[thick,->] (n2) to (n3);
    \draw[thick,->] (n3) -- (n4);
    \draw[thick,->] (n4) -- (n5);
  \end{tikzpicture}
\end{center}
The protocols of interest call for honest players to add a
\emph{single} block 
during any slot. In particular:
\begin{enumerate}[label={\textbf{A\arabic*}}., ref={\textbf{A\arabic*}}, resume=axiom]
  \item\label{axiom:honest}
   Let $k \geq 1$ be an integer. 
  If a slot $\slot_t$ was assigned to $k$ honest players but no adversarial players, 
  then $k$ blocks are created---during the entire protocol---each having the label $\slot_t$.
\end{enumerate}
Recall that blockchains are \emph{immutable} in the sense that any
block in the chain commits to the entire previous history of the
chain; this is achieved in practice by including with each block a
collision-free hash of the previous block. These properties imply that
any chain that includes a block issued by an honest player 
must also include that block's associated prefix in its entirety.

As we analyze the dynamics of blockchain algorithms, it is convenient
to maintain an entire family of blockchains at once. As a matter of
bookkeeping, when two blockchains agree on a common prefix, we can
glue together the associated paths to indicate this, as shown
below.
\begin{center}
  \begin{tikzpicture}[>=stealth', auto, semithick,
    flat/.style={circle,draw=black,thick,text=black,font=\small}]
    \node[flat]    at (0,0)  (base) {$0$};
    \node[flat]    at (1,0)  (n1) {$2$};
    \node[flat] at (2,0)  (n2) {$4$};
    \node[flat] at (3,0)  (n3) {$5$};
    \node[flat] at (4,.5)  (n4a) {$7$};
    \node[flat] at (5,.5)  (n5a) {$9$};
    \node[flat] at (4,-.5)  (n4b) {$8$};
    \node[flat] at (5,-.5)  (n5b) {$9$};
    \draw[thick,->] (base) to (n1);
    \draw[thick,->] (n1) to (n2);
    \draw[thick,->] (n2) to (n3);
    \draw[thick,->] (n3) to (n4a);
    \draw[thick,->] (n4a) to (n5a);
    \draw[thick,->] (n3) to (n4b);
    \draw[thick,->] (n4b) to (n5b);
  \end{tikzpicture}
  \end{center}
  When we glue together many chains to form such a diagram, we call it
  a ``fork''---the precise definition appears below. Observe that
  while these two blockchains agree through the vertex (block) labeled
  5, they contain (distinct) vertices labeled 9; this reflects two
  distinct blocks associated with slot 9 which, in light of the axiom
  above, 
  may be produced by either an adversarial participant assigned to slot 9 or 
  two honest participants, both assigned to slot 9.
  
  Finally, as we assume that messages from honest players are
  delivered before the next slot begins, we note a direct consequence of the longest
  chain rule:
\begin{enumerate}[label={\textbf{A\arabic*}}., ref={\textbf{A\arabic*}}, resume=axiom]
  \item\label{axiom:honest-depth} 
  If two honestly generated blocks $B_1$ and $B_2$ are labeled
  with slots $\slot_1$ and $\slot_2$ for which $\slot_1 < \slot_2$,
  then the length of the unique blockchain terminating at $B_1$ is
  strictly less than the length of the unique blockchain terminating at $B_2$.
\end{enumerate}
Recall that the honest participant(s) assigned to slot
$\slot_2$ will be aware of the blockchain terminating at $B_1$ that
was broadcast by an honest player in slot $\slot_1$ as a result of
synchronicity; according to the longest-chain rule, 
$B_2$ must have been placed on a chain that was at least this long. In contrast, not
all participants are necessarily aware of all blocks generated by
dishonest players, and indeed dishonest players may often want to
delay the delivery of an adversarial block to a participant or show
one block to some participants and show a completely different block
to others.

\paragraph{Characteristic strings, forks, and the formal axioms.}
Note that with the axioms we have discussed above, whether or not a
particular fork diagram (such as the one just above) corresponds to a valid
execution of the protocol depends on how the slots have been awarded to the parties by the
leader election mechanism. We introduce the notion of a ``characteristic'' string as a convenient
means of representing information about slot leaders in a given execution.
\begin{definition}[Characteristic string]\label{def:trivalent-char-string}
  Let $\slot_1, \ldots, \slot_{n}$ be a sequence of slots.  A
  \emph{characteristic string} $w$ is an element of
  $\{\h,\H,\A\}^n$. The string $w$ is consistent with a particular
  execution of a blockchain protocol on these slots if for each
  $t \in [n]$, (i) if $w_t = \A$, the slot $\slot_t$ is assigned to
  at least one adversarial participant, (ii) if $w_t = \h$, the slot
  $\slot_{t}$ is assigned to a unique, honest participant, and (iii)
  if $w_t = \H$, the slot $\slot_{t}$ is assigned to at least one
  honest participant and no adversarial participants.

  Observe that when an execution corresponds to a characteristic
  string $w$, it also corresponds to any string obtained from $w$ by
  replacing $\h$ symbols with $\H$ symbols.
%
\end{definition}
For two strings $x$ and $w$ on the same alphabet, 
we write $x \Prefix w$ if and only if $x$ is a strict prefix of $w$. 
Similarly, 
we write $x \PrefixEq w$ if and only if either $x = w$ or $x \Prefix w$. 
The empty string $\varepsilon$ is a prefix to any string. 
If $w_t \in \{\h, \H\}$, we say that ``$\Slot_t$ is honest'' and 
otherwise, we say that ``$\Slot_t$ is adversarial.'' 
With this discussion behind us, we set down the formal object we use
to reflect the various blockchains adopted by honest players during
the execution of a blockchain protocol. This definition formalizes the blockchains axioms discussed above.


\begin{definition}[Fork]\label{def:fork}
  Let $w\in \{\h, \H, \A\}^n$, $P = \{ i : w_i = \h\}$, and $Q = \{ j : w_j = \H\}$. 
  A \emph{fork} for the string $w$ consists of a directed and rooted
  tree $F=(V,E)$ with a labeling $\ell:V\to\{0,1,\dots,n\}$. We insist
  that each edge of $F$ is directed away from the root vertex and
  further require that
  \begin{enumerate}[label=(F{\arabic*})]
    \item\label{fork:root} the root vertex $r$ has label $\ell(r)=0$;

    \item\label{fork:monotone} the labels of vertices along any directed path are strictly increasing;

    \item\label{fork:unique-honest}\label{fork:multiply-honest}
    each index $i\in P$ 
    is the label of exactly one vertex of $F$ 
    and 
    each index $j\in Q$ 
    is the label of at least one vertex of $F$; and 

    \item\label{fork:honest-depth} 
    for any indices $i,j\in P \Union Q$, 
    if $i<j$ then 
    the depth of a vertex with label $i$ 
    is strictly less than 
    the depth of a vertex with label $j$.
  \end{enumerate}
\end{definition}

If $F$ is a fork for the characteristic string $w$, we write
$F\vdash w$.  The conditions~\ref{fork:root}--\ref{fork:honest-depth}
are analogues of the axioms~\ref{axiom:root}--\ref{axiom:honest-depth}
above. The formal reflection of axiom~\ref{axiom:honest} by
condition~\ref{fork:multiply-honest} deserves further comment: We have
chosen a definition of characteristic string that does not indicate
the number of honest victories in cases where there may be many; in
particular, the symbol $\H$ may be associated with any positive number
of (honest) vertices in the fork. Indeed, we even permit a fork to
have a \emph{single} honest vertex associated with such a symbol,
which enlarges the class of forks under consideration for a particular
characteristic string. This strengthens our results by effectively
giving the adversary the option to treat $\H$ symbols as $\h$
symbols. See Fig.~\ref{fig:fork}
for an example fork. 

\input{figures}

A final notational convention: If $F \vdash x$ and
$\hat{F} \vdash w$, we say that $F$ is a \emph{prefix} of $\hat{F}$,
written $F \fprefix \hat{F}$, if 
$x \PrefixEq w$
and $F$ appears as a
consistently-labeled subgraph of $\hat{F}$. 
(Specifically, each path of $F$ appears, with identical labels, in $\hat{F}$.)


  Let $w$ be a characteristic string.  The directed paths in the fork
  $F \Fork w$ originating from the root are called \emph{tines}; these
  are abstract representations of blockchains. (Note that a tine may 
  not terminate at a leaf of the fork.)  We naturally extend the label
  function $\ell$ for tines: i.e., $\ell(t) \triangleq \ell(v)$ where
  the tine $t$ terminates at vertex $v$. The length of a tine $t$ is
  denoted by $\length(t)$.

 \paragraph{Viable tines.}
 The longest-chain rule dictates that honest players build on chains
 that are at least as long as all previously broadcast honest
 chains. It is convenient to distinguish such tines in the analysis:
 specifically, a tine $t$ of $F$ is called \emph{viable} if its length
 is no smaller than the depth of any honest vertex $v$ for which
 $\ell(v) \leq \ell(t)$. A tine $t$ is \emph{viable at slot $s$} if
 the length of the portion of $t$ appearing over slots $0,\ldots, s$ 
 is no smaller than the depths of any honest vertices labeled from these slots. (As noted,
 the properties~\ref{fork:multiply-honest} and~\ref{fork:honest-depth}
 together imply that an honest observer at slot $s$ will only adopt a
 viable tine.)  
 The \emph{honest depth} function
 $\hdepth : P \Union Q \rightarrow [n]$, 
 defined as $\hdepth(i) = \max_{t \in F} \left\{ \length(t) \SuchThat \ell(t) = i \right\}$, 
 gives the largest depth of the (honest) vertices 
 associated with an honest slot; by~\ref{fork:honest-depth},
 $\hdepth(\cdot)$ is strictly increasing.




\subsection{Slot settlement and the Unique Vertex Property}\label{sec:model-settlement}
  
  We are now ready to explore the power of an adversary in this
  setting who has corrupted a (perhaps evolving) coalition of the
  players. We focus on the possibility that such an adversary can
  violate the consistency of the honest players'
  blockchains. In particular, we consider the possibility that, at
  some time $t$, the adversary conspires to produce two maximum-length blockchains 
  that diverge prior to a previous slot $s \leq t$; in
  this case honest players adopting the longest-chain rule may clearly
  disagree about the history of the blockchain after slot $s$. We call
  such a circumstance a \emph{settlement violation}.

  To express this in our abstract language, let $F \Fork w$ be a fork
  corresponding to an execution with characteristic string $w$. Such a
  settlement violation induces two viable tines $t_1, t_2$ with the
  same length that diverge prior to a particular slot of interest. We
  record this below.
    
  \begin{definition}[Settlement with parameters $s,k \in \NN$]\label{def:settlement}
    Let $n \in \NN$ and let $w$ be a characteristic string of length $n$. 
    Let $t \in [s + k, n]$ be an integer, $\hat{w} \PrefixEq w, |\hat{w}| = t$, and 
    let $F$ be any fork for $\hat{w}$. 
    We say that a slot $s$ is \emph{not $k$-settled} in $F$ if 
    $F$ contains two maximum-length tines $\Chain_1, \Chain_2$ 
    that ``diverge prior to $s$,'' i.e., they either
    contain different vertices labeled with $s$, or one contains a
    vertex labeled with $s$ while the other does not. 
    Otherwise, we say that \emph{slot $s$ is $k$-settled in $F$}. 
    We say that \emph{slot $s$ is $k$-settled in $w$} if, 
    for each $t \geq s+k$, 
    it is $k$-settled in every fork $F \Fork \hat{w}$ where $\hat{w} \PrefixEq w, |\hat{w}| = t$.
  \end{definition}

  \begin{figure*}[t]
    \begin{center}
      \fbox{
        \begin{minipage}{\textwidth}
          \begin{center}
            \textbf{The $(\Distribution,T;s,k)$-settlement game}
          \end{center}
          \begin{enumerate}

          \item A characteristic string $w \in \{\h,\H,\A\}^T$ is drawn from
            $\mathcal{D}$. (This reflects the results of the leader
            election mechanism.)

          \item Let $A_0 \vdash \varepsilon$ denote the initial fork for
            the empty string $\varepsilon$ consisting of a single node
            corresponding to the genesis block.

          \item For each slot $\Slot_t, t = 1, \ldots, T$ in increasing order:
            \begin{enumerate}
            \item\label{game:honest} (Honest slot.)  This case
              pertains to $w_t \in \{\h, \H\}$.  If $w_t = \h$ then
              $\Adversary$ sets $k = 1$.  If $w_t = \H$ then
              $\Adversary$ chooses an arbitrary integer $k \geq 1$.
              The challenger is then given $k$ and the fork
              $A_{t-1} \vdash w_1 \ldots w_{t-1}$.  He must determine
              a new fork $F_{t} \vdash w_1 \ldots w_t$ by adding $k$
              new vertices (all labeled with $t$) to $A_{t-1}$.  Each
              new vertex is added at the end of a maximum-length path
              in $A_{t-1}$.  If there are multiple
              candidates\footnote{ It is possible that all maximum-length 
                tines are honest. In the settlement game
                considered in~\cite{LinearConsistencySODA}, at least
                one of these tines was adversarial.} 
              for this path,
              $\Adversary$ may break the tie.  If $k \geq 2$, multiple
              vertices (all with label $t$) may be added at the end of
              the same path.

            \item 
            (Adversarial slot.)
            If $w_t = \A$, this is an adversarial slot. $\Adversary$
              may set $F_t \vdash w_1\ldots w_t$ to be an arbitrary fork
              for which $A_{t-1} \fprefix F_t$.
              
            \item (Adversarial augmentation.) $\Adversary$ determines an
              arbitrary fork $A_t \vdash w_1 \ldots, w_{t}$ for
              which $F_{t} \fprefix A_{t}$.
            \end{enumerate}
             Recall that $F \fprefix F'$ indicates that $F'$
              contains, as a consistently-labeled subgraph, the fork $F$.
          \end{enumerate}
          $\Adversary$ \emph{wins the settlement game} if slot $s$ is not
          $k$-settled in some fork $A_t, t \geq s+k$.
        \end{minipage}
      }
    \end{center}
  \end{figure*}

  \begin{definition}[Bottleneck Property (BP) and Unique Vertex Property (UVP)]\label{def:bottleneck-property}\label{def:unique-vertex-property}
    Let $w \in \{\h, \H, \A\}^T$ be a characteristic string.  
    A slot $s \in [T]$ is said to have the 
    \emph{bottleneck property in $w$} 
    if, 
    for any fork $F \Fork w$ 
    and any $k \geq s + 1$, 
    every tine viable at the onset of slot $k$ 
    contains, as its prefix, some vertex with label $s$.       
    Slot $s$ is said to have the \emph{Unique Vertex Property} 
    if, 
    for any fork $F \Fork w$, 
    there is a unique vertex $u \in F$ with label $s$ 
    so that 
    for any $k \geq s + 1$, 
    all tines viable at the onset of slot $k$ 
    contain, as their common prefix, 
    the vertex $u$.
  \end{definition}
  Thus 
  if a uniquely honest slot in $w$ has the bottleneck property, 
  it has the UVP as well.
  As a consistency property, UVP has several advantages over slot settlement. 
  First, it easily implies the slot settlement property:
  let $w \in \{\h, \H, \A\}^T, s \in [T]$, and $k \in [T - s]$. 
  \begin{equation}\label{eq:settlement-uvp}
    \text{If a slot $t \in [s, s + k]$ has UVP in $w$ 
    then $s$ is $k$-settled in $w$.}     
  \end{equation}  
  In addition, UVP has a straightforward characterization 
  using ``Catalan slots'' 
  (see Theorem~\ref{thm:unique-honest}) 
  and ``relative margin'' 
  (see Lemma~\ref{lemma:uvp-margin}); 
  these characterizations are amenable to stochastic analysis. 
  Finally, since UVP is structurally reminiscent of the traditional common prefix (CP) violations, 
  UVP easily implies CP. 
  The analogous statement ``settlement implies CP,'' however, 
  requires a lengthy proof both in~\cite{LinearConsistency} and in 
  our framework. See 
  Appendix~\ref{sec:cp-forks}
  for details.

\subsection{Adversarial attacks on settlement time; the settlement game}\label{sec:game}

  To clarify the relationship between forks and the chains at play in a
  canonical blockchain protocol, we define a game-based model below that
  explicitly describes the relationship between forks and executions.
  By design, the probability that the adversary wins this game is at
  most the probability that a slot $s$ is not $k$-settled. 

  Consider the \emph{$(\Distribution,T;s,k)$-settlement game} 
  (presented in the box), played
  between an adversary $\Adversary$ and a challenger $\Challenger$ with
  a leader election mechanism modeled by an ideal distribution
  $\Distribution$. Intuitively, the game should reflect the ability of
  the adversary to achieve a settlement violation; that is, to present
  two maximum-length viable blockchains to a future honest observer,
  thus forcing them to choose between two alternate histories which
  disagree on slot $s$.
  The challenger plays the role(s) of the honest players during the
  protocol. 

  It is important to note that the game bestows the player $\Adversary$ 
  with the power to choose the number of honest vertices in 
  a multiply honest slot. 
  Note that this setting makes the player strictly more powerful and, 
  importantly, implies that 
  the game is completely determined 
  by the choices made by $\Adversary$ 
  (i.e., the actions of the challenger are deterministic). 
  Consequently, in Definition~\ref{def:settlement-insecurity}, 
  we can use a single, implicit universal quantifier 
  over all strategies $\Adversary$; no choices of the challenger are actually necessary to fully describe the game.

  \begin{definition}[Settlement insecurity]\label{def:settlement-insecurity}
    Let $\Distribution$ be a distribution on $\{\h, \H, \A\}^T$. 
    Let $w \sim \Distribution$ be the string used in the 
    first step of a $(\Distribution, T; s, k)$-settlement game $G$. 
    The \emph{$(s,k)$-settlement insecurity} of $\Distribution$ 
    is defined as 
    \[
      \mathbf{S}^{s,k}[\Distribution] \triangleq 
        \max_{\substack{\hat{w} \PrefixEq w \\ |\hat{w}| \geq s + k}}\,
        \max_{F \Fork \hat{w}}\, 
        \Pr\left[\parbox{50mm}{\centering $F$ has two maximum-length tines that diverge prior to slot $s$}\right]
      \,.
    \]
    Note that the probability in the right-hand side is the same as 
    the probability that $\Adversary$ wins $G$.
  \end{definition}

  Note that in typical PoS settings the distribution $\Distribution$
  is determined by the combined stake held by the adversarial players,
  the leader election mechanism, and the dynamics of the protocol. The
  most common case (as seen in Snow White~\cite{SnowWhite},
  Ouroboros~\cite{Ouroboros}, and Ouroboros Praos~\cite{Praos})
  guarantees that the characteristic string $w = w_1 \ldots w_T$ is
  drawn from an i.i.d.\ distribution for which
  $\Pr[w_i = \A] \leq (1 - \epsilon)/2$ for some $\epsilon \in (0, 1)$;
  here the constant $(1-\epsilon)/2$ is directly related to the stake
  held by the adversary. Some settings involving adaptive adversaries
  (e.g., Ouroboros Praos~\cite{Praos}) yield a weaker martingale-type
  guarantee that
  $\Pr[w_i = \A \mid w_1, \ldots, w_{i-1}] \leq (1 - \epsilon)/2$.  We
  can easily handle both types of distributions in our analysis since
  the former distribution ``stochastically dominates'' the latter.
  As a rule, we denote the
  probability distribution associated with a random variable using
  uppercase script letters. 
  \begin{definition}[Stochastic dominance]\label{def:dominance} 
    Let $X$ and $Y$ be random variables taking values in some set $\Omega$ 
    endowed with a partial order $\leq$. 
    We say that $X$ \emph{stochastically dominates} $Y$, 
    written $Y \dominatedby X$, if 
    $
      \mathcal{X}(A) \geq \mathcal{Y}(A)
    $ 
    for all \emph{monotone sets} $A \subseteq \Omega$, 
    where a set $A \subseteq \Omega$ is called 
    monotone if $a \in A$ implies $a' \in A$ for all $a \leq a'$.
    As a special case, when $\Omega = \R$,  $Y \dominatedby X$ if
    $\Pr[X \geq \Lambda] \geq \Pr[Y \geq \Lambda]$
    for every $\Lambda \in \R$.  
    We extend this notion to probability
    distributions in the natural way.
  \end{definition}

  Throughout the paper, we adopt the following partial order on
  $\{\h, \H, \A\}^T$: If $T = 1$, define $\h < \H < \A$.  Otherwise,
  for two strings $xa, yb \in \{\h, \H, \A\}^T, |a| = |b| = 1$,
  $xa \leq yb$ if and only if $x \leq y$ and $a \leq b$. When
  $x \leq y$, one might say that $y$ is ``more adversarial'' than $x$:
  indeed, if $F \vdash x$ and $x \leq y$ then $F \vdash y$ so that any
  settlement violation for $x$ induces a settlement violation for $y$.

  \begin{definition}[$(\epsilon, p_\h)$-Bernoulli
    condition]\label{def:bernoulli-cond}
    Let $T \in \NN, \epsilon \in (0,1)$, and
    $p_h \in [0,(1+\epsilon)/2]$. Define $p_\A = (1-\epsilon)/2$ and
    $p_\H = 1- p_\A - p_\h$.  A random variable $w = w_1 \ldots w_T$
    taking values in $\{\h, \H, \A\}^T$ is said to satisfy the
    \emph{$(\epsilon, p_\h)$-Bernoulli condition} if each
    $w_i, i \in [T]$, is independent and identically distributed as
    follows: $\Pr[w_i = \sigma] = p_\sigma$ for
    $\sigma \in \{\h, \H, \A\}$.  The distribution of $w$ is also said
    to satisfy the $(\epsilon, p_\h)$-Bernoulli condition.

    We frequently use the notation $p_\H$ and $p_\A$ in the context of
    such a random variable when $\epsilon$ and $p_\h$ can be inferred
    from context.
  \end{definition}
  
  \begin{theorem}[Main theorem]\label{thm:main}
    Let $\epsilon, p_\h \in (0, 1)$ and $s, k, T \in \NN$.  
    Let $\mathcal{B}$ be a distribution 
    on length-$T$ characteristic strings satisfying 
    the $(\epsilon, p_\h)$-Bernoulli condition.
    Then 
    $
      \mathbf{S}^{s,k}[\mathcal{B}] 
        \leq 
        \exp\left(-k\cdot \Omega( 
          \min(\epsilon^3, \epsilon^2 p_\h) 
        \right)
    $.
    Furthermore, 
    let 
    $\mathcal{W}$ be a distribution on
    $\{\h, \H, \A\}^T$ so that 
    $\mathcal{W} \DominatedBy \mathcal{B}$. 
    Then $\mathbf{S}^{s,k}[\mathcal{W}] 
        \leq \mathbf{S}^{s,k}[\mathcal{B}]$.   
    (Here, the asymptotic notation hides constants that do not depend on $\epsilon$ or $k$.)
  \end{theorem}
  Note that the quantity $p_\h$ above cannot be zero.
  We present the proof in Section~\ref{sec:bounds-main-proofs}. 
  In Section~\ref{sec:recursion}, 
  we give a characterization of the UVP which 
  allows us to explicitly compute $\mathbf{S}^{s,k}[\mathcal{B}]$; 
  see 
  Theorem~\ref{thm:relative-margin} and 
  Section~\ref{sec:exact-prob}.

  \paragraph{Analysis in the $\Delta$-synchronous setting.} The security
   game above most naturally models a blockchain protocol over a
   synchronous network with immediate delivery (because each ``honest''
   play of the challenger always builds on a fork that contains the fork
   generated by previous honest plays). However, the model can be easily
   adapted to protocols in the $\Delta$-synchronous model by applying
   the $\Delta$-reduction mapping of~\cite{Praos} (which is specifically
   designed to lift the synchronous analysis to the $\Delta$-synchronous
   setting).
   These details appear in Section~\ref{sec:async}.

  \paragraph{Public leader schedules.} One attractive feature of this
  model is that it gives the adversary full information about the future
  schedule of leaders. The analysis of some protocols indeed demand this
  (e.g., Ouroboros, Snow White). Other protocols---especially those
  designed to offer security against adaptive adversaries (Praos,
  Genesis)---in fact contrive to keep the leader schedule private. Of
  course, as our analysis is in the more difficult ``full information''
  model, it applies to all of these systems.

  \paragraph{Bootstrapping multi-phase algorithms; stake shift.} We remark that
  several existing proof-of-stake blockchain protocols proceed in
  phases, each of which is obligated to generate the randomness (for
  leader election, say) for the next phase based on the current stake
  distribution. The blockchain security properties of each phase are
  then individually analyzed---assuming clean randomness---which yields
  a recursive security argument; in this context the game outlined above
  precisely reflects the single phase analysis.

\subsection{A consistent longest-chain selection rule}\label{sec:lcr-model} 
  %

  Let us 
  modify axiom~\ref{axiom:message-delivery} as follows:

  \begin{enumerate}[label={\textbf{A\arabic*}\ensuremath{^\prime}}., ref={\textbf{A\arabic*}\ensuremath{^\prime}}, start=0]
  \item\label{axiom:tie-breaking} 
    In addition to axiom~\ref{axiom:message-delivery},  
    an arbitrary but consistent 
    longest-chain tie-breaking rule 
    is used by all honest participants.
  \end{enumerate}
  As a consequence, 
  if two honest participants observe 
  the same set of blockchains of maximum length, 
  they will extend the same blockchain.

  \begin{definition}[Bivalent characteristic
    string]\label{def:bivalent-char-string}
    Let $\slot_1, \ldots, \slot_{n}$ be a sequence of slots. 
    A \emph{bivalent characteristic string} $w$ 
    is an element of $\{\H,\A\}^n$ 
    defined for a particular execution of a blockchain protocol on these slots so that 
    for $t \in [n]$, 
    $w_t = \A$ if $\slot_{t}$ is assigned to an adversarial participant, 
    and $w_t = \H$ otherwise.
  \end{definition}
  The definition of a fork for a bivalent characteristic string is
  identical to Definition~\ref{def:fork} (somewhat simplified as a
  bivalent string does not contain any $\h$ symbol).
  Also note that the $(\epsilon, 0)$-condition 
  from Definition~\ref{def:bernoulli-cond} 
  is well-defined for bivalent characteristic strings.

  Let $w$ be a bivalent characteristic string, $F$ a fork for $w$, and
  $F'$ a fork for $w\H$ so that $F \ForkPrefix F'$ and any honest
  vertex in $F' \setminus F$ has label $|w| + 1$.  If $F$ contains a
  maximum-length adversarial tine, there is no guarantee that two
  honest observers at slot $|w| + 1$ will agree on the longest chain:
  the adversary may chose to expose the adversarial chain to one and
  not the other.
  In this case, we say that \emph{$F$ has a tie for the longest-chain
  rule}---or, in short, that \emph{$F$ has an LCR tie}.  
  When there
  is no LCR tie (that is, no maximum-length adversarial tine), all
  honest slot leaders at slot $|w| + 1$ necessarily extend the same
  honest tine determined by the consistent longest-chain tie-breaking
  rule.

  \begin{theorem}[Main theorem; consistent tie-breaking]\label{thm:main-bivalent}
    Let $\epsilon \in (0, 1)$ and $s, k, T \in \NN$.  
    Let $\mathcal{B}$ be a distribution 
    on length-$T$ bivalent characteristic strings 
    satisfying the $(\epsilon, 0)$-Bernoulli condition. 
    Let 
    $\mathcal{W}$ be a distribution on
    $\{\H, \A\}^T$ so that 
    $\mathcal{W} \DominatedBy \mathcal{B}$. 
    Then 
    $
      \mathbf{S}^{s,k}[\mathcal{W}] 
        \leq \mathbf{S}^{s,k}[\mathcal{B}] 
        \leq 
        \exp\bigl(-k \cdot \Omega(\epsilon^3 (1 + O(\epsilon) ) )\bigr)
    $. 
    (Here, the asymptotic notation hides constants that do not depend on $\epsilon$ or $k$.)
  \end{theorem}

  The proof is deferred to Section~\ref{sec:bounds-main-proofs}.
  Note that the theorem above states that a 
  PoS protocol can achieve optimal consistency error 
  even with a leader election scheme 
  that produces no uniquely honest slots. 
  In contrast, Theorem~\ref{thm:main} requires 
  a non-zero probability for uniquely honest slots.


%% file: figures.tex
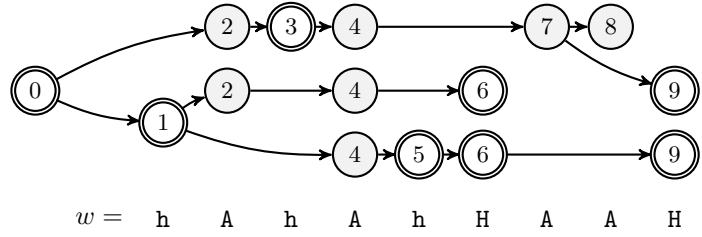
\begin{figure}[t]
  \centering
  \begin{tikzpicture}[scale=0.85,>=stealth', auto, semithick,
    honest/.style={circle,draw=black,thick,text=black,double,font=\small},
    malicious/.style={fill=gray!10,circle,draw=black,thick,text=black,font=\small}]
    \node at (0,-2) {$w =$};
    \node at (1,-2) {$\h$};
    \node[honest]    at (1,-.5)  (ab1) {$1$};
    \node at (2,-2) {$\A$};
    \node[malicious] at (2,0)  (b2) {$2$}; \node[malicious] at (2,1) (c1) {$2$};
    \node at (3,-2) {$\h$};    \node[honest]    at (3,1)  (c2) {$3$};
    \node at (4,-2) {$\A$};    \node[malicious] at (4,0)  (b3) {$4$}; \node[malicious] at (4,-1) (a2) {$4$};
    \node[malicious] at (4,1) (c3) {$4$};
    \node at (5,-2) {$\h$};    \node[honest]    at (5,-1) (a3) {$5$};
    \node at (6,-2) {$\H$};    \node[honest]    at (6,0)  (b4) {$6$}; \node[honest]    at (6,-1)  (a4) {$6$};
    \node at (7,-2) {$\A$};    \node[malicious] at (7,1)  (c4) {$7$};
    \node at (8,-2) {$\A$};    \node[malicious] at (8,1)  (c5) {$8$};
    \node at (9,-2) {$\H$};    \node[honest]    at (9,0)  (b5) {$9$}; \node[honest]    at (9,-1)  (a5) {$9$};
    \node[honest] at (-1,0) (base) {$0$};
    \draw[thick,->] (base) to[bend left=10] (c1);
    \draw[thick,->] (base) to[bend right=10] (ab1);
    \draw[thick,->] (ab1) to[bend right=10] (a2);
    \draw[thick,->] (a2) -- (a3);
    \draw[thick,->] (a3) -- (a4);
    \draw[thick,->] (ab1) to[bend left=10] (b2);
    \draw[thick,->] (b2) -- (b3);
    \draw[thick,->] (b3) -- (b4);
    \draw[thick,->] (c4) to[bend right=10] (b5);
    \draw[thick,->] (a4) -- (a5);
    \draw[thick,->] (c1) -- (c2);
    \draw[thick,->] (c2) -- (c3);
    \draw[thick,->] (c3) -- (c4);
    \draw[thick,->] (c4) -- (c5);
  \end{tikzpicture}
  \caption{A fork $F$ for the characteristic string $w = \h\A\h\A\h\H\A\A\H$;
    vertices appear with their labels and honest vertices are
    highlighted with double borders. Note that the depths of the
    (honest) vertices associated with the honest indices of $w$ are
    strictly increasing. Note, also, that this fork has three disjoint
    paths of maximum depth. 
    In addition, two honest vertices have label 6 and two more have label 9, 
    indicating the fact that two honest leaders are associated with each of the (honest) slots 6 and 9. 
    Honest vertices with the same label are concurrent and, therefore, cannot extend each other.
    Note that the two honest vertices with label 6 extend different vertices with the same depth. 
    This is allowed since any tie in the longest-chain rule is broken by the adversary. 
    }
  \label{fig:fork}
\end{figure}

%% file: definitions.tex
As we have outlined before, 
if slot $t$ in a characteristic string $w$ 
has the Unique Vertex Property (UVP) 
then the slots $s = 1, \ldots, t$ 
are settled in every fork for $w$. 
The goal of this section is to 
characterize when a slot has the UVP. 
(In Section~\ref{sec:recursion}, we show an alternative way 
to characterize the UVP; see Lemma~\ref{lemma:uvp-margin}.) 

We start with laying down some structural properties of forks. 
Next, we define the so-called Catalan slots 
and show that if a slot is Catalan then \emph{in every fork}, all sufficiently long blockchains must contain a block from that slot. 
Next, we show that this implication is actually an equivalence. 
Finally, 
we revisit the above implication assuming 
that the honest players use a consistent longest-chain tie-breaking rule.


%% file: fork.tex
  \subsection{Viable blockchains}
  A vertex of a fork is said to be \emph{honest} 
  if it is labeled with an index $i$ such that $w_i \in \{\h, \H\}$; 
  otherwise, it is said to be \emph{adversarial}.

  \begin{definition}[Tines, length, and height]
    Let $F \vdash w$ be a fork for a characteristic string.  A
    \emph{tine} of $F$ is a directed path starting from the root. For
    any tine $t$ we define its \emph{length} to be the number of edges
    in the path, and for any vertex $v$ we define its \emph{depth} to be
    the length of the unique tine that ends at $v$. 
    If a tine $t_1$ is a strict prefix of another tine $t_2$, we write $t_1 \Prefix t_2$. 
    Similarly, if $t_1$ is a non-strict prefix of $t_2$, we write $t_1 \PrefixEq t_2$.
    The longest common prefix of two tines $t_1, t_2$ is denoted by $t_1 \Intersect t_2$. 
    That is, $\ell(t_1 \Intersect t_2) = \max\{\ell(u) \SuchThat \text{$u \PrefixEq t_1$ and $u \PrefixEq t_2$} \}$. 
    The \emph{height} of
    a fork (as is usual for a tree) 
    is the length of the longest tine,
    denoted by $\height(F)$. 
  \end{definition}
  Let $F \Fork xy$ and 
  two tines $t_1, t_2 \in F$ are disjoint over $y$. 
  We say that these tines are \emph{$y$-disjoint}; 
  equivalently, we also say that \emph{$t_1$ is $y$-disjoint with $t_2$}.

  When an adversary builds a fork, it is natural to imagine that 
  he ``grows'' an existing fork by adding new vertices and edges. 
  \begin{definition}[Fork prefixes]
    Let $w, x \in \{\h, \H, \A\}^*$ so that $x\PrefixEq w$. 
    Let $F, F'$ be two forks for $x$ and $w$, respectively. 
    We say that $F$ is a \emph{prefix} of $F'$ if 
    $F$ is a consistently labeled subgraph of $F'$. 
    That is, all vertices and edges of $F$ also appear in $F'$ and 
    the label of any vertex appearing in both $F$ and $F'$ is identical. 
    We denote this relationship by $F\fprefix F'$.
  \end{definition}

  \noindent
  When speaking about a tine that appears in both $F$ and $F'$, 
  we place the fork in the subscript of relevant properties.


    For any string $x$ (on any alphabet) and a symbol $\sigma$ in that alphabet, 
    define $\#_\sigma(x)$ 
    as the number of appearances of $\sigma$ in $x$. 
    When a characteristic string $w \in \{\h, \H, \A\}^T$ is fixed from the context, 
    we extend this notation to sub-intervals of $[T]$ in a natural way: 
    For integers $i, j \in [T], i \leq j$, 
    let $I = [i, j] \subset [T]$ be a closed interval 
    and define $\#_\sigma(I) = \#_\sigma(w_i \ldots w_j)$ for $\sigma \in \{\h, \H, \A\}$. 
    A characteristic string $w$ is called $\Hheavy$\ if $\#_\h(w) + \#_\H(w) > \#_1(x)$; 
    otherwise, it is called $\Aheavy$. 
    For a given characteristic string $w$ of length $T$, 
    an interval $I = [i,j] \subseteq [T]$ is called $\Aheavy$\ 
    if the substring $w_i \ldots w_j$ is $\Aheavy$.

    \paragraph{Adversarial extensions.} 
    Let $x,y$ be two characteristic strings and $|y| \geq 0$.
    Let $F$ be a fork for $x$ and let $B$ be an honest tine in $F$. 
    We say that \emph{$B$ has an adversarial extension} 
    if there is a fork $F' \Fork xy, F \ForkPrefix F'$ and 
    an adversarial tine $t \in F'$ 
    so that $B \Prefix t$ and  
    $B$ is the last honest vertex on $t$. 
    Note that $t$ can be made disjoint with any $F$-tine 
    over the interval $[\ell(B) + 1, \ell(t)]$.

    \paragraph{Viable adversarial extensions and $\A$-heaviness.} 
    Let $w \in \{\h, \H, \A\}^T$,   
    $s \in [T + 1]$, and 
    $F \Fork w_1 \ldots w_{s - 1}$ an arbitrary fork. 
    Let $B \in F$ be an honest vertex 
    and $t$ a maximum-length \emph{honest tine} in $F$.
    Consider the following statements: 
    \begin{enumerate}[label=(\alph*)]
      \item \label{fact-part:viable-adv-ext} $B$ has an adversarial extension viable at the onset of slot $s$.
      \item \label{fact-part:Aheavy} The interval $I = [\ell(B) + 1, s - 1]$ is $\Aheavy$. 
      \item \label{fact-part:conservative} $\length(t) = \#_\h(I) + \#_\H(I) + \length(B)$.     
    \end{enumerate}

    \begin{fact}[]\label{fact:fork-structure}
    ~\ref{fact-part:viable-adv-ext} $\Longrightarrow$
    \ref{fact-part:Aheavy}.
    In addition, if we assume~\ref{fact-part:conservative}, then 
    ~\ref{fact-part:Aheavy} $\Longrightarrow$ 
    ~\ref{fact-part:viable-adv-ext}.
    \end{fact}

    \begin{proof}~
      \begin{description}[font=\normalfont\itshape\space]
        \item[\ref{fact-part:viable-adv-ext} implies~\ref{fact-part:Aheavy}.]
        Let $F' \Fork w_1 \ldots w_{s-1}$ be 
        a fork so that $F \ForkPrefix F'$ 
        and $B$ has an adversarial extension $t' \in F'$ 
        viable at the onset of slot $s$. 
        Considering the interval $I$, 
        the longest honest tine in $F'$ 
        grows by at least $\#_\h(I) + \#_\H(I)$ vertices. 
        Since the viable tine $t'$ 
        contains only adversarial vertices from the interval $I$, 
        it follows that $\#_\A(I)$ must be at least $\#_\h(I) + \#_\H(I)$. 
        Hence, $I$ is $\Aheavy$.

        \item[\ref{fact-part:conservative} and \ref{fact-part:Aheavy} implies~\ref{fact-part:viable-adv-ext}.]
        Since $I$ is $\Aheavy$,  
        $I$ contains at least $\#_\h(I) + \#_\H(I) = \length(t) - \length(B)$ 
        adversarial slots. 
        Thus, we can augment $B$ by adding 
        $\length(t) - \length(B)$ adversarial vertices 
        from these slots 
        so that 
        the resulting adversarial extension is viable 
        at the onset of slot $s$.
      \end{description}
    \end{proof}

    \begin{corollary}\label{coro:interval-honest-vertices}
      Let $w$ be a characteristic string, 
      $F$ be any fork for $w$, 
      and let $t$ be any tine in $F$.
      Let $B_1$ and $B_2$ be two honest vertices on $t$ such that 
      (i) $\ell(B_1) < \ell(B_2)$, 
      (ii) $t$ contains only adversarial vertices from $I = [\ell(B_1) + 1, \ell(B_2) - 1]$, and 
      (iii) $t$ contains at least one vertex from $I$.
      Then $I$ is $\Aheavy$. 
    \end{corollary}
    \begin{proof}
      By assumption, 
      the honest vertex $B_2$ builds on some adversarial tine $t^\prime$ 
      that is viable at the onset of slot $\ell(B_2)$ and, importantly, 
      contains $B_1$ as its last honest vertex. 
      By Fact~\ref{fact:fork-structure}, 
      the interval $I$ is $\Aheavy$.
    \end{proof}


%% file: catalan-cp.tex


Below, we define the so-called Catalan slots and show, 
in Theorems~\ref{thm:unique-honest} and~\ref{thm:multiple-honest}, 
that certain Catalan slots have the UVP.

\begin{definition}[Catalan slot]
  Let $w \in \{\h, \H, \A\}^T$ be a characteristic string and 
  let $s \in [T]$ be an integer. 
  $s$ is called a \emph{left-Catalan slot in $w$} 
  if, for any integer $\ell \in [s]$, the interval $[\ell, s]$ is $\Hheavy$ in $w$.
  $s$ is called a \emph{right-Catalan slot in $w$} 
  if, for any integer $r \in [s, T]$, the interval $[s, r]$ is $\Hheavy$ in $w$.
  Finally, $s$ is called a \emph{Catalan slot in $w$} if 
  it is both left- and right-Catalan in $w$. 
\end{definition}
Observe that a left- or right-Catalan slot must be honest. 
In addition, the slot before a left-Catalan 
(resp., after a right-Catalan) slot must be honest as well.
Thus the slots adjacent to a Catalan slot must be honest. 
A Catalan slot $c$ acts as a barrier for adversarial tine extensions 
in that in any fork, every tine viable at the onset of slot $c+1$ must be honest.

\begin{fact}\label{fact:catalan-unique-longest}
  Let $w \in \{\h, \H, \A\}^T$ be a characteristic string 
  and $s$ a left-Catalan slot in $w$. 
  In any fork for $w$, 
  every viable tine at the onset of slot $s + 1$ is 
  an honest tine from slot $s$. 
\end{fact}
\begin{proof}
  Let $\tau$ be the longest tine with label $s$. 
  ($\tau$ is an honest tine. 
  If $s$ is a uniquely honest slot, 
  $\tau$ is unique. Otherwise, 
  $\tau$ is unique up to tie-breaking among equally-long tines.)
  We claim that 
  all adversarial tines $t \in F, \ell(t) \leq s - 1$ 
  are strictly shorter than $\tau$.
  Suppose, towards a contradiction, that 
  $t$ is a viable adversarial tine at the onset of slot $s + 1$, i.e., 
  $\ell(t) \leq s - 1$ and $\length(t) \geq \length(\tau)$. 
  Let $B$ be the last honest vertex on $t$; necessarily, $\ell(B) < s$. 
  According to 
  Fact~\ref{fact:fork-structure},
  the interval $[\ell(B) + 1, s]$ is $\Aheavy$. 
  But this contradicts the assumption that $s$ is a left-Catalan slot. 
  Hence the adversarial tine $t$ cannot be viable.
\end{proof}

\begin{observation}\label{obs:multi-honest}
  If $s$ is a Catalan slot for $w$, Fact~\ref{fact:catalan-unique-longest} implies that 
  in every fork for $w$, 
  an honest slot leader at slot $s + 1$ always 
  builds on top of an honest tine with label $s$; 
  this tine, in fact, will have the maximum length among all tines with label $s$.
\end{observation}

\begin{fact}\label{fact:almost-cp-implies-catalan}
  Let $w \in \{\h, \H, \A\}^T$ be a characteristic string. 
  If an honest slot in $w$ has the bottleneck property then 
  it is a Catalan slot.
\end{fact}
\begin{proof}
  Let $s \in [T]$ be an honest slot in $w$.  
  We will prove the contrapositive: namely, 
  that if $s$ is not Catalan then $s$ does not have the bottleneck property. 
  
  Suppose $s$ is not a Catalan slot.
  Then there must be some $a, b \in [T]$ so that 
  $I = [a, b]$ is the largest $\Aheavy$ interval 
  which includes $s$. 
  Necessarily, either $b = T$, or $b + 1$ must be an honest slot. 
  Likewise, either $a = 1$, or $a - 1$ must be an honest slot. 
  
  Let $F$ be a fork for $w_1 \ldots w_b$ 
  and let $u \in F, \ell(u) = a - 1$ be an honest tine. 
  (If $a = 1$, we can take $u$ as the root vertex.)  
  Let $t$ be a maximum-length honest tine in $F$ 
  and assume that $\length(t) = \length(u) + \#_\h(I) + \#_\H(I)$.
  Since $I$ is $\Aheavy$, 
  Fact~\ref{fact:fork-structure} 
  states that 
  it is possible to augment $u$ into 
  an adversarial extension $t'$ 
  viable at the onset of slot $b + 1$. 
  As $t'$ will not contain any vertex from the honest slot $s$,  
  $s$ does not have the bottleneck property in $w$.
\end{proof}

The following theorem shows that a uniquely honest Catalan slot has the UVP.


  \begin{theorem}\label{thm:unique-honest}
    Let $w \in \{h,H,A\}^T$ be a characteristic string. 
    Let $s \in [T]$ be a uniquely honest slot in $w$. 
    Slot $s$ is Catalan in $w$ 
    if and only if 
    it has the UVP in $w$. 
  \end{theorem}

  \begin{proof}
    (The reverse implication.) 
    Since $s$ has the UVP  
    it satisfies the (weaker) bottleneck property. 
    By Fact~\ref{fact:almost-cp-implies-catalan}, 
    the honest slot $s$ must be Catalan. 

    (The forward implication.) 
    By assumption, slot $s$ has a unique honest leader. 
    Let $\tau$ be the unique honest tine at slot $s$.
    By Fact~\ref{fact:catalan-unique-longest}, 
    the honest tine $\tau$ is the only viable tine 
    at the onset of slot $s + 1$.    
    If $s = T$ then $\tau$ is the only viable tine 
    at the onset of slot $T + 1$.
    Now suppose $s \leq T - 1$.
    As $s$ is a Catalan slot, slots $s$ and $s + 1$ must be honest. 
    Let $t$ be a viable tine at the onset of some slot $k, k \geq s + 2$. 
    We claim that $\tau$ must be a prefix of $t$. 

    Suppose, for a contradiction, that $t$ does not contain $\tau$ as
    its prefix.  Let $B_1$ be the last honest vertex on $t$ such that
    $\ell(B_1) \leq s - 1$.  (If $s = 1$ or no such vertex can be
    found, take $B_1$ as the root vertex.)  Likewise, let $B_2$ be the
    first honest vertex, if it exists, on $t$ such that
    $\ell(B_2) \in [s + 1, k - 1]$.

    Suppose $B_2$ exists.  If $\ell(B_2) = s + 1$ then, by
    Observation~\ref{obs:multi-honest}, $B_2$ builds on $\tau$,
    contradicting our assumption that $\tau$ is not a prefix of $t$.
    Otherwise, suppose $\ell(B_2) \in [s + 2, k-1]$.  Let $I$ be the
    interval $[\ell(B_1) + 1, \ell(B_2) - 1]$.  Clearly, $I$ contains
    $s$.  If $t$ contains any adversarial vertex between $B_1$ and
    $B_2$ then, by Corollary~\ref{coro:interval-honest-vertices}, $I$
    must be $\Aheavy$; but this contradicts the assumption that $s$ is
    a Catalan slot.  Otherwise, $B_2$ builds on top of $B_1$ and, in
    particular, $B_1$ must be viable at the onset of slot
    $\ell(B_2) \geq s + 1$.  Since $\ell(\tau) = s$, this means
    $\length(B_1) \geq \length(\tau)$.  However, since
    $\ell(B_1) < s$, by the monotonicity of the honest-depth function
    $\hdepth(\cdot)$, $\length(\tau) \geq 1 + \length(B_1)$.  This
    contradicts the inequality above.

    Now suppose $B_2$ does not exist.  We claim that $t$ is an
    adversarial tine.  To see why, note that if $t$ were honest and
    $\ell(t) \geq s + 1$ then there would have been a $B_2$.  Since
    $s$ is a uniquely honest slot and $\tau$ is not a prefix of $t$ by
    assumption, $\ell(t) \neq s$ if $t$ is honest.

    Finally, if $t$ is honest and $\ell(t) \leq s - 1$ then, 
    by Fact~\ref{fact:catalan-unique-longest}, 
    $t$ cannot be viable at the onset of slot $s + 1$ 
    since $s$ is Catalan. 
    Since $s + 1$ is an honest slot, 
    honest tines with label $s + 1$ will be strictly longer than $t$ 
    and, therefore, 
    $t$ cannot be viable at the onset of slot $k \geq s + 2$ either. 
    We conclude that $t$ must be an adversarial tine viable at the onset of slot $k$. 
    By Fact~\ref{fact:fork-structure},       
    the interval $I = [\ell(B_1) + 1, k - 1]$ must be $\Aheavy$. 
    However, since $I$ contains $s$, it contradicts the fact that $s$ is a Catalan slot. 

    It follows that every viable tine 
    $t \in F, \ell(t) \geq s + 1$ must contain $\tau$ as its prefix. 
  \end{proof}

The following theorem shows that under axiom~\ref{axiom:tie-breaking}, 
two consecutive Catalan slots 
imply that the first slot has the UVP.

  \begin{theorem}\label{thm:multiple-honest}
    Let $w \in \{\H, \A\}^T$ be a bivalent characteristic string 
    and axiom~\ref{axiom:tie-breaking} is satisfied.
    Let $s \in[2, T]$ be an integer such that 
    $s$ and $s - 1$ are two honest slots in $w$. 
    The following statements are equivalent:
    \begin{enumerate*}[label=(\textit{\roman*})]
      \item\label{thm:part-catalan-multiple-honest} 
      Slots $s, s - 1$ are Catalan.

      \item\label{thm:part-cp-multiple-honest} 
      If $s \leq T - 1$, both $s$ and $s - 1$ have the UVP. 
      Otherwise, slot $T - 1$ has the UVP but 
      slot $T$ has the bottleneck property.

    \end{enumerate*}
  \end{theorem}
  \begin{proof}
    Since the slots $s, s - 1$ satisfy the (weaker) bottleneck property, 
    Fact~\ref{fact:almost-cp-implies-catalan} implies that 
    they must be Catalan slots. 
    This proves 
    ~\ref{thm:part-cp-multiple-honest} implies~\ref{thm:part-catalan-multiple-honest}.
    
    Now let us prove that ~\ref{thm:part-catalan-multiple-honest} implies~\ref{thm:part-cp-multiple-honest}. 
    Slots $s, s - 1$ are Catalan. 
    Let $V_{s}$ (resp. $V_{s+1}$) be the set of all viable tines 
    at the onset of slot $s$ (resp. slot $s + 1$). 
    Since $s - 1$ (resp. $s$) is a Catalan slots, we use 
    Fact~\ref{fact:catalan-unique-longest} and conclude that 
    $V_s$ (resp. $V_{s+1}$) can contain 
    only maximum-length honest tines $t, \ell(t) = s - 1$ (resp. $\ell(t) = s$). 
    Let $u_s \in V_s$ be the unique vertex determined 
    by the consistent tie-breaking rule when applied to the set $V_s$. 
    Define $u_{s+1} \in V_{s+1}$ in an analogous way for the set $V_{s+1}$.

    Let $k \in [s + 1, T+1]$ be an integer.
    We wish tho show that for every tine $t$ viable at the onset of slot $k$, 
    the following holds: 
    (i) if  $s \leq T - 1$ then $u_s \Prefix u_{s+1} \PrefixEq t$, and 
    (ii) if $s = T$ then 
    $u_{T-1} \Prefix t$ where $\ell(t) = T$. 

    All tines at the honest slot $s$ build upon $u_s$. 
    If $s = T$, we are done. 
    Otherwise, i.e., if $s \leq T - 1$, 
    let $\tau = u_{s + 1}$ and note that $u_s \Prefix u_{s + 1} = \tau$. 
    If $k = s + 1$, we are done since by Fact~\ref{fact:catalan-unique-longest}, 
    every tine at the honest slot $k$ will build upon $\tau$.  

    It remains to reason about the case $s \leq T - 2$ and $k \geq s + 2$. 
    Consider a tine $t$ which is viable at the onset of slot $k$. 
    (All we know about $t$'s label is that $\ell(t) \leq k - 1$.) 
    We claim that $\tau \Prefix t$. 
    Suppose, towards a contradiction, that 
    $\tau$ is not a prefix of $t$. 
    Let $B_1$ be the last honest vertex on $t$ such that $\ell(B_1) \leq s - 1$.
    (If no such vertex can be found, take $B_1$ as the root vertex.) 
    Likewise, let $B_2$ be the first honest vertex on $t$ such that $\ell(B_2) \in [s + 1, k - 1]$. 
    
    Below, we show that every choice for $B_1, B_2$ leads to a contradiction 
    and, therefore, $\tau$ must be a prefix of $t$. 
    If $B_2$ exists then, 
    by construction, $\ell(B_1) < s < \ell(B_2) \leq k - 1$.
    If $\ell(B_2) = s + 1$ then, as we have argued earler, 
    $B_2$ must have built on $\tau$. 
    This contradicts our assumption that 
    $\tau$ is not a prefix of $t$. 
    Otherwise, suppose $\ell(B_2) \geq s + 2$.  
    Let $I$ be the interval $[\ell(B_1) + 1, \ell(B_2) - 1]$ and note that  
    $I$ contains $s$. 
    There can be two scenarios.
    If $t$ contains an adversarial vertex between $B_1$ and $B_2$ then,  
    by Corollary~\ref{coro:interval-honest-vertices}, 
    $I$ must be $\Aheavy$; but 
    this contradicts the assumption that $s$ is a Catatan slot. 
    Otherwise,
    $B_2$ builds on top of $B_1$ and, 
    in particular, 
    $B_1$ must be viable at the onset of slot $\ell(B_2) \geq s + 1$. 
    Since $\ell(\tau) = s$, this means $\length(B_1) \geq \length(\tau)$. 
    However, since $\ell(B_1) < s$, 
    by the monotonicity of the honest-depth function $\hdepth(\cdot)$, 
    $\length(\tau) \geq 1 + \length(B_1)$. 
    This contradicts the inequality above.           

    If $B_2$ does not exist then 
    we claim that $t$ is an adversarial tine. 
    To see why, note that if $t$ were honest and $\ell(t) \geq s + 1$ 
    then there would have been a $B_2$. 
    If $t$ were honest with $\ell(t) = s, t \neq \tau$ 
    then $t$ would not be viable at the onset of slot $s + 2$. 
    This is because $s$ is a Catalan slot and as such, 
    each vertex from slot $s + 1$ builds on $\tau, \length(\tau) \geq \length(t)$. 
    Hence tines viable at the onset of slot $s + 2$ must have length at least $1 + \length(\tau) > \length(t)$. 
    Finally, if $t$ is honest and $\ell(t) \leq s - 1$ then, 
    by Fact~\ref{fact:catalan-unique-longest}, 
    $t$ cannot be viable at the onset of slot $s + 1$ 
    since $s$ is Catalan.  
    Since $s + 1$ is an honest slot, 
    honest tines with label $s + 1$ will be strictly longer than $t$ 
    and, therefore, 
    $t$ cannot be viable at the onset of slot $k \geq s + 2$ either. 
    We conclude that $t$ must be an adversarial tine viable at the onset of slot $k$. 
    By Fact~\ref{fact:fork-structure},       
    the interval $I = [\ell(B_1) + 1, k - 1]$ must be $\Aheavy$. 
    However, since $I$ contains $s$, it contradicts the fact that $s$ is a Catalan slot.     
  \end{proof}

%% file: bounds.tex
In the previous section, 
we explored the structural connection 
between the UVP and Catalan slots. 
In this section, we present two bounds 
on the stochastic event ``Catalan slots are rare.'' 
Specifically,
Bound~\ref{bound:unique-honest-catalan} 
concerns uniquely honest Catalan slots and complements Theorem~\ref{thm:unique-honest}; 
Bound~\ref{bound:two-catalans} concerns 
two consecutive Catalan slots and complements Theorem~\ref{thm:multiple-honest}. 
We defer the proofs till the next section 
and prove the main theorems below.


Recall the $(\epsilon, p_\h)$-Bernoulli condition 
from~\ref{def:bernoulli-cond}.
\begin{bound}\label{bound:unique-honest-catalan}
  Let $T, s, k \in \NN, T \geq s + k$ and  $\epsilon, q_\h \in (0, 1)$. 
  Let $w$ be a characteristic string satisfying 
  the $(\epsilon, q_\h)$-Bernoulli condition 
  and let $y = w_s \ldots w_{s+k-1}$.
  Then 
  \[
    \Pr_w[\text{$w$ does not contain a uniquely honest Catalan slot in $y$}]  
      \leq 
      \exp\left(
        -k\cdot \Omega(\min(\epsilon^3, \epsilon^2 q_\h)) 
      \right)
      \,.
  \]
\end{bound}
In particular, 
when $q_\h = (1+\epsilon)/2$, 
the bound above coincides with the 
bound in~\cite{LinearConsistency}; 
it follows that the current analysis 
subsumes their result.

\begin{bound}\label{bound:two-catalans}
  Let $T, s, k \in \NN, T \geq s + k$ and  $\epsilon \in (0, 1)$. 
  Let $w$ be a bivalent characteristic string satisfying 
  the $(\epsilon, 0)$-Bernoulli condition 
  and let $y = w_s \ldots w_{s+k-1}$.
  Then 
  \[
    \Pr_w[\text{$w$ does not contain two consecutive Catalan slots in $y$}]  
      \leq 
      \exp\left(
        - k\cdot \Omega(\epsilon^3(1 + O(\epsilon))) 
      \right)
      \,.
  \]
\end{bound}


\paragraph{Proof of Theorem~\ref{thm:main}.}

  We consider the distribution $\mathcal{B}$ first. 
  Write $w = xyz, |x| = s - 1$.
  Recall that  
  $
    \mathbf{S}^{s,k}[\mathcal{B}] 
    = \Pr_{w \sim \mathcal{B}}[\text{$s$ is not $k$-settled in $w$}]
  $. 
  Theorem~\ref{thm:unique-honest} and Equation~\eqref{eq:settlement-uvp} implies that 
  if $w$ contains a uniquely honest Catalan slot $c \in [s, s + k]$ 
  then slot $s$ must be $k$-settled in $w$. 
  In fact, by virtue of Fact~\ref{fact:catalan-unique-longest}, 
  it suffices to take $c \in [s, s + k - 1]$, 
  i.e., $|x| \leq c \leq |xy|$. 
  Thus the probability above is bounded by 
  Bound~\ref{bound:unique-honest-catalan} 
  which renames $p_\h = q_\h$.
  This proves the first inequality. 
  
  Now let us prove the second inequality. 
  For any player playing the settlement game, 
  let $C$ be the set of strings on which the player wins. 
  Clearly, $C$ is monotone 
  with respect to the partial order $\leq$ 
  defined on $\{\h, \H, \A\}^T$ 
  (see below Definition~\ref{def:dominance}).  
  To see why, note that if the player wins 
  on a specific string $w$, 
  he can certainly win on any string $w^\prime$ so that $w \leq w^\prime$. 
  By assumption, 
  $\mathcal{W} \DominatedBy \mathcal{B}$. 
  It follows from Definition~\ref{def:dominance} that 
  $\Pr_{\mathcal{W}}[w] \leq \Pr_{\mathcal{B}}[w]$ 
  for any $w$ in the monotone set $C$. 
  By referring to the definition of settlement insecurity 
  (see Definition~\ref{def:settlement-insecurity}), 
  we conclude that 
  $
    \mathbf{S}^{s,k}[\mathcal{W}] \leq \mathbf{S}^{s,k}[\mathcal{B}]
  $.
  \hfill $\qed$  

\paragraph{Proof of Theorem~\ref{thm:main-bivalent}.}
  This proof is identical to the proof of Theorem~\ref{thm:main} 
  except that 
  we need to refer to Theorem~\ref{thm:multiple-honest} in lieu of Theorem~\ref{thm:unique-honest}
  and Bound~\ref{bound:two-catalans} in lieu of Bound~\ref{bound:unique-honest-catalan}.
  \hfill $\qed$

%% file: estimates.tex
\newcommand{\gfA}{\gf{A}}
\newcommand{\gfD}{\gf{D}}
\newcommand{\gfL}{\gf{L}}
\newcommand{\gfLhat}{\gf{\hat{L}}}
\newcommand{\gfLtilde}{\gf{\tilde{L}}}
\newcommand{\gfE}{\gf{E}}
\newcommand{\gfEhat}{\gf{\hat{E}}}
\newcommand{\gfEcheck}{\gf{\check{E}}}
\newcommand{\gfF}{\gf{F}}
\newcommand{\gfFhat}{\gf{\hat{F}}}
\newcommand{\gfFcheck}{\gf{\check{F}}}
\newcommand{\gfEtilde}{\gf{\tilde{E}}}
\newcommand{\gfR}{\gf{R}}
\newcommand{\gfRhat}{\gf{\hat{R}}}
\newcommand{\gfRtilde}{\gf{\tilde{R}}}

\newcommand{\gfC}{\gf{C}}
\newcommand{\gfChat}{\gf{\hat{C}}}
\newcommand{\gfCcheck}{\gf{\check{C}}}
\newcommand{\gfCtilde}{\gf{\tilde{C}}}

\newcommand{\gfM}{\gf{M}}
\newcommand{\gfMhat}{\gf{\hat{M}}}
\newcommand{\gfMtilde}{\gf{\tilde{M}}}

\newcommand{\gfX}{\gf{X}}
\newcommand{\gfXinf}{\gf{X}_\infty}

\newcommand{\EventRCat}{E_\mathsf{right-cat}}
\newcommand{\EventReset}{E_\mathsf{reset}}

\newcommand{\SeqGF}{\longleftrightarrow}

As a rule, we denote the
probability distribution associated with a random variable using
uppercase script letters. 
Observe that if $Y \dominatedby X$ and 
$Z$ is independent of both $X$ and $Y$, 
then $Z + Y \dominatedby Z + X$. 
In addition, for any non-decreasing function $u$ defined on $\Omega$, 
$Y \dominatedby X$ implies $u(Y) \leq u(X)$.

  \paragraph{Generating functions.}
    We reserve the term
    \emph{generating function} to refer to an ``ordinary'' generating
    function which represents a sequence $a_0, a_1, \ldots$ of
    non-negative real numbers by the formal power series
    $\gf{A}(Z) = \sum_{t = 0}^\infty a_t Z^t$. 
    We denote the above correspondence as $\{a_t\} \SeqGF \gfA(Z)$. 
    When
    $\gf{A}(1) = \sum_t a_t = 1$ we say that the generating function is
    a \emph{probability generating function}; in this case, the
    generating function $\gf{A}$ can naturally be associated with the
    integer-valued random variable $A$ for which $\Pr[A = k] = a_k$. If
    the probability generating functions $\gf{A}$ and $\gf{B}$ are
    associated with the random variables $A$ and $B$, it is easy to
    check that $\gf{A} \cdot \gf{B}$ is the generating function
    associated with the convolution $A + B$ (where $A$ and $B$ are
    assumed to be independent).  Translating the notion of stochastic
    dominance to the setting with generating functions, we say that the
    generating function $\gf{A}$ \emph{stochastically dominates}
    $\gf{B}$ if $\sum_{t \leq T} a_t \leq \sum_{t \leq T} b_t$ for all
    $T \geq 0$; we write $\gf{B} \dominatedby \gf{A}$ to denote this state of
    affairs. If $\gf{B}_1 \dominatedby \gf{A}_1$ and
    $\gf{B}_2 \dominatedby \gf{A}_2$ then
    $\gf{B}_1 \cdot \gf{B}_2 \dominatedby \gf{A}_1 \cdot \gf{A}_2$ and
    $\alpha \gf{B}_1 + \beta \gf{B}_2 \dominatedby \alpha \gf{A}_1 + \beta
    \gf{A}_2$ (for any $\alpha, \beta \geq 0$).  Moreover, if
    $\gf{B} \dominatedby \gf{A}$ then it can be checked that
    $\gf{B}(\gf{C}) \dominatedby \gf{A}(\gf{C})$ for any probability
    generating function $\gf{C}(Z)$, where we write $\gf{A}(\gf{C})$ to
    denote the composition $\gf{A}(\gf{C}(Z))$.

    Finally, we remark that
    if $\gf{A}(Z)$ is a generating function which converges as a
    function of a complex $Z$ for $|Z| < R$ for some non-negative $R$, 
    $R$ is called the \emph{radius of convergence} of $\gf{A}$.  
    It follows from Theorem 2.19 in~\cite{WilfGF} that 
    $\lim_{k \rightarrow \infty} {|a_k|}R^k = 0$ and that $|a_k| = O(R^{-k})$. 
    In addition, if $\gf{A}$ is a probability generating function associated with the
    random variable $A$ then it follows that
    $\Pr[A \geq T] = O(R^{-T})$.

\subsection{Proof of Bound~\ref{bound:unique-honest-catalan}}\label{sec:catalan-estimates}
  Let $p = (1 - \epsilon)/2$ and $q = (1 + \epsilon)/2$ 
  so that $q - p = \epsilon$. 
  Let $q_\H = q - q_\h$. 
  Let $B$ denote the event that 
  $w$ does not contain a uniquely honest Catalan slot in $y$. 
  We would like to bound $\Pr_w[B]$ from above.

  Define the process $W = (W_t : t \in \NN), W_t \in \{\pm 1\}$ as $W_t = 1$ if and only if $w_t = \A$. 
  Let $S = (S_t : t \in \NN), S_t = \sum_{i \leq t} W_i$ be the position of the particle at time $t$. 
  Thus $S$ is a random walk on $\ZZ$ with $\epsilon$ negative (i.e., downward) bias. 
  By convention, set $W_0 = S_0 = 0$. 

  \paragraph{Case 1: $x$ is an empty string.} 
  In this case, we write $w = yz$ so that $|y| = k$. 
  Let $c_t$ be the probability that $t$ is the first uniquely honest Catalan slot in $w$ 
  with $c_0 = 0$, and consider the probability generating function 
  $\{c_t\} \SeqGF \gfC(Z) = \sum_{t = 0}^\infty c_t Z^t$. 
  Controlling the decay of the coefficients $c_t$ suffices
  to give a bound on $\Pr[B]$, i.e., 
  the probability that 
  $y$ \emph{does not} contain a Catalan slot, 
  because this probability is at most 
  $
    1 - \sum_{t =0}^{k-1} c_t 
      = \sum_{t = k}^\infty c_t
  $. 
  To this end, we develop a
  closed-form expression for a related probability generating function
  $\gfChat(Z) = \sum_t \hat{c}_t Z^t$ which stochastically
  dominates $\gfC(Z)$. 
  Recall that this means that for any $k, \sum_{t \geq k} c_k \leq \sum_{t \geq k} \hat{c}_k$. 
  Finally, bound the latter sum  
  by using the analytic properties of $\gfChat(Z)$. 

  Treating the random variables $W_1, \ldots$ as
  defining a (negatively) biased random walk, define $\gfD$ (resp. $\gfA$) to be
  the generating function for the \emph{descent stopping time} 
  (resp. the \emph{ascent stopping time}) 
  of the walk; this is the first time the random walk, starting at 0, visits
  $-1$ (resp. $+1$). 
  The natural recursive formulation of these descent time yield 
  simple algebraic equations for the descent generating function,
  $\gfD(Z) = qZ + pZ \gfD(Z)^2$ and $\gfA(Z) = pZ + qZ \gfA(Z)^2$, 
  and from this we may conclude
  \begin{align*}
    \gfD(Z) &= (1 - \sqrt{1 - 4pqZ^2})/2pZ\,, \\
    \gfA(Z) &= (1 - \sqrt{1 - 4pqZ^2})/2qZ\,.
  \end{align*}
  Note that while $\gfD$ is a probability generating function, 
  $\gfA$ is not: according to the classical
  ``gambler's ruin'' analysis,
  the probability
  that a negatively-biased random walk starting at 0 ever rises to 1
  is exactly $p/q$; thus $\gfA(1) = p/q$.

  Recall that a slot is Catalan in $w$ if and only if 
  it is both left-Catalan and right-Catalan. 
  A slot is left-Catalan if the walk $S$ descends to a new low at that slot. 
  In addition, the same slot (say $s$) is right-Catalan 
  if the walk never reaches to that level in future, 
  i.e., $S_s \geq S_{i}, i \geq s + 1$. 
  The probability of this event is $1 - \gfA(1) = 1 - p/q = \epsilon/q$, 
  conditioned on the fact that $W_s = -1$.
  
  Assume that the walk is now at its historical minimum. 
  (It may or may not be a new minimum.)
  We can think of the generating function $\gfC(Z)$ as a search procedure 
  for finding the first uniquely honest Catalan slot. 
  Let $v$ be the first symbol we observe. 
  Let $\gfE(Z)$ be the generating function for a walk which 
  makes an ascent with certainty 
  and then descends again to its historical minimum.
  We claim that 
  \begin{align}
    \gfC(Z) 
      &= pZ \gfD(Z) \gfC(Z) + q_\h Z\cdot \epsilon/q + q_\h Z \cdot p/q \cdot \gfE(Z) \gfC(Z) + q_\H Z \gfC(Z) \nonumber
      \\
      &= \frac{(q_\h \epsilon/q) Z}{1 - \bigl( pZ \gfD(Z) + (q_\h p/q) Z \gfE(Z) + q_\H Z\bigr)}      
    \label{eq:gfC}
    \,.
  \end{align}
  Here is the explanation. 
  Regarding the value of $v$, there can be four alternatives 
  for the walk which is currently at its historical minimum: 
  \begin{enumerate}[label=(\textit{\roman*})]
    \item With probability $p$, we have $v = \A$ and the walk moves up. 
    Then we wait till the walk makes a first descent and restart.

    \item With probability $q_\h \cdot  \epsilon/q$, we have 
    $v = \h$ and the walk diverges below. 
    Hence our search has succeeded and we stop.

    \item With probability $q_\h \cdot  (1 - \epsilon/q) = q_\h p/q$, we have 
    $v = \h$ and the walk returns to the origin from below. 
    Then we wait for the walk to match its minimum again 
    before we can restart. 
    Note that $\gfE(Z)$ is 
    the generating function for this ``guaranteed ascent then match minimum'' walk.

    \item With probability $q_\H$, we have $v = \H$ and the walk moves down. 
    Since we will reach a new minimum, we restart.
  \end{enumerate}
  Since $\gfE(1) = 1$ by assumption, 
  $p  + (q_\h p/q) + q_\H = 1 - q_\h(1 - p/q) = 1 - q_\h\epsilon/q$. 
  It follows that 
  $\gfC(1) = (q_\h \epsilon/q) / (1 -(1 - q_\h\epsilon/q) ) = 1$; 
  hence $\gfC(Z)$ is a probability generating function.    

  Instead of working directly with $\gfE(Z)$, 
  we can work with a generating function $\gfEhat(Z)$ 
  which is identical to $\gfE(Z)$ for the initial ascending part 
  but differs in the descending part. 
  Specifically, in the descending part, 
  the walk represented by $\gfEhat(Z)$ descends as many levels 
  as the number of steps it took to return to the origin. 
  Clearly, $
      \gfE(Z) \DominatedBy \gfEhat(Z) \triangleq \gfA(Z \gfD(Z) )/\gfA(1)
  $. 
  Here, an individual term in $\gfA(Z \gfD(Z)) = \sum_i a_i Z^i \gfD(Z)^i$ 
  has the interpretation 
  ``if the first ascent took $i$ steps then immediately descend $i$ levels.''
  Since $\gfA(Z)$ is not a probability generating function, 
  we have to normalize it by $\gfA(1)$ to make sure that 
  the ascent happens with certainty. 
  Writing 
  \[
    \gfF(Z) \triangleq pZ \gfD(Z) + q_\h  Z \gfA(Z \gfD(Z) ) + q_\H Z
    \,,
  \]
  note that 
  \begin{align}
    &\gfC(Z) 
        \DominatedBy \gfChat(Z) 
        \triangleq 
        (q_\h \epsilon/q) Z/(1 - \gfF(Z) )
    \,. \label{eq:gfChat}
  \end{align}
  Since $\gfF(1) = p + q_\h p/q + q_\H = 1 - q_\h(1 - p/q) = 1 - q_\h \epsilon/q$, 
  we have $\gfChat(1) = 1$, i.e.,  
  $\gfChat(Z)$ is a probability generating function. 
  It remains to establish a bound on the radius of convergence of
  $\gfChat$. 
  A sufficient condition for the convergence of
  $\gfChat(z)$ for some $z \in \RR$ is 
  that all generating functions appearing in the definition of
  $\gfChat(z)$ converge at $z$ and 
  that $\gfF(z) \neq 1$. 

  The generating functions $\gfD(z)$ and $\gfA(z)$ converge when
  the discriminant $1 - 4pqz^2$ is positive; equivalently
  $|z| < 1/\sqrt{1 - \epsilon^2}
  = 1 + \epsilon^2/2 + O(\epsilon^4)$. 
  In addition, conditioned on the convergence of $\gfA(z)$ and $\gfD(z)$, 
  we can check that 
  \begin{equation}\label{eq:conv-value-A-and-D}
    \gfA(z) < 1/2qz\, \quad \text{and}\quad \gfD(z) < 1/2pz
    \,.
  \end{equation}
  On the other hand, the convergence of $\gfF(z)$ 
  depends on the convergence of $\gfD(z)$ and $\gfA(z \gfD(z))$. 
  The convergence of $\gfA(z \gfD(z))$ is likewise determined by the 
  positivity of its discriminant, i.e.,
  $$
    1 - (1 - \epsilon^2)\, \left(z \cdot \frac{1 - \sqrt{1 - (1 - \epsilon^2) z^2}}{(1 - \epsilon) z}\right)^2  > 0
    \,.
  $$
  The inequality above implies that 
  if $\gfA(z \gfD(z))$ converges when 
  $$
    |z| 
    < R_1 
    \triangleq \left(\left(2/\sqrt{1 - \epsilon^2} - 1/(1+\epsilon)\right)/(1 + \epsilon)\right)^{1/2} 
    \,,
  $$ where 
  \begin{align}\label{eq:roc-AZDZ}
    R_1 = 1 + \epsilon^3/2 + O(\epsilon^4) 
    \approx \exp\left(\epsilon^3 (1 + O(\epsilon))/2 \right)
    \,.
  \end{align}
  \noindent
  Note that the radius of convergence of $\gfA(Z \gfD(Z))$ 
  is smaller than that of $\gfA(Z)$ or $\gfD(Z)$.

  We can check that when $\gfF(z)$ converges, 
  it satisfies $$\gfF(z) \leq \gfF(|z|)\,.$$ 
  The claim is trivial for $z = 0$. 
  Otherwise,
  note that $\gfD(z)$ is an odd function and hence, 
  $z \gfD(z) = |z|\, \gfD(|z|)$. 
  Thus, for the claim to hold, we need only show that 
  $
    z\,(q_\h \gfA(z\gfD(z)) + q_\H) 
    \leq |z|\,(q_\h \gfA(|z|\gfD(|z|)) + q_\H) 
  $. 
  But the right-hand side equals $|z|\,(q_\h \gfA(z\gfD(z)) + q_\H)$ 
  and $\gfA(x) > 0$ for real $x > 0$, 
  we can divide both sides by $q_\h \gfA(z\gfD(z)) + q_\H$. 
  The reduced inequality becomes 
  $z/|z| \leq 1$. 
  However, $z/|z| = \pm 1$ for any non-zero real $z$.
  Therefore, it suffices for us to require that $F(z) \neq 1$ for $z > 0$. 
  
  We can also check that 
  \begin{equation}\label{eq:Fz-convex-increasing}
    \text{$F(z)$ is convex and increasing for $z \in [0, R_1)$}
    \,.
  \end{equation} 
  To see why, note that since $z^2$ is convex in $z$, 
  $(1 - 4pq z^2)$ is concave. 
  Since square root is non-decreasing and convex for positive $z$, 
  $\sqrt{1 - 4pqz^2}$ is concave and consequently, 
  $-\sqrt{1 - 4pqz^2}$ is convex. 
  Since $1/z^2$ is convex, 
  it follows that $\gfD(z)$ and, by a similar reasoning, $\gfA(z)$ are convex.
  Next, observe that $\gfA(z \gfD(z))$ converges for $z \in [0, R_1)$ 
  and hence it is also convex in $z$. 
  Thus $\gfF(z)$ turns out to be a convex combination of convex functions; 
  it follows that $F(z)$ is convex for $z \in (0, R_1)$. 
  In addition, 
  since $\gfF(0) = 0$ and $\gfF(1) > 0$, 
  $\gfF(z)$ must be increasing as well.

  Let $$\text{ $R_2$ be the solution to the equation $\gfF(z) = 1, z > 0$}\,. $$ 
  Then $\gfChat(z)$ would converge for $|z| < R \triangleq \min(R_1, R_2)$. 
  It remains to characterize $R_2$ in terms of $\epsilon$ and $q_\h$. 
  Note that $R_1 < 2$ as long as $\epsilon \leq 0.97$.
  Since the final bounds will be only asymptotic in $\epsilon$, 
  it suffices for us to consider small $\epsilon$. 
  That is to say, we consider the case where $0 < z < R_1 < 2$, 
  i.e., $z - 1 < 1$.

  If we express $\gfF(z)$ as its power series around $z = 1$, we can check that 
  \begin{align*}
    \gfF(1) &= 1 - \epsilon q_\h/q\,,\\
    \gfF''(1) &= 
      \frac{1-\epsilon}{\epsilon^5}\left( q_\h (1+3\epsilon) + q_\H \epsilon^2 \right)\,, \quad\text{and}\\
    \gfF'(1) &= p(1+1/\epsilon) + q_\h(p/q)\bigl( 1+(1+1/\epsilon)/\epsilon \bigr) + q_\H
    \,.
  \end{align*}
  Since $\gfF''(1) > 0$ and $\gfF(z)$ is convex and increasing, 
  the first-order approximation 
  \begin{equation}
    f(z) = (1 - \epsilon q_\h/q) + \gfF'(1)(z-1) 
  \end{equation}
  is a lower bound for $\gfF(z)$ when $1 \leq z < R_1$. 
  The approximation error at any $z \in (1, 2)$ is 
  $\gfF(z) - f(z) = O(h(z))$ 
  where we define $$h(z) \triangleq \gfF''(1) (z-1)^2\,.$$
  Since the bounds we develop will have 
  either $O(\cdot)$ or $\Omega(\cdot)$ in the exponent, 
  it suffices to ensure that $R_2 = \Theta(R_2^*)$.
  In the exposition below, 
  we will only develop approximations $R_2^*$ satisfying 
  $R_2 = (1 - \theta) R_2^*$ 
  for a small positive constant $\theta \in (0, 1)$. 

  In the special case $q_\H = 0$, 
  $\gfF(Z)$ simplifies as $\gfF(Z) = pZ \gfD(Z) + q  Z \gfA(Z \gfD(Z) )$. 
  Note that $\gfF(Z)$ converges when $\gfA(Z \gfD(Z) )$ does 
  and it is not hard to check that $\gfF(z) < 1$. 
  Specifically, 
  we know that $\gfF(z)$ converges when $z \in[0, R_1)$ 
  and when it does, we claim that $\gfF(z) < 1$. 
  Specifically, when $z \in [0, 1]$, 
  $\gfF(z) \leq \gfF(1) = 1 - \epsilon q_\h/q = 1 - \epsilon < 1$ since $\epsilon < 1$. 
  On the other hand, 
  we can check that $\gfD(z)$ is convex for $z \geq 0$ and, 
  in particular, the first order approximation $1 + (z-1)/\epsilon$ around $z = 1$ 
  is a lower bound for $\gfD(z), z \geq 1$.
  It follows that $\gfD(z) \geq 1$ for $z \in [1, R_1)$. 
  Consequently, 
  $\gfF(z) 
  \leq pZ \gfD(Z) + q  z \gfA(z \gfD(z) )\cdot \gfD(z) 
  = pz \gfD(z) + q  x \gfA(x) 
  < 1/2 + 1/2 = 1$ 
  where we write $x = z \gfD(z)$ and use~\eqref{eq:conv-value-A-and-D}. 
  Thus 
  the radius of convergence of $\gfChat$ is $R_1$ if $q_\H = 0$.

  The remainder of the exposition considers 
  the general case $0 < q_\h < q$. 
  Let the solution to the equation $f(z) = 1$ be denoted by 
  $$
    R_2^* \triangleq 1 + \epsilon (q_\h/q)/\gfF'(1)
    \,.
  $$ 
  If $q_\h$ is small, $q = (1+\epsilon)/2, p+\epsilon = q$ and $p/q^3 \in [1,4]$, 
  we can check that 
  \[
    h(R_2^*)
    = O\left(\frac{pq}{\epsilon^3} \cdot \left( \frac{\epsilon^2 q_h/q}{p(1+\epsilon) + \epsilon q} \right)^2 \right) 
    =O\left(\frac{ \epsilon q_\h^2 \cdot pq}{q^2\,(p + \epsilon )^2}\right)
    =O\left(\frac{ \epsilon q_\h^2 \cdot p}{q^3}\right)
    = O(\epsilon q_\h^2)
    \,,
  \]
  i.e., it vanishes.
  Thus $f(z)$ is a good approximation for $\gfF(z)$.
  It follows that 
  $\gfF'(1) \approx p(1+1/\epsilon) + q = q/\epsilon$ 
  and, therefore, 
  $$R_2^* 
    \approx 1 + (\epsilon q_\h/q)/(q/\epsilon) 
    = 1 + q_\h(\epsilon/q)^2 
    \approx \exp(\epsilon^2 q_\h/q^2)
    = e^{O(\epsilon^2 q_\h)}
  $$ 
  since $q \in (1/2, 1)$. 
  (Although we have an asymptotic notation, 
  it is important that we have the right exponent on $q_\h$.)
  
  If, on the contrary, $q_\h = O(1)$ but $\epsilon$ vanishes then 
  $\gfF'(1)$ will be dominated by its second term; 
  that is to say, 
  $\gfF'(1) 
  \approx q_\h(p/q) \left(1+(1+1/\epsilon)/\epsilon \right) 
  = O(q_\h/\epsilon^2)
  $
  and, therefore, 
  $$R_2^* 
  \approx 1 + O\left( (\epsilon q_\h/q)/( q_\h/\epsilon^2) \right) 
  = 1 + O(\epsilon^3) 
  = e^{O(\epsilon^3)}
  $$ since $q \approx 1/2$.

  Recall that $R_1 = \exp\left(O(\epsilon^3 (1 + O(\epsilon)))\right)$. 
  It follows that $\gfChat(z)$ converges for 
  $|z|$ less than 
  \begin{align}\label{eq:RoC-unique-honest}
    R &= \exp\left(O(\min(\epsilon^3, \epsilon^2 q_\h))\right)
    \,.
  \end{align}

  Recall that if the radius of convergence of
  $\gfChat$ is $\exp(\delta)$ then 
  $\hat{c}_k = O(e^{-\delta k})$. 
  Hence, $\Pr[B]$ is a geometric sum and it is 
  at most $O(e^{-\delta k})$ as well. 
  We conclude that 
  $$
    \Pr_w[B] 
      \leq O\left(e^{-k \ln R }\right)
      = \exp\left(-k\cdot \Omega(\min(\epsilon^3, \epsilon^2 q_\h))\right)
      \,.
  $$

  \paragraph{Case 2: $x$ is non-empty.}
  Next, let us consider the case when $x \neq \varepsilon$, i.e., $|x| \geq 1$. 
  Let $m = |x|$ and write $w = xyz$ where $|y| = k$. 
  Recall the processes $(W_t)$ and $(S_t)$ defined on $w$
  and, in addition, define $M = (M_t : t \in \NN), M_t = \min_{0 \leq i \leq t } S_i$ 
  and $X = (X_t : t \in \NN), X_t = S_t - M_t$. 
  By convention, set $M_0 = X_0 = 0$. 
  Thus $X_t$ denotes the height of the walk $S$, at time $t$, 
  with respect to its minimum $M_t$.

  For a fixed value $h = X_m$, the relevant generating function 
  would be $\gfD(Z)^{h}\gfChat$. 
  Hence the final generating function we seek is
  $$
    \gfCtilde(Z) \defeq \sum_{h = 0}^\infty \Pr[X_m = h] \cdot \gfD(Z)^h  \gfChat(Z)
  $$
  whose $t$th coefficient is the probability that 
  $t$ is a Catalan slot in $y$.

  Note that $X = (X_t)$ is an $\epsilon$-downward biased random walk 
  on non-negative integers with a reflective barrier at $-1$. 
  Specifically, 
  for any $h \in \NN, \Pr[X_t = h \Given X_{t-1} = h -1] = p$ and 
  $\Pr[X_t = h - 1 \Given X_{t-1} = h] = \Pr[X_t = 0 \Given X_{t-1} = 0] = q$. 
  In~\cite[Lemma 6.1]{LinearConsistencySODA}, it is proved that 
  the distribution of $X_m$ is stochastically dominated by 
  the distribution of $X_\infty$, written $\mathcal{X}_\infty$ and 
  defined, for $k = 0, 1, 2, \ldots$, as 
  \begin{equation}
    \label{eq:stationary}
      \mathcal{X}_\infty(k) = \Pr[X_\infty = k] 
      \defeq 
        \left(\frac{2\epsilon}{1+\epsilon}\right)\cdot \left(\frac{1-\epsilon}{1 + \epsilon}\right)^k
        = 
      (1 - \beta) \beta^k
  \end{equation}
  where $\beta \defeq (1-\epsilon)/(1+\epsilon)$. 
  Let 
  $$
    \{\mathcal{X}_\infty(k)\} \SeqGF \gfXinf(Z) = 
    \frac{1 - \beta}{1 - \beta Z}
    \,.
  $$ 
  It follows that $\gfCtilde(Z)$ is dominated by 
  $$
      \sum_{h = 0}^\infty \mathcal{X}_\infty(h) \gfD(Z)^h \gfChat(Z)
    = \gfXinf(\gfD(Z)) \gfChat(Z)
    = \frac{(1 - \beta)\gfChat(Z)}{1 - \beta \gfD(Z)}
    \,.
  $$

  Let $\star$ denote the quantity above. 
  For it to converge, 
  we need to check that $\gf{D}(Z)$
  should never converge to $1/\beta$.  
  Since the radius of convergence of $\gf{D}(Z)$---which is
  $(1-\epsilon^2)^{-1/2}$---is strictly less than 
  $(1+\epsilon)/(1-\epsilon)$ for $\epsilon > 0$, 
  we conclude that $\star$ converges if
  both $\gf{D}(Z)$ and $\gfChat(Z)$ converge.  The radius of
  convergence of $\star$ would be the smaller of the radii
  of convergence of $\gfD(Z)$ and $\gfChat(Z)$.  We already
  know from the previous analysis that $\gfChat(Z)$ has the
  smaller radius of convergence of these two; 
  therefore, the bound
  on $\Pr_w[B]$ from the previous case holds for $|x| \geq 0$. 
  \hfill$\qed$

\subsection{Proof of Bound~\ref{bound:two-catalans}}
  \input{proof-bound-two-cat}

%% file: proof-bound-two-cat.tex
  Let $p = (1 - \epsilon)/2$ and $q = 1 - p$; 
  thus $q - p = \epsilon$. 
  Let $B$ denote the event that 
  $w$ does not contain two consecutive Catalan slots in $y$. 
  We would like to bound $\Pr_w[B]$ from above.

  Define the process $W = (W_t : t \in \NN), W_t \in \{\pm 1\}$ as $W_t = 1$ if and only if $w_t = \A$. 
  Let $S = (S_t : t \in \NN), S_t = \sum_{i \leq t} W_i$ be the position of the particle at time $t$. 
  Thus $S$ is a random walk on $\ZZ$ with $\epsilon$ negative (i.e., downward) bias. 
  By convention, set $W_0 = S_0 = 0$. 

  \paragraph{Case 1: $x$ is an empty string.} 
  In this case, we write $w = yz$ so that $|y| = k$. 
  Let $m_t$ denote the probability that 
  $t$ is the first index so that both $t$ and $t+1$ are Catalan slots in $w$, 
  with $m_0 = 0$, and consider the probability generating function 
  $\{m_t\} \SeqGF \gfM(Z) = \sum_{t = 0}^\infty m_t Z^t$. 
  Controlling the decay of the coefficients $m_t$ suffices
  to give a bound on $\Pr[B]$, i.e., 
  the probability that 
  $y$ \emph{does not} contain two consecutive Catalan slots, 
  because this probability is at most 
  $
    1 - \sum_{t =0}^{k-1} m_t 
      = \sum_{t = k}^\infty m_t
  $. 
  To this end, we develop a
  closed-form expression for a related probability generating function
  $\gfMhat(Z) = \sum_t \hat{m}_t Z^t$ which stochastically
  dominates $\gfM(Z)$. 
  Recall that this means that for any $k, \sum_{t \geq k} m_k \leq \sum_{t \geq k} \hat{m}_k$. 
  Finally, bound the latter sum  
  by using the analytic properties of $\gfMhat(Z)$.

  Recall the ``first ascent'' and ``first descent'' 
  generating functions $\gfA(Z)$ and $\gfD(Z)$ 
  from the proof of Bound~\ref{bound:unique-honest-catalan}. 
  We wish to devise the generating function for 
  the first occurrence of a left-Catalan slot 
  immediately followed by a right-Catalan slot. 
  To that end, 
  note that $\gfD(Z)$ is the generating function for 
  the first left-Catalan slot. 
  The generating function for the first right-Catalan slot 
  can be devised as follows. 
  Consider the walk $S$ starting at the origin. 
  With probability $q(1 - p/q) = \epsilon$, the walk will 
  immediately descend a step and never return to the origin. 
  But this means $S_1 \leq S_t, t \geq 2$ and hence 
  the first slot is a right-Catalan slot and we are done. 
  Otherwise, i.e., with probability $1 - \epsilon$, 
  the walk makes a (guaranteed) return to the origin in future. 
  In this case, we will have to restart our search 
  for the next consecutive Catalan slots but, 
  before that, 
  we will have to ensure that we are in a ``safe position.'' 
  In particular, we can safely restart our search if 
  Specifically, if the current position (i.e., level) of the walk is at its historical minimum, 
  we can restart our search by applying $\gfD(Z)$ to find the next left-Catalan slot.
  Thus an ``epoch'' begins with a guaranteed return and 
  ends when the walk descends to a new level for the first time. 
  Let $\gfE(Z)$ be the generating function of an epoch. 
  Thus we can write 
  \begin{align}\label{eq:gfM}
    \gfM(Z) 
    &= \gfD(Z) \cdot \{\epsilon + (1-\epsilon)\gfE(Z)\gfM(Z) \} \nonumber \\
    &= \frac{\epsilon \gfD(Z)}{1 - (1 - \epsilon) \gfE(Z) }
    \,.
  \end{align}

  An epoch can have two shapes. 
  If an epoch starts with an up-step (i.e., an ``up'' shape), 
  it is easy to see that the epoch ends as soon as the walk 
  returns to the origin from above and, importantly, 
  that the walk will (eventually) return to the origin with probability one. 
  However, if the epoch starts with a down-step (i.e., a ``down'' shape), 
  we have to ``remember'' the lowest level $\ell$ touched 
  by the walk in its way to its (sure) ascent to the origin 
  and then descend $\ell$ levels to end the epoch. 
  In particular, we have to ensure that we return to the origin with probability one. 
  
  A generating function of a stopping time of a random walk 
  is ill suited to ``remember'' its historical minimum/maximum. 
  However, it can remember the length of the walk for free. 
  Thus, instead of working directly with $\gfE(Z)$, 
  we work with a generating function $\gfEhat(Z)$ 
  which is identical to $\gfE(Z)$ for the up shape 
  but differs in the down shape. 
  Specifically, in the down shape, 
  the walk represented by $\gfEhat(Z)$ descends as many levels 
  as the number of steps it took to return to the origin. 
  Clearly, $\gfE \DominatedBy \gfEhat$ where 
  \[
      \gfEhat(Z) \triangleq p Z \gfD(Z) + q Z \gfA(Z \gfD(Z) )/\gfA(1)
      \,.
  \] 
  Here, the first term denotes the ``return to origin from above'' shape. 
  An individual term in $\gfA(Z \gfD(Z)) = \sum_t a_t Z^t \gfD(Z)^t$ 
  has the interpretation 
  ``if the first ascent took $t$ steps then follow it by descending $t$ levels.''
  Since $\gfA(Z)$ is not a probability generating function, 
  we have to normalize it by $\gfA(1)$ to denote that 
  the ascent happens with certainty. 
  This implies, 
  \[
      \gfM(Z) 
          \DominatedBy \gfMhat(Z) 
          \triangleq \frac{\epsilon \gfD(Z)}{1 - (1 - \epsilon) \gfEhat(Z) }
  \]

  It remains to establish a bound on the radius of convergence of
  $\gfMhat$. 
  A sufficient condition for the convergence of
  $\gfMhat(z)$ for some $z \in \RR$ is 
  that all generating functions appearing in the definition of
  $\gfMhat$ converge at $z$ and 
  that $(1-\epsilon) \gfEhat(Z) \neq 1$. 

  By retracing our footsteps as in the proof of Bound~\ref{bound:unique-honest-catalan}, 
  we can see that $\gfD(z), \gfA(z)$, and $\gfA(z \gfD(z))$ converge 
  when $|z|$ satisfies~\eqref{eq:roc-AZDZ}. 
  Moreover, since $\gfD(Z)$ is a probability generating function, 
  it follows that $\gfEhat(Z)$ is stochastically dominated by 
  $p Z \gfD(Z) + q Z \gfA(Z \gfD(Z) )/\gfA(1) \cdot \gfD(Z)$.
  Therefore, when $\gfEhat(z)$ converges for some $z$, it satisfies 
  \begin{align*}
      \gfEhat(z)
      &\leq pz\gfD(z) + (q/p) (q z\gfD(z))\gfA(z\gfD(z)) \\
      &< 1/2 + (q/p)/2
  \end{align*}
  since $\gfA(1) = p/q, pz\gfD(z) < 1/2$, 
  and $qx\gfA(x) < 1/2$ for any $z, x$ so that $\gfA(x)$ and $\gfD(z)$ converge, respectively. 
  Therefore, $(1-\epsilon)\gfEhat(z) = 2p \gfEhat(z) < p + q = 1$. 
  It follows that 
  $\gfMhat(z)$ converges for
  $|z| < 1 + \epsilon^3/2 + O(\epsilon^4) \leq \exp(\epsilon^3/2 + O(\epsilon^4))$. 
  Recall that if the radius of convergence of
  $\gfMhat$ is $\exp(\delta)$ then 
  $\Pr[B]$ is 
  $O(1) \cdot e^{-\delta k}$. 
  We conclude that
  \begin{align}
    \Pr_w[B] 
      &\leq O(1) \cdot e^{-\epsilon^3(1 + O(\epsilon))k/2} \,.
  \label{eq:prob_two_catalan_gf}
  \end{align}

  \paragraph{Case 2: $x$ is non-empty.}
  This part of the proof is the same as the $|x| \geq 1$ case 
  in the proof of Bound~\ref{bound:unique-honest-catalan}. 
  The only difference is that 
  $\gfChat(Z)$ and $\gfCtilde(Z)$ would be replaced by 
  $\gfMhat(Z)$ and $\gfMtilde(Z)$, respectively, where 
  \[
    \gfMtilde(Z)\DominatedBy \sum_{h = 0}^\infty \mathcal{X}_\infty(h) \gfD(Z)^h \gfMhat(Z)
    \,.
  \]
  We conclude that the bound
  in~\eqref{eq:prob_two_catalan_gf} holds when $|x| \geq 0$. 
  \hfill$\qed$

%% file: recursive-formulation.tex
In this section, we introduce additional elements of the fork framework from~\citet{LinearConsistency}, 
most notably the notions of ``reach'' and ``relative margin.'' 
We show that relative margin is just as expressive 
as the Catalan slots 
for characterizing slot settlement. 
Next, we prove a recurrence relation for relative margin; 
it can be used to compute 
the probability that a given slot is $k$-settled, 
when the symbols of the characteristic string are i.i.d\ .
Finally, we present an adversary who, 
given a characteristic string one symbol at a time, 
optimally attacks the settlement of all slots at once.

\input{reach-and-margin}

\subsection{Balanced forks, settlement violations, and relative margin}

A natural structure we can use to reason about settlement times 
(see Definition~\ref{def:settlement}) 
is that of a ``balanced fork.''

\begin{definition}[Balanced fork]\label{def:balanced-fork} A
  fork $F$ is \emph{balanced} if it contains a pair of tines $t_1$ and
  $t_2$ for which $t_1\nsim t_2$ and
  $\length(t_1)=\length(t_2)=\height(F)$. We define a relative notion
  of balance as follows: a fork $F \vdash xy$ is \emph{$x$-balanced}
  if it contains a pair of tines $t_1$ and $t_2$ for which
  $t_1 \not\sim_x t_2$ and $\length(t_1) = \length(t_2) = \height(F)$.
\end{definition}

Thus, balanced forks contain two completely disjoint, maximum-length
tines, while $x$-balanced forks contain two maximum-length tines that
may share edges in $x$ but must be disjoint over the rest of the
string. 
See Figures~\ref{fig:balanced} and~\ref{fig:x-balanced} 
for examples of balanced forks.
\begin{figure}[ht]
  \centering
  \begin{minipage}{0.45\textwidth}\centering

    \begin{tikzpicture}[>=stealth', auto, semithick,
      honest/.style={circle,draw=black,thick,text=black,double,font=\small},
     malicious/.style={fill=gray!10,circle,draw=black,thick,text=black,font=\small}]
      \node at (0,-2) {$w =$};
    \node at (1,-2) {$\h$}; \node[honest] at (1,-.7) (b1) {$1$};
    \node at (2,-2) {$\A$}; \node[malicious] at (2,.7) (a1) {$2$};
    \node at (3,-2) {$\h$}; \node[honest] at (3,.7) (a2) {$3$};
    \node at (4,-2) {$\A$}; \node[malicious] at (4,-.7) (b2) {$4$};
    \node at (5,-2) {$\h$}; \node[honest] at (5,-.7) (b3) {$5$};
    \node at (6,-2) {$\A$}; \node[malicious] at (6,.7) (a3) {$6$};
      \node[honest] at (-1,0) (base) {$0$};
    \draw[thick,->] (base) to[bend left=10] (a1);
        \draw[thick,->] (a1) -- (a2);
        \draw[thick,->] (a2) -- (a3);
    \draw[thick,->] (base) to[bend right=10] (b1);
        \draw[thick,->] (b1) -- (b2);
        \draw[thick,->] (b2) -- (b3);
      \end{tikzpicture} 
    \caption{A balanced fork}
    \label{fig:balanced}
  \end{minipage}
  \hfill
  \begin{minipage}{0.45\textwidth}\centering

    \begin{tikzpicture}[>=stealth', auto, semithick,
      honest/.style={circle,draw=black,thick,text=black,double,font=\small},
     malicious/.style={fill=gray!10,circle,draw=black,thick,text=black,font=\small}]
      \node at (0,-2) {$w =$};
    \node at (1,-2) {$\h$}; \node[honest] at (1,0) (ab1) {$1$};
    \node at (2,-2) {$\h$}; \node[honest] at (2,0) (ab2) {$2$};
    \node at (3,-2) {$\h$}; \node[honest] at (3,.7) (a3) {$3$};
    \node at (4,-2) {$\A$}; \node[malicious] at (4,-.7) (b3) {$4$};
    \node at (5,-2) {$\h$}; \node[honest] at (5,-.7) (b4) {$5$};
    \node at (6,-2) {$\A$}; \node[malicious] at (6,.7) (a4) {$6$};
      \node[honest] at (-1,0) (base) {$0$};
    \draw[thick,->] (base) -- (ab1);
    \draw[thick,->] (ab1) -- (ab2);
    \draw[thick,->] (ab2) to[bend left=10] (a3);
      \draw[thick,->] (a3) -- (a4);
    \draw[thick,->] (ab2) to[bend right=10] (b3);
      \draw[thick,->] (b3) -- (b4);
      \end{tikzpicture} 
    \caption{An $x$-balanced fork, where $x=\h\h$}
    \label{fig:x-balanced}

  \end{minipage}
\end{figure}

A fundamental question arising in typical blockchain settings is how
to determine \emph{settlement time}, the delay after which the
contents of a particular block of a blockchain can be considered
stable. The existence of a balanced fork is a precise indicator for
``settlement violations'' in this sense. Specifically, consider a
characteristic string $xy$ and a transaction appearing in a block
associated with the first slot of $y$ (that is, slot $|x| + 1$). One
clear violation of settlement at this point of the execution is the
existence of two chains---each of maximum length---which diverge
\emph{prior to $y$}; in particular, this indicates that there is an
$x$-balanced fork $F$ for $xy$. Let us record this observation below.\footnote{
  A balanced fork in~\cite{LinearConsistency} 
  had the property that 
  at least one maximum-length tine was adversarial. 
  But this is not true in our setting since we allow multiply honest slots.
}

\begin{observation}\label{obs:settlement-balanced-fork}
  Let $s, k \in \NN$ be given and 
  let $w$ be a characteristic string. 
  Slot $s$ is not $k$-settled for the characteristic string $w$ 
  if 
  there exist a decomposition $w = xyz$, 
  where $|x| = s - 1$ and $|y| \geq k+1$, 
  and an $x$-balanced fork for $xy$. 
\end{observation}

In particular, to provide a rigorous $k$-slot settlement
guarantee---which is to say that the transaction can be considered
settled once $k$ slots have gone by---it suffices to show that with
overwhelming probability in choice of the characteristic string
determined by the leader election process (of a full execution of the
protocol), no such forks are possible. Specifically, if the protocol
runs for a total of $T$ time steps yielding the characteristics string
$w = xy$ (where $w \in \{0,1\}^T$ and the transaction of interest
appears in slot $|x| + 1$ as above) then it suffices to ensure that
there is no $x$-balanced fork for $x\hat{y}$, where $\hat{y}$ is an
arbitrary prefix of $y$ of length at least $k + 1$. 
Note that
for systems adopting the longest chain rule, this condition must
necessarily involve the \emph{entire future dynamics} of the
blockchain. We remark that our analysis below will in fact let us take
$T = \infty$.

Let $w$ be a characteristic string. 
Writing $w = xy$, 
consider any tine-pair $(t_x, t_\rho)$ in a fork $F \Fork w$ so that 
$\reach_F(t_\rho) = \rho(F)$ and $t_x$ is $y$-disjoint with $t_\rho$.
Observe that if $\mu_x(y) < 0$ then $\reach_F(t_x) < 0$. 

\begin{fact}\label{fact:margin-balance}
  Let $xy \in \{\h, \H, \A\}^*$ be a characteristic string. 
  There is no 
  $x$-balanced fork for $xy$ 
  if and only if 
  $\mu_x(y) < 0$.
\end{fact}
\begin{proof}[Proof sketch.]  
  If a fork $F \Fork xy$ 
  satisfies $\mu_x(F) \geq 0$, 
  it contains two $y$-disjoint tines $t_1, t_2$, 
  each with a non-negative reach, 
  so that $\min(\reach(t_1), \reach(t_2)) = \mu_x(F)$. 
  As $\reserve(t_i) \geq \gap(t_i)$ for $i \in \{1,2\}$, 
  we can extend these tines using only new adversarial vertices 
  so that both these extensions 
  have the maximum length 
  in the augmented fork. 
  Thus the augmented fork is $x$-balanced.

  On the other hand, if a fork $F \Fork xy$ is $x$-balanced, 
  there must be two $y$-disjoint 
  maximum-length tines $t_1, t_2 \in F$. 
  As the gap of a maximum-length tine is zero, 
  we must have $\reach(t_i) = \reserve(t_i) \geq 0$ 
  for $i \in \{1, 2\}$. 
  It follows that $\mu_x(y) \geq \mu_x(F) \geq \min_i \reach(t_i) \geq 0$.
\end{proof}


\subsection{Relative margin 
to characterize 
the UVP}
Let $w$ be a characteristic string. 
Recall that in Theorem~\ref{thm:unique-honest}, 
we showed that whether a slot has the UVP in $w$ --- a 
structural property of the forks for $w$ --- is 
characterized by the ``Catalan-ness'' of the said slot. 
Below, we show that relative margin 
has the same expressive power 
as the Catalan slots 
in terms of characterizing the UVP.


\begin{lemma}\label{lemma:uvp-margin}
  Let $T \in \NN, w \in \{\h, \H, \A\}^T$, and 
  $s \in [T]$ so that $w_s = \h$. 
  Let $x = w_1 \ldots w_{s-1}$.
  Slot $s$ has the UVP in $w$ 
  if and only if 
  for every prefix $xy \PrefixEq w$, 
  $\mu_x(y) < 0$. 
\end{lemma}


\begin{proof}~

  \begin{description}[font=\normalfont\itshape\space]
    \item[The $\Longleftarrow$ direction.]
      Suppose that 
      for every prefix $xy \PrefixEq w$ where $|y| \geq 1$, 
      we have $\mu_x(y) < 0$. 
      We wish to show that $s$ has the UVP in $w$.

      Let $F$ be any fork for $xy$ 
      and let 
      $t \in F, \ell(t) \leq s - 1$ be an honest tine. 
      Since it is disjoint with any tine in $F$ over the suffix $y$, 
      $\reach(t) < 0$ and, by Fact~\ref{fact:fork-structure-reach}, 
      $t$ does not have an adversarial extension $t' \in F, t \Prefix t'$ that is 
      viable at the onset of slot $|xy| + 1$. 
      Therefore, if a tine in $F$ 
      is viable at the onset of slot $|xy| + 1$, 
      it must contain an honest vertex with label at least $s$. 
      However, since an honest vertex builds only on top of a viable tine, 
      it follows that any viable tine must contain 
      the unique honest vertex with label $s$.

    \item[The $\Longrightarrow$ direction.]
      Suppose $s$ has the UVP in $w$.
      Let $k \in [s, T]$ be an integer and 
      write $w = xyz$ with $|xy| = k$. 
      (Note that $y_1 = w_s$.)
      We wish to show that $\mu_x(y) < 0$.

      Let $F$ be any fork for $xy$.
      Since slot $s$ belongs to $y$, 
      $F$ cannot contain two tines 
      such that 
      \begin{enumerate*}[label=(\roman*)]
        \item both tines are viable at the onset of slot $|xy| + 1$ 
        and, at the same time, 
        \item disjoint over the length of $y$ 
        since they must contain the unique vertex with label $s$. 
      \end{enumerate*}
      In particular, 
      $F$ cannot be $x$-balanced. 
      As $F$ was an arbitrary fork for $xy$, 
      no fork for $xy$ can be $x$-balanced for our choice of $k$.
      We use Fact~\ref{fact:margin-balance} 
      to conclude that 
      $\mu_x(y)$ must be negative.

  \end{description}
\end{proof}

\subsection{An optimal online adversary against slot settlement}
\label{sec:opt-adversary}
Let $w$ be a characteristic string. 
For a fixed decomposition $w = xy$, 
there is an adversary\footnote{
  Specifically, 
  let $w' = xyb$ 
  where $b \in \{\h, \H, \A\}$. 
  This strategy recursively builds a closed fork $F \Fork xy$. 
  Then, upon encountering $b$, 
  it augments $F$ 
  by making zero, one, or two conservative extensions, as follows: 
  If $b = \A$, it does nothing. 
  If $b = \h$, it extends a zero-reach tine if possible; 
  otherwise,it extends a maximum-reach tine. 
  If $b = \H$, it extends a pair of tines that 
  witness $\mu_x(F)$. 
  By following the arguments in~\cite{LinearConsistencySODA}, 
  one can show that 
  if $\mu_x(F) = \mu_x(y)$ then 
  $\mu_x(F')$ is 
  at least as large as 
  the right-hand side in~\eqref{eq:mu-relative-recursive}.   
} 
who builds a fork $F \Fork xy$ 
so that the $\mu_x(F)$ is 
at least as large as 
the right-hand side of~\eqref{eq:mu-relative-recursive}. 
However, 
in light of Lemma~\ref{lemma:uvp-margin}, 
if an adversary wants to violate the settlement 
\emph{of all possible slots of $w$ at once}, 
he needs to produce a fork $F$ for $w$ 
so that $\mu_x(F) \geq 0$ 
for every prefix $x \PrefixEq w$. 
In Figure~\ref{fig:adv-opt}, 
we describe a strategy $\Adversary^*$ 
which does even better: 
it produces a fork $F$ so that $\mu_x(F) = \mu_x(y)$ 
for every prefix $x \PrefixEq w$.

$\Adversary^*$ builds a fork for $w = w_1 \ldots w_{n+1}$ 
in an online fashion, i.e., 
it scans $w$ once, from left to right, 
maintains a fork $F_n$ after scanning 
the first $n$ symbols, 
and augments $F_n$ by conservatively extending 
zero-reach tine(s) using label $n + 1$.
Specifically, if $w_{n+1} = \A$, $\Adversary^*$ does nothing. 
If $w_{n+1} = \h$, it (obviously) makes a single extension. 
Now suppose $w_{n+1} = \H$. 
It still makes a single extension
if either $F_n$ contains exactly one zero-reach tine 
or $F_n$'s reach is positive. 
Otherwise, 
if $\rho(F_n) = 0$ 
and there are at least two zero-reach tines in $F_n$, 
$\Adversary^*$ extends two zero-reach tines 
that diverge earliest in $F_n$.

\begin{figure}[!h]
  \begin{center}
    \fbox{
      \begin{minipage}{.9 \textwidth}
        \begin{center}
          \textbf{The strategy $\Adversary^*$}
        \end{center}
        Let $n$ be a non-negative integer, 
        $w \in \{\h, \H, \A\}^n$, 
        and $w_{n + 1} \in \{\h, \H, \A\}$. 
        If $n = 0$, set $F_0 \Fork \varepsilon$ as 
        the trivial fork comprising a single vertex. 
        Otherwise, 
        let $F_n$ be the closed fork 
        built recursively by $\Adversary^*$ for the string $w$. 
        If $w_{n + 1} = \A$, 
        output $F_n$ (as a fork for $w w_{n+1}$). 
        Otherwise, 
        let $Z$ and $R$ be the set of zero-reach tines 
        and maximum-reach tines in $F_n$, respectively.

        \begin{enumerate}
          \item 
          Identify a set $S$ as follows: 
          If $|Z| = 1$ then set $S = Z$. 
          Otherwise, 
          let $r_1 \in R, z_1 \in Z$ be two tines so that 
          $\ell(r_1 \Intersect z_1) = 
          \min\{ \ell(r \Intersect z) :  r \in R, z \in Z \}$ 
          and set  
          \[
          S = \begin{cases}
            \{z_1\} & \text{
              if $w_{n + 1} = \h$ 
              or $\rho(F_n) \geq 1$ 
              }\,, \\
            \{z_1, r_1\} & \text{otherwise}\,.
          \end{cases}
          \]

          \item
          Conservatively extend 
          each tine in $S$ 
          using label $n + 1$. 
          Let $F_{n + 1} \Fork w w_{n+1}$ 
          be the new closed fork. 
          Output $F_{n+1}$.
        \end{enumerate}
      \end{minipage}
    }
  \end{center}
  \caption{Optimal online adversary $\Adversary^*$}
  \label{fig:adv-opt}
\end{figure}

\begin{definition}[Canonical fork]
  A \emph{canonical fork} for $w \in \{\h, \H, \A\}^*$ 
  is a closed fork $F \Fork w$ so that 
  $\rho(F) = \rho(w)$ 
  and, for all prefixes $x \Prefix w$, $\mu_x(F) = \mu_x(y)$. 
  If $|w| = 0$, $F$ is 
  the unique fork with a single (honest) vertex and no edge. 
\end{definition}

It is not obvious whether a canonical fork always exists 
or whether it can be found algorithmically. 
The theorem below gives us the assurance:

\begin{theorem}\label{thm:opt-adversary-canonical}
  Let $w \in \{\h, \H, \A\}^*$. 
  The strategy $\Adversary^*$ in Figure~\ref{fig:adv-opt}
  outputs a canonical fork for $w$.  
\end{theorem}
That is, for every characteristic string $w$ 
there is a fork $F \Fork w$ so that 
for every prefix $x \PrefixEq w$, $\mu_x(F) = \mu_x(y)$. 
Note that if one's objective is to create a fork 
which contains many early-diverging tine-pairs (that witness large relative margins), 
a canonical fork is the best one can hope for. 
This is why $\Adversary^*$ is called an \emph{optimal} online adversary. 
The proof of Theorem~\ref{thm:opt-adversary-canonical} 
is given in Section~\ref{sec:margin-proof}.

%% file: reach-and-margin.tex
\subsection{Closed forks, reach, and extensions}
\begin{definition}[Closed fork]
  A fork $F$ is \emph{closed} if every leaf is honest. For convenience, we say the trivial fork is closed.
\end{definition}
Closed forks have two nice properties that make them especially useful in reasoning about the view of honest parties.
First, 
all honest observers will select a unique longest tine from this fork 
(since all longest tines in a closed fork are honest, 
honest parties are aware of all previous honest blocks, 
they observe the longest chain rule, and they employ the same consistent tie-breaking rule).  
Second, 
closed forks intuitively capture decision points for the adversary.
The adversary can potentially show many tines to many honest parties, 
but once an honest node has been placed on top of 
a tine, any adversarial blocks beneath it are part of the public record and are visible to all honest parties. 
For these
reasons, we will often find it easier to reason about closed forks than arbitrary forks. 

The next few definitions are the start of a general toolkit for reasoning about an adversary's capacity to build highly diverging paths in forks, based on the underlying characteristic string.

\begin{definition}[Gap, reserve, and reach]\label{def:gap-reserve-reach}
For a closed fork $F \vdash w$ and its unique longest tine $\hat{t}$, we define the \emph{gap} of a tine $t$ to be $\gap(t)=\length(\hat{t})-\length(t)$.
Furthermore, we define the \emph{reserve} of $t$, denoted $\reserve(t)$, to be the number of adversarial indices in $w$ that appear after the terminating vertex of $t$. More precisely, if $v$ is the last vertex of $t$, then
\[
  \reserve(t)=|\{\ i \mid w_i=1 \ and \ i > \ell(v)\}|\,.
  \]
These quantities together define the \emph{reach} of a tine: $
\reach(t)=\reserve(t)-\gap(t)$.
\end{definition}

The notion of reach can be intuitively understood as a measure of
the resources available to our adversary in the settlement
game. Reserve tracks the number of slots in which the adversary has
the right to issue new blocks.  When reserve exceeds gap (or
equivalently, when reach is nonnegative), such a tine could be
extended---using a sequence of dishonest blocks---until it is as long
as the longest tine. Such a tine could be offered to an honest player
who would prefer it over, e.g., the current longest tine in the
fork. In contrast, a tine with negative reach is too far behind to be
directly useful to the adversary at that time.

\begin{definition}[Maximum reach]
For a closed fork $F\vdash w$, we define $\rho(F)$ to be the largest reach attained by any tine of $F$, i.e., 
\[
\rho(F)=\underset{t}\max \ \reach(t)\,.
\]
Note that $\rho(F)$ is never negative (as the longest tine of any fork always has reach at least 0). We overload this notation to denote the maximum reach over all forks for a given characteristic string: 
\[
\rho(w)=\underset{\substack{F\vdash w\\\text{$F$ closed}}}\max\big[\underset{t}\max \ \reach(t)\big]\,.
\]
\end{definition}

Reach of vertices is always non-increasing as we move down a tine. 
That is, if $B_1, B_2, \ldots$ are vertices on the same tine in the root-to-leaf order, then 
$\reach(B_i) \leq \reach(B_{i+1})$. 
The inequality is strict if $B_{i + 1}$ is honest. 
Consequently, the reach of an adversarial tine is no more than 
the reach of the last honest vertex in that tine. 
In any fork, the reach of a maximum-length tine is always non-negative. 
Hence, an honest tine with the maximum length over all honest tines 
will always have a non-negative reach. 
Thanks to the monotonicity of the honest-depth function $\hdepth(\cdot)$, 
if there are multiple honest tines 
having the (same) maximum length among all honest tines, 
they must have the same label. 
Therefore, if $h$ is the last honest slot in $w$ and 
$t$ a maximum-length honest tine with label $h$,  
then $\reach(t) \geq 0$.

\paragraph{Non-negative reach, $\A$-heaviness, and viable adversarial extensions.}
Let $w \in \{\h, \H, \A\}^T$,   
$s \in [T + 1]$, and 
$F \Fork w_1 \ldots w_{s - 1}$ an arbitrary fork. 
Let $B \in F$ be an honest vertex 
and $t$ a maximum-length tine in $F$.
Consider the following statements: 
\begin{enumerate}[label=(\alph*)]
  \item \label{fact-reach-part:viable-adv-ext} $B$ has an adversarial extension viable at the onset of slot $s$.
  \item \label{fact-reach-part:nonneg-reach} $\reach_{F}(B)$ is non-negative.
  \item \label{fact-reach-part:Aheavy} The interval $I = [\ell(B) + 1, s - 1]$ is $\Aheavy$. 
  \item \label{fact-reach-part:conservative} $\length(t) = \#_\h(I) + \#_\H(I) + \length(B)$.     
\end{enumerate}

\begin{fact}\label{fact:fork-structure-reach}
    ~\ref{fact-reach-part:viable-adv-ext} $\Longrightarrow$
    \ref{fact-reach-part:nonneg-reach} $\Longrightarrow$
    \ref{fact-reach-part:Aheavy}.
    In addition, if we assume~\ref{fact-reach-part:conservative}, then 
    ~\ref{fact-reach-part:Aheavy} $\Longrightarrow$ 
    ~\ref{fact-reach-part:nonneg-reach} $\Longrightarrow$
    ~\ref{fact-reach-part:viable-adv-ext}.
\end{fact}
Fact~\ref{fact:fork-structure-reach} can be seen as 
a refinement of Fact~\ref{fact:fork-structure} 
when $F$ is a closed fork. 
\begin{proof}~
  \begin{description}[font=\normalfont\itshape\space]
    \item[\ref{fact-reach-part:viable-adv-ext} implies~\ref{fact-reach-part:nonneg-reach}.]
      An adversarial extension of $B$ 
      contains only adversarial vertices from $I$. 
      If this extension is viable at the onset of slot $s$, 
      $\#_\A(I)$ must be at least $\gap_{F}(B)$.
      Since $\reserve_{F}(B) = \#_\A(I)$, we have 
      $\reach_{F}(B) = \reserve_{F}(B) - \gap_{F}(B) \geq 0$.  
      
    \item[\ref{fact-reach-part:nonneg-reach} implies~\ref{fact-reach-part:Aheavy}.]
      By assumption, $\reach_{F}(B)  = \reserve_{F}(B) - \gap_{F}(B) \geq 0$.               
      $t$ contains at least $\#_\h(I) + \#_\H(I)$ vertices from 
      the interval $I$; 
      hence, $\gap_{F}(B) \geq \#_\h(I) + \#_\H(I)$. 
      Since $\reserve_{F}(B) = \#_\A(I)$, 
      it follows that $\#_\A(I) \geq \#_\h(I) + \#_\H(I)$. 

    \item[\ref{fact-reach-part:conservative} and~\ref{fact-reach-part:Aheavy} implies~\ref{fact-reach-part:nonneg-reach}.]
      Since $I$ is $\Aheavy$, 
      $\reserve_F(B) = \#_\A(I) \geq \#_\h(I) + \#_\H(I)$. 
      However, since~\ref{fact-reach-part:conservative} holds, 
      the latter quantity equals $\length(t) - \length(B) = \gap_F(B)$. 
      It follows that $\reach_F(B) = \reserve_F(B) - \gap_F(B) \geq 0$. 

    \item[\ref{fact-reach-part:conservative} and~\ref{fact-reach-part:nonneg-reach} implies~\ref{fact-reach-part:viable-adv-ext}.]
      $I$ contains at least $\gap_F(B)$ adversarial slots. 
      We can use these slots augment $B$ 
      into an adversarial tine $t'$ 
      of length at least $\length(t)$. 
      Thus $t'$ will be viable at the onset of slot $s$.
  \end{description}  
\end{proof}

Observe that for any characteristic string $x$, 
one can \emph{extend} (i.e., augment) a closed fork prefix $F \Fork x$ 
into a larger closed fork $F' \Fork x0$ so that $F \ForkPrefix F'$. 
A \emph{conservative extension} is a minimal extension in that 
it consumes the least amount of reserve (cf. Definition~\ref{def:gap-reserve-reach}), 
leaving the remaining reserve to be used in future.
Extensions and, in particular, conservative extensions 
play a critical role in the exposition that follows. 

\begin{definition}[Extensions]\label{def:extension}  
  Let $w \in \{\h, \H, \A\}^*$ be a characteristic string 
  and $F$ a closed fork for $w$. 
  Let $F'$ be a closed fork for $wb, b \in \{\h, \H\}$ 
  so that $F \ForkPrefix F'$. 
  We say that \emph{$F'$ is an extension of $F$} if 
  every honest vertex in $F'$ either belongs to $F$ or has label $|w| + 1$. 
  Let $\sigma \in F'$ be an honest vertex with $\ell(\sigma) = |w| + 1$ 
  and let $s$ be the longest honest prefix of $\sigma$. 
  (Necessarily, $s \in F$.)
  We say that \emph{$\sigma$ is an extension of $s$}. 
  The new tine $\sigma$ is a \emph{conservative extension} if 
  $\height(F') = \height(F) + 1$.  
\end{definition} 
Since $F'$ is closed, all longest tines in $F'$ are honest and they have label $|w| + 1$.
Let $\hat{t}$ be the unique longest honest tine in $F'$ 
under the consistent longest-chain selection rule 
in Axiom~\ref{axiom:tie-breaking}.
Now consider a tine $\sigma \in S$. 
Since $\sigma$ is honest, 
it follows that 
$\length(\sigma) 
\geq 1 + \height(F) 
= 1 + \length(s) + \gap_F(s)$ 
where $s \in F$ is the longest honst prefix of $\sigma$.
The root-to-leaf path in $F^\prime$ 
that ends at $\sigma$ 
contains at least $\gap_F(s)$ adversarial vertices $u \in F'$ 
so that $\ell(u) \in [\ell(s) + 1, |w|]$ and 
$u \not \in F$. 
If $\sigma$ is a conservative extension, 
the number of such vertices is exactly $\gap_F(s)$. 

\begin{fact}[Extensions and reach]\label{fact:reach-fork-ext}
  Let $b \in \{\h, \H\}$. 
  Let $F \Fork w$ and $F^\prime \Fork wb$ be closed forks so that 
  $F \ForkPrefix F^\prime$ and 
  $F^\prime$ is obtained from $F$ via one or more extensions 
  $\sigma \in F^\prime, \ell(\sigma) = |w| + 1$.
  Then $\reach_{F^\prime}(t) \leq \reach_F(t) - 1$ for every $t \in F$. 
  If all these extensions are conservative, then 
  $\reach_{F^\prime}(t) = \reach_F(t) - 1$ for every $t \in F$. 
  Furthermore, a conservative extension $\sigma$ satisfies 
  $\reach_{F^\prime}(\sigma) = 0$.
\end{fact}
The above fact follows from the claims below.
\begin{claim}\label{claim:nex}
  Let $b \in \{\h, \H\}$. 
  Consider a closed fork $F\vdash w$ and some closed fork $F'\vdash wb$ such that $F\fprefix F'$. 
  If $t \in F$ then 
  $\reach_{F'}(t)\leq \reach_{F}(t) - 1$. 
  The inequality becomes and equality 
  if $F'$ is obtained via 
  conservative extensions from $F$.
\end{claim}
\begin{proof}
  We know that $\reach_{F'}(t)=\reserve_{F'}(t)-\gap_{F'}(t).$ From $F$ to $F'$, the length of the longest tine increases by at least one, and the length of $t$ does not change. 
  It follows that $\gap_{F'}(t) \geq \gap_{F}(t) + 1$. 
  The inequality becomes an equality 
  if $F'$ is obtained from $F$ via only conservative extensions. 
  The reserve of $t$ does not change, because there are no new $\A$s in the characteristic string. Therefore, 
  $
    \reach_{F'}(t)
    =\reserve_{F'}(t)-\gap_{F'}(t)
    \leq \reserve_{F}(t)-\gap_{F}(t) - 1
    =\reach_{F}(t) - 1
  $. 
\end{proof}
\begin{claim}\label{claim:ex}
  Conservative extensions have reach zero.
\end{claim}
\begin{proof}
  Let $b \in \{\h, \H\}$. 
  Consider closed forks $F\vdash w, F'\vdash wb$ 
  such that $F\fprefix F'$. 
  Let $t \in F'$ be a conservative extension. 
  This means $t$ is honest, $\ell(t) = |w| + 1$, 
  and 
  $t$ is a longest tine in $F'$. 
  The last statement implies $\gap_{F'}(t)=0$. 
  Since $\reserve_{F'}(t)=0$, it follows that 
  $\reach_{F'}(t)=\reserve_{F'}(t)-\gap_{F'}(t) = 0$. 
\end{proof}

\subsection{Relative margin}
\begin{definition}[The $\sim_x$ relations]
  For two tines $t_1$ and $t_2$ of a fork $F$, we write $t_1 \sim t_2$
  when $t_1$ and $t_2$ share an edge; otherwise we write
  $t_1 \nsim t_2$. We generalize this equivalence relation to reflect
  whether tines share an edge over a particular suffix of $w$: for
  $w = xy$ we define $t_1 \sim_x t_2$ if $t_1$ and $t_2$ share an edge
  that terminates at some node labeled with an index in $y$;
  otherwise, we write $t_1 \nsim_x t_2$ (observe that in this case the
  paths share no vertex labeled by a slot associated with $y$).  We
  sometimes call such pairs of tines \emph{disjoint} (or, if
  $t_1 \nsim_x t_2$ for a string $w = xy$, \emph{disjoint over
    $y$}). Note that $\sim$ and $\sim_\varepsilon$ are the same
  relation.
\end{definition}

\begin{definition}[Margin]\label{def:margin}
The \emph{margin} of a fork $F\vdash w$, denoted $\mu(F)$, is defined as 
\begin{equation}\label{eq:margin-absolute}
\mu(F)=\underset{t_1\nsim t_2}\max \bigl(\min\{\reach(t_1),\reach(t_2)\}\bigr)\,,
\end{equation}
where this maximum is extended over all pairs of disjoint tines of
$F$; thus margin reflects the ``second best'' reach obtained over all
disjoint tines. In order to study splits in the chain over particular portions of a
string, we generalize this to define a ``relative'' notion of margin:
If $w = xy$ for two strings $x$ and $y$ and, as above, $F \vdash w$,
we define
\[
  \mu_x(F)=\underset{t_1\nsim_x t_2}\max \bigl(\min\{\reach(t_1),\reach(t_2)\}\bigr)\,.
\]
Note that $\mu_\varepsilon(F) = \mu(F)$.

For convenience, we once again overload this notation to denote the
margin of a string. $\mu(w)$ refers to the maximum value of $\mu(F)$
over all possible closed forks $F$ for a characteristic string $w$:
\[
\mu(w)=\underset{\substack{F\vdash w,\\ \text{$F$ closed}}}\max \, \mu(F)\,.
\]
Likewise, if $w = xy$ for two strings $x$ and $y$ we define
\[
\mu_x(y)=\underset{\substack{F\vdash w,\\ \text{$F$ closed}}} \max \, \mu_x(F)\,.
\]
\end{definition}
Note that, at least informally, 
disjoint tines with large reach are of natural
interest to an adversary who wants to build an $x$-balanced fork, 
since such a fork contains two (partially disjoint) long tines.
It is easy to see that 
if $w = xx'y$ and 
$\mu_{xx'}(y)$ is negative then $\mu_x(x'y)$ is negative as well.


The theorem below shows how to recursively compute $\mu_x(y)$ 
for a given decomposition $w = xy$.

\begin{theorem}\label{thm:relative-margin}
  Let $\varepsilon$ be the empty string 
  and $b \in \{\h, \H\}$. 
  Then $\rho(\varepsilon) = 0$ 
  and, for all nonempty strings $w \in \{\h, \H, \A\}^*$ 
  \begin{equation}
    \rho(w\A) = \rho(w) + 1\,, \qquad\text{and}\qquad
    \rho(wb) = \begin{cases} 0 & \text{if $\rho(w) = 0$,}\\
      \rho(w)-1 & \text{otherwise.}
    \end{cases}
    \label{eq:rho-recursive}
  \end{equation}

  Furthermore, for any strings $x, y \in\{\h, \H, \A\}\text{\emph{*}}$,
  $\mu_x(\varepsilon) =\rho(x)$, 
  \begin{equation}
    \mu_x(y\A)= \mu_x(y)+1\,,\qquad\text{and}\qquad
    \mu_x(yb)= \begin{cases}
      0 & \text{if $\rho(xy) > \mu_x(y)=0$}\,, \\
      0 & \text{if $\rho(xy) = \mu_x(y) = 0$ and $b = \H$}\,, \\
      \mu_x(y)-1 & \text{otherwise.}
    \end{cases}
    \label{eq:mu-relative-recursive}
  \end{equation}

\end{theorem}
The proof of Theorem~\ref{thm:relative-margin} is given in Section~\ref{sec:margin-proof}. 
Let $w$ be a characteristic string and 
let $m, k \in \NN$ so that $m + k \leq |w|$. 
Let $x \Prefix w, |x| = m-1$ and $xy \PrefixEq w, |xy| \geq m + k$.
If the symbols in $w$ are independent and identically distributed, 
the recursive formulation in~\eqref{eq:mu-relative-recursive} implies an algorithm --- which takes time and space $O(|w|^3)$ --- 
for computing the probability that $\mu_x(y) \geq 0$. 
But this is exactly the probability that slot $m$ is not $k$-settled, 
according to~\eqref{eq:settlement-uvp} 
and Lemma~\ref{lemma:uvp-margin} below. 
In Section~\ref{sec:exact-prob}, 
we describe this algorithm in more detail and 
compile some explicit values for this probability.

%% file: exact-probabilities.tex
Let $m, k \in \NN$, $\epsilon \in (0,1]$, and $p_\h \in (0, (1+\epsilon)/2]$.
Let $T = m + k, \alpha = (1 - \epsilon)/2$, and $p_\H = 1 - \alpha - p_\h$. 
Let $w \in \{\h, \H, \A\}^T$ such that 
the symbols $w_i, i \in [T]$ are i.i.d.\, with 
$\Pr[w_i = \A] = \alpha, \Pr[w_i = \h] = p_\h$, and $\Pr[w_i = \H] = p_\H$. 
Write $w$ as $w = xy$ where $|x| = m, |y| = k$.
The recursive definition of relative margin 
(cf. Theorem~\ref{thm:relative-margin}) 
implies an algorithm for computing the probability
$\Pr[\mu_x(y) \geq 0]$ 
in $O(T^3)$ time and space. 

In typical circumstances, however, it is more interesting to establish an
explicit upper bound on $\Pr[\mu_x(y) \geq 0]$ where
$|x| \rightarrow \infty$; this corresponds to the case where the
distribution of the initial reach $\rho(x)$ is the dominant distribution
$\mathcal{X}_\infty$ in 
~\eqref{eq:stationary}. 
Due to dominance, $\mathcal{X}_\infty(m)$ serves as an
upper bound on $\rho(x)$ for any finite $m = |x|$. 
For this purpose, one can implicitly
maintain a sequence of matrices $M_t, t = 0, 1, 2, \ldots, k$
such that $M_0(r, r) = \mathcal{X}_\infty(r)$ for all $0 \leq r \leq 2k$ and
the invariant
\[
  M_t(r, s) = \Pr_{y: |y| = t}[\rho(xy) = r \text{ and }
  \mu_x(y) = s ]
\]
is satisfied for every integer $t \in [1, k]$,
$r \in [0, 2k]$, and $s \in [-2k, 2k]$. 
Here, $M(i,j)$ denotes the entry at the $i$th row and $j$th column of a matrix $M$.
Observe that $M_t(r,s)$ can be computed solely 
from the relevant neighboring cells of $M_{t-1}$, that is, 
from the values $M_{t-1}(r\pm 1, s \pm 1)$. 
Of course, only the transitions approved by~\eqref{eq:mu-relative-recursive} should be considered.

Finally, one can compute $\Pr[\mu_x(y) \geq 0]$ by summing $M_k(r,s)$ for
$r, s \geq 0$. 
This is precisely the probability that, 
given a characteristic string $xy$ where $|x| \rightarrow \infty$, 
the slot $|x| + 1$ incurs a $|y|$-settlement violation. 
Table~\ref{table:exact-probs} (on page~\pageref{table:exact-probs}) contains these probabilities for various values of $\alpha, |y|$, and $p_\h$. 


A \texttt{C++} implementation of the above algorithm is publicly available 
at~\href{https://github.com/saad0105050/multihonest-code}{https://github.com/saad0105050/multihonest-code}~\cite{PrForkableMultihonestCode}.

\newcommand{\EndRow}{\cline{2-8} \multicolumn{1}{|c||}{} &}

\begin{table}[t]
	\centering
	\caption{
    Exact probabilities of $k$-settlement violations 
    where the symbols $\h, \H, \A$ are independent and identically distributed as $\Pr[\A] = \alpha \in (0, 0.5)$ and $\Pr[\H] = 1 - \alpha - \Pr[\h]$.    
	} 
	\label{table:exact-probs}

	\begin{tabular}{|c||l||l|l|l|l|l|l|}
    \hline
    \multirow{2}{*}{$\dfrac{\Pr[\h]}{1-\alpha}$} & 
  	\multicolumn{1}{|c||}{\multirow{2}{*}{$k$}} & 
    \multicolumn{6}{c|}{$\alpha$} \\ 
    \cline{3-8} 
    \multicolumn{1}{|c||}{} &
  	\multicolumn{1}{|c||}{} &
    0.01 & 0.10 & 0.20 & 0.30 & 0.40 & 0.49\\ 
  	\hhline{|=#=#=|=|=|=|=|=|}
    \multicolumn{1}{|c||}{\multirow{5}{*}{$1.0$}}&
    100 & 5.70E-054 & 5.10E-018 & 2.28E-008 & 8.00E-004 & 1.37E-001 & 9.05E-001 \\ \EndRow
    200 & 1.64E-106 & 9.82E-035 & 1.61E-015 & 1.60E-006 & 3.36E-002 & 8.73E-001 \\ \EndRow
    300 & 4.70E-159 & 1.89E-051 & 1.14E-022 & 3.25E-009 & 8.52E-003 & 8.50E-001 \\ \EndRow
    400 & 1.35E-211 & 3.64E-068 & 8.02E-030 & 6.59E-012 & 2.18E-003 & 8.29E-001 \\ \EndRow
    500 & 1.02E-264 & 3.90E-085 & 4.00E-037 & 1.10E-014 & 5.16E-004 & 8.05E-001 \\

    \hhline{|=#=#=|=|=|=|=|=|}
    \multicolumn{1}{|c||}{\multirow{5}{*}{$0.9$}}&
    100 & 9.75E-052 & 1.24E-017 & 3.24E-008 & 9.27E-004 & 1.44E-001 & 9.08E-001   \\ \EndRow
    200 & 3.04E-102 & 4.95E-034 & 2.96E-015 & 2.03E-006 & 3.60E-002 & 8.77E-001   \\ \EndRow
    300 & 9.46E-153 & 1.98E-050 & 2.71E-022 & 4.50E-009 & 9.30E-003 & 8.53E-001   \\ \EndRow
    400 & 2.95E-203 & 7.91E-067 & 2.48E-029 & 9.96E-012 & 2.43E-003 & 8.33E-001   \\ \EndRow
    500 & 1.83E-254 & 1.63E-083 & 1.54E-036 & 1.78E-014 & 5.80E-004 & 8.08E-001   \\

    \hhline{|=#=#=|=|=|=|=|=|}
    \multicolumn{1}{|c||}{\multirow{5}{*}{$0.8$}}&
    100 & 6.16E-048 & 4.13E-017 & 5.10E-008 & 1.11E-003 & 1.53E-001 & 9.11E-001 \\ \EndRow
    200 & 7.58E-095 & 4.61E-033 & 6.58E-015 & 2.73E-006 & 3.91E-002 & 8.81E-001 \\ \EndRow
    300 & 9.32E-142 & 5.14E-049 & 8.48E-022 & 6.78E-009 & 1.04E-002 & 8.57E-001 \\ \EndRow
    400 & 1.15E-188 & 5.74E-065 & 1.09E-028 & 1.68E-011 & 2.77E-003 & 8.38E-001 \\ \EndRow
    500 & 1.94E-236 & 3.02E-081 & 9.16E-036 & 3.28E-014 & 6.70E-004 & 8.12E-001 \\

    \hhline{|=#=#=|=|=|=|=|=|}
    \multicolumn{1}{|c||}{\multirow{5}{*}{$0.5$}}&
    100 & 4.80E-028 & 6.53E-014 & 6.21E-007 & 2.80E-003 & 1.99E-001 & 9.26E-001 \\ \EndRow
    200 & 2.46E-055 & 6.31E-027 & 6.40E-013 & 1.31E-005 & 5.86E-002 & 8.98E-001 \\ \EndRow
    300 & 1.26E-082 & 6.10E-040 & 6.60E-019 & 6.19E-008 & 1.76E-002 & 8.77E-001 \\ \EndRow
    400 & 6.46E-110 & 5.90E-053 & 6.81E-025 & 2.92E-010 & 5.33E-003 & 8.59E-001 \\ \EndRow
    500 & 1.28E-138 & 1.75E-066 & 3.65E-031 & 9.61E-013 & 1.39E-003 & 8.31E-001 \\

    \hhline{|=#=#=|=|=|=|=|=|}
    \multicolumn{1}{|c||}{\multirow{5}{*}{$0.25$}}&
    100 & 1.22E-012 & 3.13E-008 & 8.94E-005 & 1.65E-002 & 3.17E-001 & 9.48E-001 \\ \EndRow
    200 & 1.51E-024 & 1.06E-015 & 9.36E-009 & 3.36E-004 & 1.25E-001 & 9.27E-001 \\ \EndRow
    300 & 1.86E-036 & 3.62E-023 & 9.80E-013 & 6.86E-006 & 4.94E-002 & 9.10E-001 \\ \EndRow
    400 & 2.30E-048 & 1.23E-030 & 1.03E-016 & 1.40E-007 & 1.96E-002 & 8.96E-001 \\ \EndRow
    500 & 5.06E-062 & 7.72E-039 & 4.06E-021 & 1.66E-009 & 6.20E-003 & 8.65E-001 \\

    \hhline{|=#=#=|=|=|=|=|=|}
    \multicolumn{1}{|c||}{\multirow{5}{*}{$0.01$}}&
    100 & 3.77E-001 & 4.91E-001 & 6.38E-001 & 7.95E-001 & 9.31E-001 & 9.97E-001 \\ \EndRow
    200 & 1.42E-001 & 2.41E-001 & 4.08E-001 & 6.34E-001 & 8.72E-001 & 9.95E-001 \\ \EndRow
    300 & 5.37E-002 & 1.18E-001 & 2.61E-001 & 5.06E-001 & 8.17E-001 & 9.94E-001 \\ \EndRow
    400 & 2.03E-002 & 5.81E-002 & 1.67E-001 & 4.04E-001 & 7.66E-001 & 9.92E-001 \\ \EndRow
    500 & 7.89E-005 & 3.23E-003 & 2.71E-002 & 1.40E-001 & 4.83E-001 & 9.54E-001 \\

    \hline
  \end{tabular}

   

\end{table}

%% file: margin-proof.tex
The proof of Theorem~\ref{thm:relative-margin} is presented in two parts. 
Let $w \in \{\h, \H, \A\}^*$.
First, for a given decomposition $w = xy$, 
we prove an upper bound on $\mu_x(y)$. 
Next, 
considering the fork $F \Fork w$ built by the strategy $Adversary^*$ 
(see Figure~\ref{fig:adv-opt}), 
we show that 
for every decomposition $w = xy$, 
$\mu_x(F)$ is at least as large as the upper bound proven in the first part; 
thus $F$ is canonical. 

As a warm-up, we start with the following claim.

\begin{claim}\label{claim:rho-mu-A}
  $\rho(\varepsilon) = 0$. 
  For any $x, y \in \{\h, \H, \A\}^*$, 
  $\mu_x(\varepsilon) = \rho(x)$, 
  $\rho(xy\A) = \rho(xy) + 1$, 
  and $\mu_x(y\A) = \mu_x(y) + 1$.
\end{claim}
\begin{proof}
  The only possible fork for the empty string $\varepsilon$ 
  contains a single honest vertex 
  with reserve and gap both zero; 
  hence $\rho(\varepsilon) = 0$.

  Let $F$ be a closed fork for the characteristic string $xy$. 
  Let $t_\rho, t_x \in F$ be the two tines that witness $\mu_x(F)$, 
  i.e., $\reach(t_\rho) = \rho(F), \reach_F(t_x) = \mu_x(F)$, 
  and $t_\rho, t_x$ are disjoint over $y$. 

  In the base case, where $y=\varepsilon$, 
  observe that any two tines of $F$ 
  are disjoint over $y$. 
  Moreover, a single tine $t \in F$ 
  is disjoint with itself over the empty suffix $\varepsilon$. 
  Therefore, the relative margin $\mu_x(\varepsilon)$ must be at least $\rho(x)$. 
  As $\mu_x(F)$ can be no more than $\rho(x)$, it follows that 
  $\mu_x(\varepsilon) = \rho(x)$.

  Now consider a pair of closed forks $F\vdash xy$ and $F'\vdash xy\A$ 
  such that $F \fprefix F'$ and $x,y\in\{\h, \H, \A\}^*$. 
  We must have $F' = F$ since $F'$ is closed. 
  In addition, for any tine $t \in F$, 
  $\reach_{F'}(t) = \reach_F(t) + 1$ 
  since the reserve has increased by one 
  but the gap is unchanged 
  (as no new tine is added). Therefore, 
  $\rho(xy\A) = \rho(xy) + 1$ and 
  $\mu_x(y\A) = \mu_x(y)+1$.
\end{proof}

\subsection{An upper bound on relative margin}

\begin{proposition}\label{prop:mu-upperbound}
  Let $w,x, y \in \{\h, \H, \A\}^*$ 
  and $b \in \{\h, \H\}$, 
  Then 
  \begin{equation}
    \rho(xyb) \leq \begin{cases}
      0 & \text{if $\rho(xy) = 0$}\,, \\
      \rho(xy) - 1 & \text{otherwise.}
    \end{cases}
    \label{eq:rho-upperbound}
  \end{equation}
  Furthermore,
  \begin{equation}
    \mu_x(yb) \leq \begin{cases}
      0 & \text{if $\rho(xy) > \mu_x(y)=0$}\,, \\
      0 & \text{if $\rho(xy) = \mu_x(y) = 0$ and $b = \H$}\,, \\
      \mu_x(y)-1 & \text{otherwise.}
    \end{cases}
    \label{eq:mu-upperbound}
  \end{equation}
\end{proposition}
\begin{proof}
  Suppose $F'\vdash xyb$ is a closed fork such that 
  $\rho(xyb)=\rho(F')$ and $\mu_x(yb)=\mu_x(F')$. 
  Let $t_\rho, t_x \in F'$ be a pair of $y$-disjoint tines such that $\reach_{F'}(t_\rho)=\rho(F')$ and $\reach_{F'}(t_x)=\mu_x(F')$. 
  (If there are multiple candidates for $t_\rho$ or $t_x$, 
  select the one with the smallest $\leq_\pi$ rank.)
  Let $F\vdash xy$ be the unique closed fork such that $F\fprefix F'$.  
  Note that while $F'$ is obtained from one or more extensions 
  of $F$-tines, 
  these extensions are not necessarily conservative. 
  Recall that $\reach_{F'}(t) \leq 0$ for any tine $t \in F', \ell(t) = |xy| + 1$.

  \paragraph{Proving ~\eqref{eq:rho-upperbound}.} 
  Let $A$ be the set of all $F'$-tines with label $|xy| + 1$.
  Let $\sigma \in A$ be the first tine in the $\leq_\pi$ ordering so that $\reach(\sigma) = \max_{t \in A}\{\reach_{F'}(t)\}$.
  By Fact~\ref{fact:reach-fork-ext}, 
  $\reach_{F'}(\sigma) \leq 0$ and, 
  in addition, for any $t \in F$, 
  $\reach_{F'}(t) \leq \reach_F(t) - 1$.   
  Let $\hat{t}$ be the maximum-reach tine in $F$ 
  with the smallest $\leq_\pi$ rank.

  If $\rho(F) = 0$ then 
  $\reach_{F'}(t) < 0$ for all $t \in F$. 
  Hence $t_\rho = \sigma$ and, consequently, 
  $\rho(xyb) \leq 0$. 
  If $\rho(F) \geq 2$ then $t_\rho \in F$
  and, therefore, 
  $\rho(xyb) = \reach_{F'}(t_\rho) \leq \rho(F) - 1 \leq \rho(xy) - 1$. 
  If $\rho(F) = 1$ and $t_\rho \in F$ then, 
  as before, 
  $\rho(xyb) = \reach_{F'}(\hat{t}) = \reach_{F}(\hat{t}) - 1 = \rho(F) - 1 \leq \rho(xy) - 1$.
  If $\rho(F) = 1$ and $t_\rho \not \in F$ then, as we have seen before, 
  $\rho(xyb) = \reach_{F'}(\sigma) \leq 0 = \rho(F) - 1\leq \rho(xy) - 1$.
  Thus we have proved~\eqref{eq:rho-upperbound}.

  \paragraph{Proving ~\eqref{eq:mu-upperbound}.} 
  If $\ell(t_\rho) = |xy| + 1$ then we are done: 
  by our preceding argument, $\reach_{F'}(t_\rho) \leq 0$. 
  On the other hand, 
  Note that $t_\rho \not \in F$ since, by Fact~\ref{fact:reach-fork-ext}, reach of any $F$ tine can only decrease
  $t_\rho$ must have been an extension of a maximum-reach $F$-tine.

  \paragraph{Case 1: $\rho(xy)>0$ and $\mu_x(y)=0$.} 
    We wish to show that $\mu_x(yb) \leq 0$.
    Suppose (toward a contradiction) that $\mu_x(yb) > 0$. 
    Then neither $t_\rho$ nor $t_x$ is a conservative extension because, as we proved in Claim ~\ref{claim:ex}, conservative extensions have reach zero. This means that $t_\rho$ and $t_x$ existed in $F$, and their $F$-reach was strictly greater than their $F'$-reach (by Claim ~\ref{claim:nex}). 
    Because $t_\rho$ and $t_x$ 
    are 
    disjoint over $y0$, they must also be disjoint over $y$; therefore, $\mu_x(F)$ must be at least $\min(\reach_F(t_\rho),\reach_F(t_x))$. 
    It follows that 
    $0 
    = \mu_x(y) 
    \geq \min(\reach_F(t_\rho),\reach_F(t_x))
    > \min(\reach_{F'}(t_\rho),\reach_{F'}(t_x))
    = \mu_x(F') = \mu_x(yb)
    $. 
    The last term is strictly positive by assumption and hence, a contradiction ensues.

  \paragraph{Case 2: $\rho(xy)=0$.}
    We wish to show that 
    \begin{enumerate*}[label=(\textit{\roman*})]
      \item $\mu_x(yb) \leq 0$ if $b = \H$ and $\mu_x(y) = 0$, and 
      \item $\mu_x(yb) \leq \mu_x(y) - 1$ otherwise.
    \end{enumerate*}
    First, we claim that $t_\rho$ must arise from an extension. 
    Suppose, toward a contradiction, that $t_\rho$ is not an extension, 
    i.e., $t_\rho \in F$. 
    The fact that $t_\rho$ achieves the maximum reach in $F'$ 
    implies that 
    $t_\rho$ has a non-negative reach 
    since the longest honest tine always achieves reach zero. 
    Furthermore, 
    Claim ~\ref{claim:nex} states that 
    all $F$-tines see their reach decrease. 
    Therefore, $t_\rho \in F$ must have had a strictly positive reach. 
    But this contradicts the central assumption of the case, i.e., 
    that $\rho(xy)=0$. 
    Therefore, we conclude that $t_\rho \in F' \setminus F$.

    Let $s \in F$ be the tine-prefix of $t_\rho \in F'$ so that 
    $t_\rho$ is an extension of $s$. 
    Observe that $\reach_F(s)$ must be non-negative since 
    otherwise, $s$ could not have been extended. 
    In fact, our assumption $\rho(xy)=0$ implies that 
    $\reach_F(s) = 0$. 
    In addition, since $t_x$ and $t_\rho$ are disjoint over $yb$, 
    so are $t_x$ and $s$. 
    \begin{description}[font=\normalfont\itshape\space]
      \item[If $b = \h$,] 
      $t_\rho$ is the only extension in $F'$ and hence 
      $t_x$ must be in $F$. 
      Consequently, 
      $\min( \reach_F(s),\reach_F(t_x) ) \leq \mu_x(y)$. 
      Because $\reach_F(s)=0$ and $\reach_F(t_x) \leq \rho(xy)=0$, it follows that $\reach_F(t_x) \leq \mu_x(y)$. 
      Finally, since $t_x \in F$, 
      Claim ~\ref{claim:nex} tells us that 
      $\reach_{F'}(t_x) < \reach_F(t_x)$. 
      Taken together, these two inequalities show that 
      $\mu_x(yb) = \reach_{F'}(t_x) < \reach_F(t_x) \leq \mu_x(y)$. 
      The last inequality follows since $s$ and $t_x$ are disjoint over $y$ and $\reach_F(s) = 0 = \rho(xy)$. 
      We conclude that $\mu_x(yb) \leq \mu_x(y) - 1$.

      \item[If $b = \H$ and $\mu_x(y) < 0$,] 
      we claim that $t_x \in F$. 
      To see why, note that as $t_x$ is $yb$-disjoint with $t_\rho$, 
      it must extend some $F$-tine $t$ that is $y$-disjoint with $t_\rho$. 
      However, as $\mu_x(y) < 0$, $t$ must have negative reach and hence cannot be extended into $t_x$; this is a contradiction. 
      Therefore, $t_x \in F$ and we can apply the argument in the ``$b = \h$'' case above 
      to conclude that $\mu_x(yb) \leq \mu_x(y) - 1$.

      \item[If $b = \H$ and $\mu_x(y) = 0$,] 
      then there are two alternatives depending on 
      whether $t_x$ is an extension. 
      If $t_x$ is not an extension, we can apply the argument in the ``$b = \h$'' case above and conclude that 
      $\mu_x(yb) \leq \mu_x(y) - 1 = -1$.
      On the other hand, if $t_x \not\in F$, 
      both $t_x$ and $t_\rho$ are extensions and, 
      by Fact~\ref{fact:reach-fork-ext}, 
      $\max(\reach_{F'}(t_x), \reach_{F'}(t_\rho) ) \leq 0$. 
      In addition, Fact~\ref{fact:reach-fork-ext} states that for all $t \in F$, 
      $\reach_{F'}(t) < \reach_{F}(t) \leq \rho(xy) = 0$. 
      We conclude that $\mu_x(yb) \leq 0$.

    \end{description}

  \paragraph{Case 3: $\rho(xy)>0$ and $\mu_x(y)\neq0$.}
    We wish to show that $\mu_x(yb) \leq \mu_x(y) - 1$ 
    or, equivalently, that $\mu_x(yb) < \mu_x(y)$. 
    We will break this case into two sub-cases. 
    \begin{description}[font=\normalfont\itshape\space]
      \item[If both $t_\rho, t_x \in F$,] 
      then $\mu_x(yb) = \reach_{F'}(t_x) < \reach_{F}(t_x) \leq \mu_x(y)$. 
      Here, the first inequality follows from Fact~\ref{fact:reach-fork-ext} 
      and the second inequality follows from the fact that 
      $t_x, t_\rho$ is $y$-disjoint and 
      $\reach(t_x)$ is at most $\reach(t_\rho)$ by design.

      \item[Otherwise,] 
      at least one of $t_x, t_\rho$ arose from an extension. 
      Since $\reach_{F'}(t_x) \leq \reach_{F'}(t_\rho)$ by design, 
      it follows that $\reach_{F'}(t_x) \leq 0$ 
      as the reach of an extension is at most zero.
      If $\mu_x(y) > 0$ then we are done: $\mu_x(yb) \leq 0 < \mu_x(y)$.  
      On the other hand, suppose $\mu_x(y) < 0$. 
      Recall the tine $s$ mentioned before.
      As $t_x$ is $y$-disjoint with $s$ and 
      $\mu_x(y)$ is negative by assumption, 
      $\reach_F(t_x)$ is at most $\mu_x(y)$. 
      We conclude that 
      $
      \mu_x(yb) = \reach_{F'}(t_x) 
      < \reach_F(t_x) 
      \leq \mu_x(y)
      $ 
      where the inequality follows from Fact~\ref{fact:reach-fork-ext}.

    \end{description}
\end{proof}

\subsection{\texorpdfstring{$\Adversary^*$}{The optimal adversary} simultaneously maximizes all relative margins}

\begin{proposition}\label{prop:mu-lowerbound}
  Let $w \in \{\h, \H, \A\}^*$ 
  and $b \in \{\h, \H, \A\}$. 
  Assume that Theorem~\ref{thm:opt-adversary-canonical} 
  holds for characteristic strings of length $|w|$.
  Let $F'$ be the fork built by $\Adversary^*$ 
  for the characteristic string $wb$. 
  Then 
  \begin{equation}
    \rho(F') \geq \begin{cases}
     \rho(xy) + 1 & \text{if $b = \A$} \,, \\
     0 & \text{if $b \in \{\h, \H\}$ and $\rho(xy) = 0$}\,, \\
     \rho(xy) - 1 & \text{otherwise}\,.
    \end{cases}
    \label{eq:rho-lowerbound}    
  \end{equation}
  Furthermore, for any decomposition $w = xy, |y| \geq 0$, 
  \begin{equation}
    \mu_x(F') \geq \begin{cases}
      \mu_x(y) + 1 & \text{if $b = \A$}\,, \\
      0 & \text{if $b \in \{\h, \H\}$ and $\rho(xy) > \mu_x(y)=0$}\,, \\
      0 & \text{if $b = \H$ and $\rho(xy) = \mu_x(y) = 0$}\,, \\
      \mu_x(y)-1 & \text{otherwise.}
    \end{cases}
    \label{eq:mu-lowerbound}
  \end{equation}
  
\end{proposition}
\begin{proof}  
  Let $w' = wb$.
  Let $F$ and $F'$ be the forks built by $\Adversary^*$ 
  for the characteristic string $w$ and $wb$, respectively, 
  so that $F \fprefix F'$.
  By assumption, $F$ is a canonical fork for $w$; 
  this means $\rho(F) = \rho(w)$ and
  for all $x \Prefix w$, $\mu_x(F) = \mu_x(y)$. 
  It will be helpful for the reader to recall Fact~\ref{fact:reach-fork-ext} before proceeding.

  \paragraph{Proving~\eqref{eq:rho-lowerbound}.} 
  We wish to show that 
  $\rho(F')$ satisfies~\eqref{eq:rho-lowerbound}. 
  If $b = \A$ then, by construction,
  $F' = F$. 
  The symbol $b = \A$ increases the reserve 
  of every tine by one. 
  Thus  
  $
  \rho(F') 
  = \rho(F) + 1
  = \rho(xy)  + 1
  $. 
  Now suppose $b \in \{\h, \H\}$. 
  Since all tines $\sigma \in F'$ with label $|xy| + 1$ 
  are conservative extensions, $\reach_{F'}(\sigma) = 0$ 
  and the $F'$-reach of all $F$-tines decreases by one. 
  Let $t$ be a maximum-reach tine in $F$; 
  since $F$ is canonical, $\reach_F(t) = \rho(F) = \rho(xy)$.
  Therefore, 
  $\rho(F') \geq \reach_{F'}(t) = \reach_{F}(t) - 1 = \rho(xy) - 1$. 
  If $\rho(F) = 0$ then this inequality can be tightened, as follows. 
  As all $F$-tines have negative $F'$-reach, 
  any maximum-reach $F'$-tine 
  must be one of the extensions; 
  it follows that $\rho(F') = 0$. 
  Thus we have proved~\eqref{eq:rho-lowerbound}.

  \paragraph{Proving~\eqref{eq:mu-lowerbound}.} 
  Let $w = xy$ be an arbitrary decomposition; 
  this $x$ remains fixed for the remainder of the proof. 
  (Note that $\Adversary^*$ is unaware of this decomposition.) 

  Let $\tau_x, \tau_{\rho x} \in F'$ be two $yb$-disjoint tines 
  so that 
  $\reach_{F'}(\tau_{\rho x}) = \rho(F')$, 
  $\reach_{F'}(\tau_x) = \mu_x(F')$, 
  and, of all $yb$-disjoint tine pairs in $F'$
  that attain this requirement, 
  these two tines diverge the earliest. 
  We say that the tines $\tau_x, \tau_{\rho x}$ \emph{witness} $\mu_x(F')$.

  Designate the witness tines $t_x, t_{\rho x} \in F$ 
  in the same way as we have designated $\tau_x, \tau_{\rho x} \in F'$; 
  specifically, 
  $w, y$, and $F$ would substitute $w', yb$, and $F'$ 
  in the recipe above. 
  By assumption, $F$ is a canonical fork for $xy$. 
  Therefore, 
  $\rho(F) = \reach_F(t_{\rho x}) = \rho(xy)$, 
  $t_x$ is $y$-disjoint with $t_{\rho x}$, 
  and 
  $\mu_x(F) = \reach_F(t_x) = \mu_x(y)$. 
  We wish to show that 
  $\mu_x(F')$ satisfies~\eqref{eq:mu-lowerbound}. 

  If $b = \A$ then, by construction,
  $F' = F$ and, therefore,
  $t_x$ and $t_{\rho x}$ are $yb$-disjoint in $F'$.
  Note that the $F'$-reach of every $F$-tine is one plus its $F$-reach. 
  Therefore, 
  $
  \mu_x(F') 
  \geq \min(\reach_{F'}(t_{\rho x}), \reach_{F'}(t_x))
  = \reach_{F'}(t_x)
  = \reach_F(t_x) + 1
  = \mu_x(y) + 1
  $.

  If $b \in \{\h, \H\}$, 
  all tines in $F'$ with label $|w| + 1$ arise from conservative extensions. 
  Since the tines $t_x, t_{\rho x}$ are $yb$-disjoint in $F'$, 
  it follows that $\mu_x(F') \geq \min(\reach_{F'}(t_x), \reach_{F'}(t_{\rho x})) \geq \reach_{F'}(t_x) = \reach_{F}(t_x) - 1 = \mu_x(y) - 1$. 
  Here, the first inequality follows from the definition of relative margin and 
  the second one from the fact that $\reach(t_x) \leq \reach(t_{\rho x})$ by assumption. 
  The first equality follows from Fact~\ref{fact:reach-fork-ext} and 
  the second one follows from our assumption that the tines $t_{\rho x}, t_x\in F$ witness $\mu_x(F) = \mu_x(y)$. 

  However, we can tighten the above inequality when $\mu_x(y)$ is zero, as follows. 
  Recall the sets $Z, S, R$, 
  the zero-reach tine $z_1$, 
  and the maximum-reach tine $r_1$ 
  from Figure~\ref{fig:adv-opt}. 
  Also recall that $z_1$, of all zero-reach tines, 
  diverges earliest from any maximum-reach tine.
  As $\reach_F(z_1) = \mu_x(F) = \mu_x(y) = 0$, 
  it follows that 
  $z_1$ and $r_1$ must be $y$-disjoint.
  Let $\sigma_1 \in F'$ be the conservative extension of $z_1$.

  \begin{description}[font=\normalfont\itshape\space]
    \item[If $\rho(xy) \geq 1$ and $\mu_x(y) = 0$]
      then $\sigma_1$ is the only new extension in $F'$ 
      and it has reach zero in $F'$.
      Note that 
      $\reach_{F'}(r_1) = \reach_F(r_1) - 1 = \rho(F) - 1 \geq 0$ 
      since $\rho(F) = \rho(xy) \geq 1$ by assumption. 
      It follows that $\mu_x(F') \geq \min(\reach_{F'}(\sigma_1), \reach_{F'}(r_1) ) \geq \reach_{F'}(\sigma_1) = 0$.

    \item[If $\rho(xy) = 0$ and $\mu_x(y) = 0$] 
      then $Z = R$ and $|Z| \geq 2$. 
      If $b = \h$, 
      $\sigma_1$
      is the only tine in $F'$ with the maximum reach, zero. 
      Note that 
      $\reach_{F'}(r_1) = \reach_F(r_1) - 1 = \rho(F) - 1 \geq -1$. 
      Since $\sigma_1$ and $r_1$ are $yb$-disjoint, 
      it follows that 
      $\mu_x(F') \geq \min(\reach_{F'}(\sigma_1), \reach_{F'}(r_1) ) 
      \geq \reach_{F'}(r_1) \geq = -1$.

      On the other hand, if $b = \H$ then 
      $F'$ contains two new conservative extensions, 
      $\sigma_1$ and $\sigma_2$, 
      both with label $|xy| + 1$, 
      where $z_1 \Prefix \sigma_1$ and $r_1 \Prefix \sigma_2$.  
      These extensions, therefore, are $yb$-disjoint 
      and have zero reach.
      It follows that $\mu_x(F') \geq 0$.
  \end{description}
\end{proof}

Note that 
if we want \eqref{eq:mu-lowerbound} 
to hold only for \emph{a given prefix} $x \PrefixEq w$ 
(a scenario pertinent in~\cite{LinearConsistencySODA}), 
the adversary $\Adversary^*$ 
(which produces a canonical fork) 
would be an overkill. 
Instead, 
we can use a simpler, prefix-aware adversary 
such as the one mentioned 
at the outset of Section~\ref{sec:opt-adversary}; 
let us call this strategy $\Adversary$. 
In addition, 
instead of assuming Theorem~\ref{thm:opt-adversary-canonical}, 
it suffices to assume 
Proposition~\ref{prop:mu-lowerbound} inductively 
for all strings of length $|w|$. 
Let $F$ be the fork 
built by $\Adversary$ for the string $w = xy$.
In conjunction with Proposition~\ref{prop:mu-upperbound}, 
this would imply ``$\rho(F) = \rho(w)$ and 
$\mu_x(F) = \mu_x(y),$'' 
a critical property used inside the above proof. 
We omit further details.

\subsection{Proof of Theorem~\ref{thm:relative-margin} and Theorem~\ref{thm:opt-adversary-canonical}}

\paragraph{Proof of Theorem~\ref{thm:relative-margin}.}
Let $w \in \{\h, \H, \A\}^*$. 
If $w = \varepsilon$ then, by Claim~\ref{claim:rho-mu-A}, 
$\rho(\varepsilon) = 0$. 
If $|w| \geq 1$,~\eqref{eq:rho-recursive} is implied by 
the combination of 
Claim~\ref{claim:rho-mu-A},~\eqref{eq:rho-upperbound} and~\eqref{eq:rho-lowerbound}.

Let $w = xy$ be an arbitrary decomposition. 
We proceed by induction on $|y|$. 
If $|y| = 0$ then 
Claim~\ref{claim:rho-mu-A} implies that $\mu_x(\varepsilon) = \rho(x)$. 
Otherwise,~\eqref{eq:mu-relative-recursive} 
is implied by the combination of 
Claim~\ref{claim:rho-mu-A},~\eqref{eq:mu-upperbound} and~\eqref{eq:mu-lowerbound}.

\hfill\qed

\paragraph{Proof of Theorem~\ref{thm:opt-adversary-canonical}.}
The proof is by induction on $|w|$. 
If $w$ is the empty string $\varepsilon$, 
the only fork $F \Fork \varepsilon$ is the trivial fork 
containing a single (honest) root vertex. 
By Claim~\ref{claim:rho-mu-A}, 
$F$ satisfies $\rho(\varepsilon) = 0$ 
and 
$\mu_\varepsilon(\varepsilon) = \rho(\varepsilon) = 0$.

Now, let $n$ be a non-negative integer and 
let $w$ be a characteristic string of length $n+1$. 
Assume that Theorem~\ref{thm:opt-adversary-canonical} 
holds for all characteristic strings of length $0, 1, \ldots, n$. 
Note that this assumption satisfies 
the premise in Proposition~\ref{prop:mu-lowerbound}. 
A combined application of 
Claim~\ref{claim:rho-mu-A}, Proposition~\ref{prop:mu-upperbound}, 
and Proposition~\ref{prop:mu-lowerbound} 
implies 
Theorem~\ref{thm:opt-adversary-canonical} 
for $|w| = n + 1$.
\hfill\qed

%% file: async.tex
We set the stage by stating the $\Delta$-synchronous model.

\input{model-async}

\paragraph{Road-map for the proof.}
Let $w \in \{\perp, \h, \H, \A\}^*$, 
$w' = \Reduce(w), n = |w|$, and $m = |\Reduce|$. 
Our roadmap forward is as follows:
\begin{enumerate}
  \item 
  Show that there is a bijection between 
  $\Delta$-forks for $w$ and 
  synchronous forks for $w'$. 
  In particular, for each $\Delta$-fork $F \DeltaFork w$ 
  there is an isomorphic synchronous fork $F' \Fork_0 w'$ 
  and a bijective map $\{i \in [n] : w_i \neq \perp\} \rightarrow [m]$. 
  This is shown in Proposition~\ref{prop:reduction-bijection}.

  \item Show that if $w$ violates $\Delta$-settlement 
  then some prefix $b \Prefix \Reduce(w)$ violates 
  a suitably-defined combinatorial event $B_\Delta$.   
  It is important that we can analyze this event 
  using the techniques and results we have already established.
  This is done in Lemma~\ref{lemma:async-catalan-uvp}.

  \item Since the decisions made by $\Reduce$ at each slot 
  depends on the $\Delta$ future slots, 
  the distribution of the last few symbols of $\Reduce(w)$ 
  will be ``distorted'' no matter how $w$ is distributed. 
  Assuming $w$ has i.i.d.\ symbols, we need to 
  show that the symbols 
  in the aforementioned prefix $b \Prefix \Reduce(w)$ 
  are i.i.d.\ as well. 
  This is done in Lemma~\ref{prop:reduction-indep}.

  \item Obtain a bound on $Pr[B_\Delta]$ in Bound~\ref{bound:unique-honest-catalan-Delta} and 
  proceed to prove Theorem~\ref{thm:main-async}.
\end{enumerate}

\subsection{Structural properties of the reduction map}
An important property of the reduction $w' = \Reduce(w)$ is that 
it readily provides a bijection between $\Delta$-forks for $w$ 
and synchronous forks for $w'$.

\begin{proposition}\label{prop:reduction-bijection}
  Let $w \in \{\perp, \h, \H, \A\}^*$ 
  and $w' = \Reduce(w)$. 
  Then, for every $\Delta$-fork $F \Fork w$ there is 
  a synchronous fork $F' \Fork_0 w'$ 
  which is isomorphic to $F$. 
  $F'$ is called the \emph{image of $F$ under $\Reduce$}.
\end{proposition}
\begin{proof}[Proof sketch.]
  Let $F'$ be a copy of $F$. 
  Establish the natural bijection $m: V(F) \rightarrow V(F')$ 
  given by the copying proess, 
  i.e., $u \mapsto m(u)$, and 
  relabel the vertices as 
  \begin{equation}\label{eq:pi-reduction}
    \text{$\ell(m(u)) = \pi(\ell(u))$ for each vertex $u \in F$}
    \,.
  \end{equation}
  Set $r(F') = m(r(F))$ and $\ell(r(F')) = 0$.
  It suffices to check that $F' \Fork_0 w'$, 
  i.e., $F'$ is a valid (synchronous) fork for $w'$. 
  Specifically, 
  if there are two honest slots $h_1, h_2$ in $w$ 
  within a distance $\Delta$ of each other, 
  then the former honest slot is mapped to 
  an adversarial slot in $w'$. 
  Therefore, in $F'$, 
  an honest vertex is aware of 
  all honest vertices with smaller labels.
%
\end{proof}

Next, we show that 
a $\Delta$-settlement violation in $w$ implies 
a combinatorial event in $\Reduce(w)\TrimSlot{\Delta} \in \{\h, \H, \A\}^*$. 
It follows that we can use our existing stochastic techniques 
to bound $\Delta$-settlement violations on $w$. 

Let $w' \in \{\h, \H, \A\}^*$ be a characteristic string. 
Define $b_i \in \{\pm 1\}$ as $b_i = 1$ iff $w'_i = \A$. 
Let $S = (S_i)_{i =0}^{|w'|}$ be a simple biased walk on $\Z$ 
defined as $S_0 = 0, S_i = S_{i-1} + b_i$.

\begin{lemma}\label{lemma:async-catalan-uvp}
  Let $w \in \{\perp, \h, \H, \A\}^*, \Delta, s, k \in \NN$ 
  so that $|x| = s$ and $x_{s} \neq \perp$. 
  Let $w' = \Reduce(w)$ and 
  write $w' = x'y'z'a'$ so that $|a'| = \Delta$ and $|y'| \geq 2k$. 
  Recall the simple biased walk $S = (S_i)$ on $w'$ defined above.
  Let $E$ denote the event that 
  a slot $c'$ in $y'$ is Catalan in $x'y'z'$ 
  \emph{and} $S_{c' + k + i} \leq S_{c'} - \Delta$ for all $i \geq 0$.
  If $E$ occurs then 
  $s$ is $(|y'|,\Delta)$-settled in $w$. 
\end{lemma}
\begin{proof}
  Let $\pi$ be the bijection described after Definition~\ref{def:reduction-map}.
  Note that $|x'| = \pi(s)$. 
  Assume that $E$ occurs. 
  Thus $y'$ contains a uniquely honest slot $c'$ which is Catalan in $x'y'z'$. 
  Note that $S_{|w'|} \leq S_{|x'y'z'|} + \Delta \leq (S_{c} - \Delta) + \Delta \leq S_{c'}$ 
  where the second inequality follows from the assumption that $E$ occurs. 
  It follows that $c'$ is Catalan in $w'$ as well. 
  Therefore, 
  by Theorem~\ref{thm:unique-honest}, 
  $c'$ has the UVP in $w'$.     
  Let $c$ be the integer satisfying $c' = \pi(c)$. 

  Let $b \PrefixEq xyz, |b| \geq |xy|$ and $b' = \Reduce(b) \PrefixEq x'y'z'$. 
  (Necessarily, $|b'| \geq |x'y'|$.)
  Since the reduction map gives an isomorphism between every 
  $\Delta$-fork for $b$ and 
  its unique image (which is a synchronous fork for $b'$) 
  under the reduction $\Reduce$, 
  it follows that $c$ has the UVP in $w$. 

  For any $\Delta$-fork $F \DeltaFork b$, 
  let $u \in F, \ell(u) = c$ be 
  the unique vertex contained by every tine 
  $t \in F$ viable at the onset of any slot after $c$. 
  Consider all tines $\tau \in F$ so that 
  $\tau$ has at least $|y'|$ vertices with label at least $s + 1$.
  and $\tau$ is viable at the onset of slot $\ell(\tau) + 1$. 
  Since $\ell(\tau) \geq |xy| \geq c$, 
  it follows that $u \PrefixEq \tau$. 
  Thus all these tines $\tau$ 
  agree about slot $s$ since $s < c = \ell(u)$. 
  In particular, if $F$ contains two maximum-length tines $\tau_1, \tau_2$, 
  each with at least $|y'|$ vertices after slot $s$, 
  then they would agree about slot $s$. 
  In fact, $\ell(\tau_1 \Intersect \tau_2) \geq c > s$. 
  Hence $s$ must be $(|y'|, \Delta)$-settled in $F$ and, 
  since $F$ was arbitrary, 
  $s$ must be $(|y'|, \Delta)$-settled in $w$. 
  %
\end{proof}

\subsection{Stochastic properties of the reduction map}
It turns out that 
if the bits in $w$ are i.i.d.\ then 
so are the bits in a suitable prefix of $\Reduce(w)$ 
albeit with a slightly different distribution 
(which accounts for the absence of the empty slots). 
Specifically, for any string $x = x_1 x_2 \ldots $ on any alphabet and any $k \in \NN$, 
define $x\TrimSlot{k} \triangleq x_1 \ldots x_{|x| - k}$.

\begin{proposition}\label{prop:reduction-indep}
  Let $T \in \NN, w = w_1 \ldots w_T \in \{\perp, \h, \H, \A\}^T$ 
  be a sequence of i.i.d.\ symbols, 
  and define $p_\sigma \triangleq \Pr[w_1 = \sigma]$ for each $\sigma \in \{\perp, \h, \H, \A\}$.
  Let $x = \Reduce(w)$ and let $\ell = |x|$. 
  Write $f = 1 - p_\perp$ and $\alpha = (1 - f)^\Delta$. 
  Then the symbols in the string $x\TrimSlot{\Delta}$ are i.i.d.\ 
  with 
  \begin{align}\label{eq:reduction-dist}
    \begin{matrix*}[l]
      \Pr[x_i = \h] &=& p_\h \cdot \alpha / f\,, \\
      \Pr[x_i = \H] &=& p_\H \cdot \alpha / f\,, \quad \text{and}\\
      \Pr[x_i = \A] &=& 1 - \alpha + p_\A \cdot \alpha / f\,
    \end{matrix*}
  \end{align}
  for each $i \in [\ell - \Delta]$.
\end{proposition}
\begin{proof}
  First let us pretend for a moment that $T = \infty$; 
  then $\ell = \infty$ as well. 
  Let us write the infinite sequence $w$ as a concatenation of segments 
  of $\perp$s punctuated by a single non-$\perp$ symbol. 
  That is, write $w = b_0 e_1 b_1 e_2 b_2 \ldots$ 
  where, for $i = 0, 1, \ldots$, $b_i = \perp^*$ and $e_i \in \{\h, \H, \A\}$. 
  The reduction map $\Reduce$ translates a segment $e_i b_i$ into a symbol $z_i$ 
  as follows:
  \begin{align*}
    z_i &= \begin{cases}
      \A &\quad \text{if $e_i = \A$ or $|b_i| \leq \Delta - 1$}\, \\
      e_i &\quad \text{if $e_i \in \{\h, \H\}$ and $|b_i| \geq \Delta$}\,.
    \end{cases}
  \end{align*}
  In particular, the segments $e_i b_i$ as well as 
  the events that determine the value of an $z_i$ are disjoint. 
  Therefore, the symbols in the infinite sequence 
  $z_1 z_2 \ldots = \Reduce(w_1 w_w \ldots)$ are 
  independent and identically distributed.

  If $T$ is finite, however, the last $\Delta$ symbols 
  of $x = \Reduce(w)$ are ``distorted'' 
  in that the translated symbols in this region will be more favored to be $\A$s. 
  However, since the last $\Delta$ symbols of $x$ must correspond to 
  at least $\Delta$ trailing symbols of $w$, 
  it follows that $x_1 \ldots x_{\ell - \Delta}$ 
  is a prefix of $z_1 z_2 \ldots\ $. 

  It remains to compute the probabilities. 
  Let $q_\sigma = \Pr[z_i = \sigma]$ for any $i$ and $\sigma \in \{\h, \H, \A\}$. 
  Then 
  $
    q_\h = p_\h/(1 - p_\perp) p_\perp^{\Delta} = p_\h \alpha/f,
    q_\H = p_\H \alpha/f
  $, 
  and 
  $
  q_\A 
    = 1 - (q_\h + q_\H) 
    = 1 - (p_\h + p_\H)\alpha/f 
    = 1 - (f - p_\A)\alpha/f 
    = 1 - \alpha + p_\A \alpha/f
  $.
\end{proof}

The final ingredient to proving Theorem~\ref{thm:main-async} 
is a tail bound for (the complement of) the event $E$ 
in Lemma~\ref{lemma:async-catalan-uvp}.

\begin{bound}\label{bound:unique-honest-catalan-Delta}
  Let $T, s, k \in \NN, T \geq s + 2k + \Delta$ and 
  $\epsilon, q_\h \in (0, 1)$ so that 
  the characteristic string $w' \in \{\h, \H, \A\}^T$ 
  satisfies the $(\epsilon, q_\h)$-Bernoulli condition. 
  Write $w' = x'y'z'$ so that $y' = w_s \ldots w_{s + 2k - 1}$. 
  Let $G$ denote the event that 
  $w'$ has a Catalan slot $c$ 
  which belongs to $y'_1 \ldots y'_k$. 
  Condition on $G$. 
  Let $\Delta \in \NN$ and 
  recall the simple biased random walk $S = (S_i)$ on $w'$ 
  defined above Lemma~\ref{lemma:async-catalan-uvp}. 
  Let $B_\Delta$ be the event that 
  $S_{c + k + i} \geq S_{c} - \Delta$ for some $i \geq 0$. 
  Then for large $k$, 
  \begin{align}\label{eq:prob-Delta-after-catalan}
    \Pr_w[B_\Delta \mid G] \leq 
      \exp\left( 
        -k\cdot \Omega(\epsilon^2)
        + 
        \frac{\epsilon(1+\Delta)}{1 - \epsilon} 
    \right)
    \,.
  \end{align}
\end{bound}
\begin{proof}~
  For simplicity, write $p = q_\A$, and $q = q_\h + q_\H$. 
  Conditioned on $G$, $S_c \geq S_{c + i}$ for all $i \geq 1$. 
  Let $y = y'[c + 1 : c + k]$ so that $|y| = k$.
  Moreover, $\#_\A(y) \leq \#_\h(y) + \#_\H(y)$. 
  Let $f_i(k), i = 0, 1, \ldots$ be the probability that 
  $S_{c + k} = S_c - i$. 
  Thus we wish to upper-bound $f(\Delta, k) \triangleq \sum_{i = 0}^\Delta f_j(k)$.

  Write $a = E_\A(y)$ and $h = k - a$ 
  and suppose $h - a = j$ for some $j = 0, 1, 2, \ldots$\ .
  Hence, for a fixed $j$, we have $h = (k+j)/2$ and $a = (k-j)/2$. 
  In addition, $k$ and $j$ has the same parity.
  Thus, 
  \begin{align*}
    f_j(k) 
    &= {k\choose (k+j)/2} p^{(k-j)/2} q^{(k+j)/2} 
    = {k\choose (k+j)/2} (pq)^{k/2} (q/p)^{j/2} 
    \leq {k\choose k/2} (pq)^{k/2} (q/p)^{j/2} \\
    &= O(1)\cdot \frac{2^k}{\sqrt{\pi k}} \cdot (1 - \epsilon^2)^k 2^{-k} \cdot (q/p)^{j/2} 
    = O(1)\cdot  \frac{(1 - \epsilon^2)^{k/2}}{\sqrt{k}} \cdot (q/p)^{j/2}
  \end{align*}
  since $p = (1 - \epsilon)/2$ and $q = (1+\epsilon)/2$. 
  It follows that 
  \begin{align*}
    f(\Delta, k) &=\sum_{j = 0}^\Delta f_j(k) 
    \leq \frac{O(1)}{\sqrt{k}}\cdot  (1 - \epsilon^2)^{k/2} \sum_{j = 0}^\Delta (q/p)^{j/2}
    \leq \frac{O(1)}{\sqrt{k}}\cdot  \exp(- k \epsilon^2/2) \cdot (1 + \Delta) (q/p)^{\Delta/2}
    \,.
  \end{align*}
  Since 
  \begin{align*}
    (q/p)^{1/2}
    &= \left(\frac{1+\epsilon}{1 - \epsilon}\right)^{1/2}
    = \left(1 + \frac{2\epsilon}{1 - \epsilon}\right)^{1/2} 
    \leq \exp(\epsilon/(1 - \epsilon)) 
    \,,
  \end{align*}
  we have 
  \begin{align*}
    f(\Delta, k)
    &\leq \frac{O(1 + \Delta)}{\sqrt{k}}\cdot \exp\bigl(-k \epsilon^2/2 + (1 + \Delta) \epsilon/(1 - \epsilon) \bigr)
    \,.
  \end{align*}
  Note that for fixed $\epsilon$ and $\Delta$, $f(\Delta, k)$ decreases geometrically in $k$. 
  Thus $\Pr[B_\Delta \mid G_c] = \sum_{t \geq k} f(\Delta, t)$ 
  is no more than the quantity in~\eqref{eq:prob-Delta-after-catalan}.
\end{proof}

\subsection{Proof of Theorem~\ref{thm:main-async}}
  The symbols in $w$ are independent and identically distributed.
  Write 
  $w' = \Reduce(w), w' = x'y'z'a', |a'| = \Delta$ and $|y'| \geq 1+\Delta$. 
  Let $k$ be an integer so that $|y'| = 2k$. 
  Recall the random walk $S = (S_i)$ on $w'$ 
  defined above Lemma~\ref{lemma:async-catalan-uvp}. 
  Let $G_1$ denote the (good) event that 
  a slot $c'$ in $y'$ is Catalan in $x'y'z'$. 
  Let $G_2$ denote the (good) event that 
  $S_{c' + k + i} \leq S_{c'} - \Delta$ for all $i \geq 0$. 
  By Lemma~\ref{lemma:async-catalan-uvp}, 
  $G_1 \Intersect G_2$ implies $\overline{A}$. 
  (Here, $\overline{\cdot}$ denotes the complement.) 
  The contrapositive of the above statement gives us 
  \begin{equation}\label{eq:G1G2}
    \Pr[A] \leq \Pr[\overline{G_1}] + \Pr[\overline{G_2} \mid G_1]
    \,.
  \end{equation}

  The terms on the right-hand side can be bounded from above 
  using Bounds~\ref{bound:unique-honest-catalan} 
  and~\ref{bound:unique-honest-catalan-Delta}, respectively, 
  provided the symbols in $x'y'z'$ are i.i.d.\ with 
  $\Pr[x'_1 = \A] = (1 - \epsilon)/2$. 
  Let us check whether this condition holds. 
  We have $f = 1 - p_\perp$ and $\alpha = (1-f)^\Delta$. 
  Proposition~\ref{prop:reduction-indep} 
  states that 
  the symbols of $x'y'z'$ are i.i.d.\ 
  with distribution given by~\eqref{eq:reduction-dist}. 
  For each $\sigma \in \{\h, \H, \A\}$ we write 
  $p'_\sigma = \Pr[x'_1 = \sigma]$. 

  The condition~\eqref{eq:condition-prob-delta} 
  can be equivalently stated as 
  $1 - (1 - p_\A/f)\alpha = (1 - \epsilon)/2$. 
  We check that 
  $p'_\A 
  = 1 - (p'_\h + p'_\H) 
  = 1 - (p_\h + p_\H)\alpha/f
  = 1 - (f - p_\A)\alpha/f
  = 1 - (1 - p_\A/f)\alpha
  = (1 - \epsilon)/2
  $ 
  and, consequently, $p'_\h + p'_\H = (1 + \epsilon)/2$. 

  Hence we can directly apply 
  Bound~\ref{bound:unique-honest-catalan-Delta} 
  on the terms in the right-hand side of~\eqref{eq:G1G2} 
  to conclude that 
  \[
    \Pr[A] \leq \exp\left( 
      -k \left(\Omega(\min(\epsilon^3, \epsilon^2 q_\h)) \right) 
      + 
      \frac{\epsilon(1+\Delta)}{1 - \epsilon} 
    \right)
    \,.
  \]

  The claim involving the distribution $\mathcal{W}$ 
  follows from the analogous claim in Theorem~\ref{thm:main}. 
  \hfill $\qed$

%% file: model-async.tex
  \begin{definition}[Semi-synchronous characteristic string]\label{def:semisync-char-string}
    Let $\slot_1, \ldots, \slot_{n}$ be a sequence of slots. 
    A \emph{semi-synchronous characteristic string} $w$ 
    is an element of $\{\h,\H,\A, \perp\}^n$ 
    defined for a particular execution of a blockchain protocol 
    on these slots so that 
    for $t \in [n]$, 
    $w_t = \perp$ if $\slot_{t}$ was assigned to no participants; otherwise, 
    $w_t = \A$ if $\slot_{t}$ was assigned to an adversarial participant; otherwise, 
    $w_t = \h$ if $\slot_{t}$ was assigned to a single honest participant; otherwise 
    $w_t = \H$.
  \end{definition}

  In the $\Delta$-synchronous setting, axiom~\ref{axiom:honest-depth} 
  is replaced by 
  \begin{enumerate}[label={\textbf{A\arabic*}\textsubscript{$\Delta$}}., ref={\textbf{A\arabic*}\textsubscript{$\Delta$}}, start=4]
  \item\label{axiom:honest-depth-delta} 
    In a $\Delta$-synchronous execution, 
    if two honestly generated blocks $B_1$ and $B_2$ are labeled
    with slots $\slot_1$ and $\slot_2$ for which $\slot_1 + \Delta < \slot_2$,
    the length of the unique blockchain terminating at $B_1$ is
    strictly less than the length of the unique blockchain terminating at $B_2$.
  \end{enumerate}

  \begin{definition}[$\Delta$-Fork]\label{def:delta-fork}
    Let $w\in \{\h, \H, \A, \perp\}^n, \Delta \in \{0, 1, 2, \ldots \}$, 
    $P = \{ i : w_i = \h\}$, and $Q = \{ j : w_j = \H\}$. 
    A \emph{$\Delta$-fork} for the semi-synchronous string $w$ consists of a directed and rooted
    tree $F=(V,E)$ with a labeling $\ell:V \to \{0,1,\ldots,n\}$. We insist
    that each edge of $F$ is directed away from the root vertex. 
    We require conditions~\ref{fork:root}--\ref{fork:unique-honest} 
    from Definition~\ref{def:fork} and 
    \begin{enumerate}[label=(F{\arabic*}\textsubscript{$\Delta$}), start=4]



      \item\label{fork:honest-depth-delta} 
      for any indices $i,j\in P \Union Q$, 
      if $i + \Delta < j$ then 
      the depth of a vertex with label $i$ 
      is strictly less than 
      the depth of a vertex with label $j$.
    \end{enumerate}
  \end{definition}
  If $F$ is a $\Delta$-fork for the semi-synchronous characteristic string $w$, we write
  $F\vdash_\Delta w$.  
  A $\Delta$-fork generalizes a synchronous fork in Definition~\ref{def:fork} 
  since the latter is a $\Delta$-fork with $\Delta = 0$. 
  We sometimes emphasize this fact by writing $F' \Fork_0 w'$ 
  where $w'$ is a synchronous characteristic string and $F'$ is a synchronous fork.
  Note that 
  condition~\ref{fork:honest-depth-delta} 
  is a direct analogue of axiom~\ref{axiom:honest-depth-delta}.
  (We already know that conditions~\ref{fork:root}--\ref{fork:unique-honest} 
  are direct
  analogues of axioms~\ref{axiom:root}--~\ref{axiom:honest}.)

  \begin{definition}[Reduction map]\label{def:reduction-map}
    For $\Delta \in \NN$, 
    we define the function $\Reduce : \{\perp, \h, \H, \A\}^* \rightarrow \{\h, \H, \A\}^*$ 
    inductively as follows: $\Reduce(\varepsilon) = \varepsilon$ and 
    for $w \in \{\perp, \h, \H, \A\}^*$, 
    \begin{align}
      \Reduce(b w) &= 
      \begin{cases}
        \Reduce(w) & \quad\text{if $b = \perp$}\,, \\
        b \Reduce(w) & \quad\text{if $b \in \{\h, \H\}$ and $\{\perp, \A\}^{\Delta} \PrefixEq w$}\,, \\
        \A \Reduce(w) & \quad\text{otherwise}
        \,.
      \end{cases}
    \end{align}    
  \end{definition}
  \noindent
  Note that in the above definition, 
  if $w' = \Reduce(w)$ and 
  $A = \{i : w_i \neq \perp\}$ then $|A| = |w'|$. 
  Also note that 
  the reduction $\Reduce$ implicitly defines, for each $w$, 
  a 
  bijective, increasing function $\pi : A \rightarrow [|w'|]$. 
  Note that 
  $\Reduce$ turns an $\h$ or $\H$ symbol in $w$ 
  into an $\A$ symbol in $w'$ with a constant probability. 
  Therefore, for any slot $t$ in $w$, 
  the reduction map $\Reduce$ 
  amplifies the probability that the slot $\pi(t)$ 
  in $w' = \Reduce(w)$ is adversarial.

  \begin{definition}[$\Delta$-settlement with parameters $s,k \in \NN$]\label{def:settlement-delta}
    Let $n \in \NN$ and let $w \in \{\perp, \h, \H, \A\}^n$. 
    Let $t \in [s + k, n]$ be an integer, $\hat{w} \PrefixEq w, |\hat{w}| = t$, and 
    let $F$ be any $\Delta$-fork for $\hat{w}$. 
    We say that a slot $s$ is \emph{not $(k, \Delta)$-settled in $F$} 
    if $F$ contains 
    two maximum-length tines $\Chain_1, \Chain_2$ so that 
    at least one of these tines contains a vertex with label $s$,
    both tines contain at least $k$ vertices after slot $s$, and 
    the label of their last common vertex is at most $s - 1$. 
    Otherwise, we say that \emph{slot $s$ is $(k, \Delta)$-settled in $F$}. 
    We say that \emph{slot $s$ is $(k, \Delta)$-settled in $w$} if, 
    for each $t \geq s+k$, 
    it is $(k, \Delta)$-settled in every $\Delta$-fork $F \Fork \hat{w}$ 
    where $\hat{w} \PrefixEq w, |\hat{w}| = t$.
  \end{definition}
  Note that in the above definition, 
  we truncated $k$ trailing blocks from a tine 
  whereas in Definition~\ref{def:settlement}, 
  we truncated from a tine 
  all trailing blocks corresponding to the last $k$ slots. 
  Note that this change of perspective is necessary since 
  $w$ may contain $\perp$ symbols, i.e., empty slots.


    \begin{theorem}[Main theorem; $\Delta$-synchronous setting]\label{thm:main-async}
      Let $f, \epsilon \in (0,1)$ and $\Delta \in \{0, 1, 2, \ldots\}$. 
      Let $s, k, T \in \NN$ so that $T \geq s + k + \Delta$. 
      Write $p_\perp = 1 - f$ and $\beta = (1 - f)^\Delta$. 
      Let $p_\A \in [0, f)$ so that $p_\A, f, \epsilon$, and $\beta$ 
      satisfy 
      \begin{equation}\label{eq:condition-prob-delta}
        p_\A  \beta/f + (1 - \beta) \leq (1 - \epsilon)/2 
        \,.
      \end{equation}      
      Let $p_\h \in (0, f - p_\A]$ and write $p_\H = f - p_\A - p_\h$. 
      Let $w \in \{\perp, \h, \H, \A\}^T$ be a random variable 
      so that each 
      $w_i, i \in [T]$, is independent and identically distributed as
      follows: $\Pr[w_i = \sigma] = p_\sigma$ for
      $\sigma \in \{\h, \H, \A, \perp\}$. 
      Let $\mathcal{B}$ be the distribution of $w$. 
      Then 
      $$
        \Pr_w[\text{slot $s$ is not $(k, \Delta)$-settled in $w$}]  
          \leq
        \exp\left( 
          -k\cdot \Omega(\min(\epsilon^3, \epsilon^2 p_\h \beta/f ) ) 
          + 
          \frac{\epsilon(1+\Delta)}{1 - \epsilon} 
        \right)
        \,.
      $$
      (Here, the asymptotic notation hides constants that do not depend on $\epsilon$ or $k$.)
    \end{theorem}

    The main observation for proving the theorem above is that 
    a $\Delta$-settlement violation in $w$, 
    implies a certain combinatorial event (parameterized by $\Delta$) 
    in a prefix of $\Reduce(w)$. 
    Specifically, we can analyze the latter event 
    using techniques developed in proving Theorem~\ref{thm:main}. 

    \paragraph{A comment on consistent chain selection.} 
    Assuming axiom~\ref{axiom:tie-breaking} is satisfied, 
    it is easy to prove an analogue of Theorem~\ref{thm:main-bivalent} 
    in the $\Delta$-synchronous setting; 
    we need only use Bound~\ref{bound:two-catalans} 
    in lieu of Bound~\ref{bound:unique-honest-catalan}. 
    The resulting bound on the probability of a $(k, \Delta)$-settlement violation would be 
    $$
      \exp\left( 
        -k\cdot \Omega(\epsilon^3) 
        + 
        \frac{\epsilon(1+\Delta)}{1 - \epsilon} 
      \right)
      \,.
    $$
    We omit further details.


%% file: cp.tex
  For the sake of simplicity, 
  assume the synchronous communication model from Section~\ref{sec:game}; 
  the $\Delta$-synchronous setting can be handled in the same way 
  as delineated in Sections~\ref{sec:async-model} and~\ref{sec:async}.

  The common prefix property with parameter $k$ asserts
  that, for any slot index $s$, if an honest observer at slot $s + k$
  adopts a blockchain $\Chain$, the prefix $\Chain[0 : s]$ will be
  present in every honestly-held blockchain at or after slot $s + k$.
  (Here, $\Chain[0 : s]$ denotes the prefix of the blockchain $\Chain$
  containing only the blocks issued from slots $0, 1, \ldots, s$.)

  We translate this property into the framework of forks.  Consider a
  tine $t$ of a fork $F \vdash w$.  The \emph{trimmed} tine
  $t\TrimSlot{k}$ is defined as the portion of $t$ labeled with slots
  $\{ 0, \ldots, \ell(t) - k\}$. For two tines, we use the notation
  $t_1 \PrefixEq t_2$ to indicate that the tine $t_1$ is a
  prefix of tine $t_2$.

  \begin{definition}[Common Prefix Property with parameter $k \in \NN$]\label{def:cp-slot}
    Let $w$ be a characteristic string. A fork $F \vdash w$ satisfies
    $\kSlotCP$ if, for all pairs $(t_1, t_2)$ of viable tines $F$ for
    which $\ell(t_1) \leq \ell(t_2)$, we have $t_1\TrimSlot{k} \PrefixEq t_2$. 
    Otherwise, we say that the tine-pair $(t_1, t_2)$ is a witness to a $\kSlotCP$ violation.
    Finally, \emph{$w$ satisfies $\kSlotCP$} if every fork $F \vdash w$ satisfies $\kSlotCP$.
  \end{definition} 
  If a string $w$ does not possess the $\kSlotCP$ property, 
  we say that \emph{$w$ violates $\kSlotCP$}.
  Observe that traditionally
  (cf. \cite{GKL17}), 
  the truncated chain 
  is defined in terms of
  deleting a suffix of (block-)length $k$ from $\Chain$. 
  We denote this traditional version of the common prefix property as the
  $\kCP$ property. Note, however, that a $\kCP$ violation immediately
  implies a $\kSlotCP$ violation; hence, bounding the probability of a
  $\kSlotCP$ violation is sufficient to rule out both events.

  \paragraph{Connection with the UVP.}
  Note that 
  if $w$ admits a $\kSlotCP$ violation, 
  then there must be a fork $F$ containing 
  two distinct viable tines $t_1, t_2, \ell(t_1) \leq \ell(t_2)$ 
  so that $\ell(t_1) - \ell(t_1 \Intersect t_2) \geq k + 1$. 
  Then $t_1$ must contain a vertex $v, \ell(t_1 \Intersect t_2) < \ell(v) \leq \ell(t_1) - k$ 
  so that $v$ does not belong to $t_2$. 
  If every substring $x$ of $w$ with $|x| \geq k$, contained a slot with the UVP then 
  we would never have a $\kSlotCP$ violation. 
  Therefore, 
  \begin{align}\label{eq:cp-uvp}
    \parbox{20mm}{\centering $w$ violates $\kSlotCP$}
    &\implies
    \parbox{45mm}{\centering
      $w$ has a substring $y, |y| \geq k$ so that 
      no slot indexed by $y$ has the UVP in $w$.
    }
  \end{align}

  Recall that a uniquely honest Catalan slot has the UVP. 
  This fact allows us to bound 
  the probability of common prefix violations by 
  reasoning only about Catalan slots.\footnote{ 
  One can also prove Theorem~\ref{thm:main-CP} 
  by 
  directly showing---as is done in~\cite{LinearConsistency}---that 
  a $\kSlotCP$ violation implies a $k$-settlement violation 
  and then appealing to Theorem~\ref{thm:main}. 
  The proof of the implication 
  turns out to be quite long and complicated 
  compared to the short proof above; 
  see Appendix~\ref{sec:cp-forks}.
  A positive side of this this alternate proof, however,  
  is that it shows how arguments in~\cite{LinearConsistency} 
  can be adapted to our generalized fork framework.
  }

  \begin{theorem}[Main theorem; CP version] \label{thm:main-CP} Let
    $\epsilon, p_\h \in (0,1)$ and $T, k \in \NN, T \geq k$. 
    Let $w$ be a length-$T$ characteristic string satisfying the $(\epsilon, p_\h)$-Bernoulli condition. 
    Then
    $$
      \Pr_{w}[\text{$w$ violates $\kCP$}] 
        \leq 
      \Pr_{w}[\text{$w$ violates $\kSlotCP$}] 
        \leq T \cdot 
        \exp\bigl(-k\cdot \Omega(\min(\epsilon^3, \epsilon^2 p_\h))\bigr)
        \,.
    $$
    
    Next, suppose that axiom~\ref{axiom:tie-breaking} is satisfied. 
    If $w$ is a length-$T$ bivalent characteristic string satisfying the $(\epsilon, 0)$-Bernoulli condition 
    then
    $$
      \Pr_{w}[\text{$w$ violates $\kCP$}] 
        \leq 
      \Pr_{w}[\text{$w$ violates $\kSlotCP$}] 
        \leq T \cdot 
        \exp\bigl(-k \cdot \Omega(\epsilon^3 (1 + O(\epsilon)))\bigr)
        \,.
    $$
  \end{theorem}
  \begin{proof}
    \newcommand{\EpsCat}{\mathsf{\varepsilon}}
    (The first claim.) 
    Let $s \in [T - k]$.
    Let $\EpsCat_{k}$ be 
    the probability that $y = w_s \ldots w_{s+k-1}$ contains no slot with the UVP in $w$. 
    Then, recalling~\eqref{eq:cp-uvp}, we can apply a union bound 
    over all substrings of $w$ of length at least $k$ to get 
    $
      \Pr[\text{$w$ violates $\kSlotCP$}] 
        \leq T\, \sum_{r \geq  k} \EpsCat_{r}
    $
    where the factor $T$ represents a summation over all $s \in [T - k + 1]$. 
    By Theorem~\ref{thm:unique-honest}, 
    if a substring $y$ of $w$ does not contain a slot with the unique vertex property in $w$, 
    $y$ cannot contain a uniquely honest slot that is Catalan in $w$.
    Therefore, $\EpsCat_k$ is no more than the error probability from Bound~\ref{bound:unique-honest-catalan}. 
    Since $\EpsCat_k$ decreases exponentially in $k$, 
    we can write 
    \begin{align*}
      \Pr[\text{$w$ violates $\kSlotCP$}] 
        &\leq T\cdot O(1) \cdot \EpsCat_{k}
        \,.
    \end{align*}
    This proves the second inequality. 
    The first inequality follows since, 
    in a given characteristic string, 
    a $\kCP$ violation implies a $\kSlotCP$ violation.

    (The second claim.) 
    The proof in this case 
    is identical to the preceding argument except that 
    we need to refer to Theorem~\ref{thm:multiple-honest} in lieu of Theorem~\ref{thm:unique-honest}
    and Bound~\ref{bound:two-catalans} in lieu of Bound~\ref{bound:unique-honest-catalan}.
  \end{proof}
  \paragraph{The $\Delta$-synchronous setting.} 
  A $\kCP$ violation in a $\Delta$-fork for a string $w \in \{\perp, \h, \H, \A\}^*$ 
  would imply 
  a $\kCP$ violation in the corresponding synchronous fork 
  in the string $\Reduce(w) \in \{\h, \H, \A\}^*$ 
  and, consequently, a $\kSlotCP$ violation in $\Reduce(w)$. 
  We omit further details.

%% file: cp-forks.tex
Balanced forks played a critical role 
in the analysis of~\cite{LinearConsistency}. 
Specifically, a balanced fork was equivalent to a settlement violation in their setting 
and a CP violation would also imply a balanced fork.
In the current analysis, 
we have analyzed settlement and CP violations through 
their connections with the UVP and Catalan slots; 
thus balanced forks are not necessary in our analysis. 
However, it is instructive to see 
whether the statement ``a CP violation implies a balanced fork'' 
still holds in our model 
and, importantly, 
how the existing proof needs to be modified. 

Thus the the goal of this section is to prove 
Theorem~\ref{thm:divergence-settlement} below which 
would yield an alternative proof of Theorem~\ref{thm:main-CP} 
without using the Catalan slots.
However, the simplicity of the proof of Theorem~\ref{thm:main-CP} 
in Section~\ref{sec:cp} 
demonstrates the expressive power of the UVP and Catalan slots 
compared to relative margin and balanced forks.

\paragraph{A $\kSlotCP$ violation implies a $k$-settlement violation.}
Let $w$ be a characteristic string, written $w = xy$, 
and let $F$ be a fork for $w$. 
Recall that a slot $s = |x| + 1$ is not $k$-settled 
if and only if $F$ contains 
two maximum-length tines that diverge prior to $s$, 
i.e., $F$ is $x$-balanced (see Definition~\ref{def:balanced-fork}).

\begin{definition}[Slot divergence]\label{def:slot-divergence}
  Let $w \in \{\h, \H, \A\}^*$ and let $F$ be a fork for $w$. 
  Define the \emph{slot divergence} of 
  two tines $t_1, t_2 \in F$ 
  as 
  \begin{equation}\label{eq:slot-divergence-tines}
    \SlotDivergence(t_1, t_2) \defeq \ell(t_1) - \ell(t_1 \Intersect t_2)
    \quad\text{where $\ell(t_1) \leq \ell(t_2)$}
    \,.
  \end{equation}
  We can generalize this notion for forks and characteristic strings as follows: 
  $\SlotDivergence(F) \triangleq \max_{t_1, t_2 \in F} \SlotDivergence(t_1, t_2)$ and 
  $\SlotDivergence(w) \triangleq \max_{F \Fork w} \SlotDivergence(F)$. 
\end{definition}

By definition, a $\kSlotCP$ violation 
implies the existence of a fork with a slot divergence at least $k + 1$. 
Theorem~\ref{thm:divergence-settlement} below 
shows that a if a fork has a slot divergence at least $k+1$ then 
there is a balanced fork for a prefix of the same characteristic string so that 
two maximum-length tine diverge prior to last $k$ slots. 
Therefore, a $\kSlotCP$ violation implies an $(s,k)$-settlement violation 
for some slot $s$.

\begin{theorem}\label{thm:divergence-settlement}
  Let $k, T \in \NN$.  
  Let $w \in \{\h, \H, \A\}^T$ be a characteristic string 
  so that $\SlotDivergence(w) \geq k + 1$.
  Then 
  there is a decomposition $w = xyz$ and a fork $\hat{F} \Fork xy$, 
  where $|y| \geq k$, 
  so that 
  $\hat{F}$ is $x$-balanced.
\end{theorem}

\newcommand{\Final}[1]{\tilde{#1}}

  
  Recall that $\ell(t)$ is the slot index of the last vertex of tine
  $t$.  
  Define $A \triangleq \bigcup_{F \Fork w} A_F$ where, for a
  given fork $F \Fork w$, define
  \[
    A_F \triangleq \left\{
      (\tau_1, \tau_2) \SuchThat \parbox{60mm}{       
      $\tau_1, \tau_2$ are two viable tines in the fork $F$, 
      $\ell(\tau_1) \leq \ell(\tau_2)$, and 
      $\SlotDivergence(\tau_1, \tau_2) \geq k + 1$
      }
     \right\}
     \,.
  \]
  Notice that there must be a tine-pair $(t_1, t_2) \in A$ which satisfies the following two conditions: 
    \begin{equation}\label{eq:tines}
      \SlotDivergence(t_1, t_2) 
      \text{ is maximal over $A$\,,}
    \end{equation}
  \begin{equation}\label{eq:minimality}
    \parbox{0.85\columnwidth}{\centering
    $| \ell(t_2) - \ell(t_1) |$ 
      is minimal among all tine-pairs in $A$ 
      for which~\eqref{eq:tines} holds\,, 
      }
  \end{equation}
  and
  \begin{equation}\label{eq:length-multihonest}
    \parbox{0.85 \columnwidth}{\centering
    For a fixed $t_2$, 
    the tine $t_1$ has the maximum length 
    over all tines $t_1', \ell(t_1') = \ell(t_1)$ \\
    such that $(t_1', t_2)$ 
    satisfies~\eqref{eq:tines} and~\eqref{eq:minimality}\,. 
    }
  \end{equation}
  (Note that $t_1, t_2$ are not uniquely identified.)
  The tines $t_1, t_2$ will play a special role in our proof; 
  let $F$ be a fork containing these tines. 

  Recall given a characteristic string $w \in \{\h, \H, \A\}^*$, 
  a uniquely honest slot contains the symbol $\h$, 
  a multiply honest slot contains the symbol $\H$, 
  and an adversarial slot contains the symbol $\A$.
  We call a slot honest if it contains either an $\h$ or an $\H$; 
  otherwise, we call it an adversarial slot. 

  \paragraph{The prefix $x$, fork $F_x$, and vertex $u$.} 
  Let $u$ denote the last vertex on the tine
  $t_1 \cap t_2$, as shown in the diagram below, and let
  $\alpha \triangleq \ell(u) = \ell(t_1 \cap t_2)$. 
  Let $x \triangleq w_1, \ldots, w_\alpha$ 
  and let $F_x$ be the fork-prefix of $F$ supported on $x$. 
  We will argue that $\alpha$ must be a uniquely honest slot and, 
  in addition, that 
  $F_x$ must contain a unique longest tine $t_u$ terminating 
  at the vertex $u$. 
  We will also identify a substring 
  $y, |y| \geq k$ 
  such that $w$ can be written as $w = xyz$. 
  Then we will construct a balanced fork $\tilde{F}_y \Fork y$ by 
  modifying the subgraph of $F$ supported on $y$. 
  We will finish the proof by constructing an $x$-balanced fork by 
  suitably appending $\tilde{F}_y$ to $F_x$.
    
  \begin{center}
      \begin{tikzpicture}[>=stealth', auto, semithick,
        unknown/.style={circle,draw=black,thick,font=\small},
        honest/.style={circle,draw=black,thick,double,font=\small},
        malicious/.style={fill=gray!10,circle,draw=black,thick,font=\small}]
        \node[honest] at (0,0) (u) {$u$};
        \node[malicious] at (3,.5)  (z1) {};
        \node[malicious] at (5,-.5)   (z2) {};
        \path (z1) ++(.4,.4) node {$t_1$};
        \path (z2) ++(.4,.4) node {$t_2$};
        \draw[thick,<-] (u) to (-1,0);
        \draw[thick,<-,gray] (z1) to[out=180,in=20] (u);
        \draw[thick,<-,gray] (z2) to[out=180,in=-20] (u);
      \end{tikzpicture}
    \end{center}

    \paragraph{$\alpha$ must be a uniquely honest slot.}
    We observe, first of all, that the slot $\alpha$ can neither be adversarial nor multiply honest:
    otherwise it is easy to construct a fork
    $F^\prime \Fork w$ and a pair of tines in $F^\prime$ that violate~\eqref{eq:tines}. 
    Specifically, construct $F^\prime$ from $F$ by
    adding a new vertex $u^\prime$ to $F$ for which
    $\ell(u^\prime) = \ell(u)$, adding an edge to $u^\prime$ from the
    vertex preceding $u$, and replacing the edge of $t_1$ following $u$
    with one from $u^\prime$; then the other relevant properties of the
    fork are maintained, but the slot divergence of the resulting tines has
    increased by at least one. (See the diagram below.)
    \begin{center}
      \begin{tikzpicture}[>=stealth', auto, semithick,
        unknown/.style={circle,draw=black,thick,font=\small},
        honest/.style={circle,draw=black,thick,double,font=\small},
        malicious/.style={fill=gray!10,circle,draw=black,thick,font=\small}]
        \node[malicious] at (2,0) (v) {$u$};
        \node[malicious,dotted] at (2,1) (u) {$u^\prime$};
        \node[unknown] at (4,-.5)  (b1) {};
        \node[unknown] at (4,.5)  (a1) {};
        \node[unknown] at (0,0) (base) {};
        \node at (7,.5) (t1) {$t_1$};
        \node at (7,-.5) (t2) {$t_2$};
        \draw[thick,->] (base) -- (v);
        \draw[thick,->] (v) -- (a1);
        \draw[thick,->] (v) -- (b1);
        \draw[thick,->,dotted] (u) -- (a1);
        \draw[thick,->,dotted] (base) -- (u);
        \draw[thick,<-,gray] (t1) to[in=20,out=200] (a1);
        \draw[thick,<-,gray] (t2) to[in=20,out=200] (b1);
        \draw[thick,<-,gray] (base) to (-1,0);
      \end{tikzpicture}
    \end{center}
    
    \paragraph{$F_x$ has a unique, longest (and honest) tine $t_u$.}
    A similar argument implies that the fork
    $F_x$ has a unique vertex of depth $\depth(u)$: namely, $u$ itself. In
    the presence of another vertex $u^\prime$ (of $F_x$) with depth
    $\depth(u)$, ``redirecting'' $t_1$ through $u^\prime$ (as in the
    argument above) would likewise result in a fork with 
    a larger slot divergence. 
    To see this, notice that $\ell(u^\prime)$ must be strictly less than $\ell(u)$ 
    since $\ell(u)$ is an honest slot (which means $u$ is the only vertex at that slot).
    Thus $\ell(\cdot)$ would indeed be increasing along
    this new tine (resulting from redirecting $t_1$).
    As $\alpha$ is the last index of the string $x$, this additionally
    implies that $F_x$ has no vertices of depth exceeding $\depth(u)$. 
    Let $t_u \in F_x$ be the tine with $\ell(t_u) = \alpha$. 
    \begin{equation}\label{eq:tu}
        \text{The honest tine $t_u$ is the unique longest tine in $F_x$}
        \,.
    \end{equation}

    \paragraph{Identifying $y$.}
    Let $\beta$ denote the smallest honest index of $w$ for which 
    $\beta \geq \ell(t_2)$, with the convention that if there is no such
    index we define $\beta = T + 1$. 
    Thus $\beta \geq \ell(t_2) \geq \ell(t_1)$.
    These indices, $\alpha$ and $\beta$, distinguish the
    substrings $y = w_{\alpha+1} \ldots w_{\beta-1}$ and 
    $z = w_{\beta} \ldots w_T$; 
    we will focus on $y$ in the remainder of the proof. 
    Since the function
    $\ell(\cdot)$ is strictly increasing along any tine, observe that
    \begin{align*}
        |y| 
        &= (\beta - 1) - (\alpha + 1) + 1 
        = \beta - \alpha - 1 
        \geq (\ell(t_1) - \ell(u)) - 1 
        \geq (k + 1) - 1 
        = k
        \,.
    \end{align*}
    Hence $y$ has the desired length and it suffices to establish that it is forkable.\footnote{
      In~\citet{LinearConsistency}, $|y|$ was at least $k + 1$. 
      The difference is due to the fact that 
      in their analysis, a slot with multiple vertices 
      was necessarily adversarial. 
    }

    \paragraph{Honest indices in $xy$ have small depths.}
    The minimality assumption~\eqref{eq:minimality} implies that any honest
    index $h$ for which $h < \beta$ has depth no more than
    $\min(\length(t_1),\length(t_2))$: specifically, we claim that 
    \begin{equation}\label{eq:honest-depth}
      h < \beta \quad\Longrightarrow \quad \hdepth(h) \leq \min(\length(t_1), \length(t_2))\,.
    \end{equation}
    To see this, consider an honest index $h,h < \beta$ and a tine $t_h$
    for which $\ell(t_h) = h$. 
    If $\ell(t_2)$ is honest then $h < \beta = \ell(t_2)$. 
    Otherwise, $h < \ell(t_2) < \beta$ since $\ell(t_2)$ is adversarial. 
    In any case, $h < \ell(t_2)$ and, 
    since $t_2$ is viable, it follows immediately that
    $\hdepth(h) \leq \length(t_2)$. 
    Similarly, if $h < \ell(t_1)$ 
    then $\hdepth(h) \leq \length(t_1)$ since $t_1$ is viable as well. 
    
    Now consider the case $h = \ell(t_1)$. 
    We claim that 
    \begin{equation}\label{eq:hdepth-t1}
      \text{
        If $h = \ell(t_1) < \beta$ then $\hdepth(h) = \length(t_1)$ 
        }
      \,.
    \end{equation}
    We can rule out the case $h = \ell(t_1) = \ell(t_2)$ 
    since if this happens, 
    $\ell(t_2)$ is honest and $\beta = \ell(t_2)$, 
    contradicting our assumption that $h < \beta$. 
    Thus, it must be the case that $h = \ell(t_1) < \ell(t_2)$.    
    In this case, the claim follows trivially 
    if $\ell(t_1)$ is a uniquely honest slot. 
    Otherwise, let $t$ be a tine 
    with the maximum length among all tines 
    labeled with the multiply honest slot $h = \ell(t_1) < \ell(t_2)$. 
    We wish to show that $\length(t_1) = \length(t)$. 
    There are four contingencies to consider; 
    the first three of these lead to contradictions 
    and for the last one, we get $\length(t_1) = \hdepth(h) = \length(t)$.
    \begin{itemize}

      \item If $(t, t_2) \not \in A$, 
      $\SlotDivergence(t, t_2)$ is at most $k$.
      Since $\SlotDivergence(t_1, t_2)$ is at least $k + 1$, 
      $t$ must share a vertex with $t_2$ after slot $\ell(u)$. 
      But this means $\ell(t \Intersect t_1) = \ell(u)$ 
      and $\SlotDivergence(t, t_1) = \SlotDivergence(t_1, t_2) \geq k + 1$. 
      As a result, $(t, t_1) \in A$. 
      However, this violates~\eqref{eq:minimality} 
      since $|\ell(t) - \ell(t_1)| = 0 < |\ell(t_2) - \ell(t_1)|$ by assumption. 

      \item 
      If $(t, t_2)$ is in $A$ and 
      $\ell(t \Intersect t_1) < \ell(u)$, 
      then $\SlotDivergence(t, t_1) > \SlotDivergence(t_1, t_2)$, 
      violating~\eqref{eq:tines}. 

      \item 
      If $(t, t_2)$ is in $A$ and 
      $\ell(t \Intersect t_1) = \ell(u)$, 
      this means $t$ is disjoint with $t_1$ after $\ell(u)$. 
      Then~\eqref{eq:minimality} is violated 
      since $\SlotDivergence(t, t_1) = \SlotDivergence(t_1, t_2)$ but 
      $|\ell(t) - \ell(t_1)| = 0 < |\ell(t_2) - \ell(t_1)|$ by assumption. 

      \item 
      If $(t, t_2)$ is in $A$ and 
      $\ell(t \Intersect t_1) > \ell(u)$, 
      this means $t$ shares a vertex with $t_1$ after $\ell(u)$. 
      Then $\SlotDivergence(t, t_2) = \SlotDivergence(t_1, t_2)$ 
      and $|\ell(t_2) - \ell(t_1)| = |\ell(t_2) - \ell(t)|$. 
      By~\eqref{eq:length-multihonest}, 
      $\length(t_1) \geq \length(t)$; 
      hence $\length(t_1) = \length(t)$ since by assumption, 
      $t$ has the maximum length among all tines with label $\ell(t_1)$. 
      Hence $\length(t_1) = \hdepth(h)$.

    \end{itemize}
    The remaining case for proving~\eqref{eq:honest-depth}, 
    i.e., when $\ell(t_1) < h < \ell(t_2)$, 
    can be ruled out by the argument below.

    \paragraph{There is no honest index between $\ell(t_1)$ and $\ell(t_2)$.}
    We claim that 
    \begin{equation}\label{eq:no-honest-index}
        \text{There is no honest index $h$ satisfying $\ell(t_1) < h < \ell(t_2)$}
        \,.
    \end{equation}
    The claim above is trivially true if $\ell(t_1) = \ell(t_2)$.
    Otherwise, suppose (toward a contradiction) 
    that $h$ is an honest index satisfying $\ell(t_1) < h < \ell(t_2)$. 
    Let $t_h$ be an honest tine at slot $h$. 
    The tine-pair $(t_1, t_h)$ may or may not be in $A$. 
    We will show that both cases lead to contradictions.
    \begin{itemize}
      \item If $(t_1, t_h)$ is in $A$ and $\ell(t_1 \Intersect t_h) \leq \ell(u)$, 
      $\SlotDivergence(t_1, t_h)$ is at least $\SlotDivergence(t_1, t_2)$. 
      In fact, due to~\eqref{eq:tines}, this inequality must be an equality. 
      However, the assumption $\ell(t_1) < h < \ell(t_2)$ contradicts~\eqref{eq:minimality}. 

      \item If $(t_1, t_h)$ is in $A$ and $\ell(t_1 \Intersect t_h) > \ell(u)$, 
      it follows that $\SlotDivergence(t_h, t_2) > \SlotDivergence(t_1, t_2)$. 
      As the latter quantity is at least $k + 1$, $(t_h, t_2)$ must be in $A$. 
      The preceding inequality, however, contradicts~\eqref{eq:tines}.

      \item If $(t_1, t_h) \not \in A$, 
      $\SlotDivergence(t_1, t_h)$ is at most $k$.
      As $\SlotDivergence(t_1, t_2)$ is at least $k + 1$, 
      $t_h$ and $t_1$ must share a vertex after slot $\ell(u)$. 
      Since $\ell(t_1) < h < \ell(t_2)$ by assumption, 
      $\SlotDivergence(t_h, t_2) > \SlotDivergence(t_1, t_2) \geq k + 1$ 
      and, as a result, $(t_h, t_2) \in A$. 
      However, the strict inequality above violates~\eqref{eq:tines}. 
    \end{itemize}
    We conclude that~\eqref{eq:no-honest-index}---and thus~\eqref{eq:honest-depth}---is true. 
    (Note that in the above argument, all we needed was that $t_h$ is a viable tine 
    since in all cases, $t_h$ appears in a tine-pair in $A$. 
    Thus~\eqref{eq:no-honest-index} can be generalized as saying 
    ``there is no fork for $w$ with a viable tine $t$ so that $\ell(t_1) < \ell(t) < \ell(t_2)$.'')

  \paragraph{A fork $\pinch{u}{F}$ where all long tines go through $u$.}
    In light of the remarks above, we observe that the fork $F$ may be
    ``pinched'' at $u$ to yield an essentially identical fork
    $\pinch{u}{F} \vdash w$ with the exception that all tines of length
    exceeding $\depth(u)$ pass through the vertex $u$. Specifically, the
    fork $\pinch{u}{F} \vdash w$ is defined to be the graph obtained
    from $F$ by changing every edge of $F$ directed towards a vertex of
    depth $\depth(u) + 1$ so that it originates from $u$. To see that
    the resulting tree is a well-defined fork, it suffices to check that
    $\ell(\cdot)$ is still increasing along all tines of
    $\pinch{u}{F}$. For this purpose, consider the effect of this
    pinching on an individual tine $t$ terminating at a particular
    vertex $v$---it is replaced with a tine $\pinch{u}{t}$ defined so
    that:
    \begin{itemize}
    \item If $\length(t) \leq \depth(u)$, the tine $t$ is unchanged:
      $\pinch{u}{t} = t$.
    \item Otherwise, $\length(t) > \depth(u)$ and $t$ has a vertex $v$
      of depth $\depth(u) + 1$; note that $\ell(v) > \ell(u)$ because
      $F_x$ contains no vertices of depth exceeding $\depth(u)$. Then
      $\pinch{u}{t}$ is defined to be the path given by the tine
      terminating at $u$, a (new) edge from $u$ to $v$, and the suffix
      of $t$ beginning at $z$. (As $\ell(v) > \ell(u)$ this has the
      increasing label property.)
    \end{itemize}
    Thus the tree $\pinch{u}{F}$ is a legal fork on the same vertex set;
    note that the depths of vertices in $F$ and $\pinch{u}{F}$ are
    identical.
    
    \paragraph{Constructing a fork $F_y \Fork y$ containing two long tines.}
    By excising the tree rooted at $u$ from this pinched fork
    $\pinch{u}{F}$, we may extract a fork for the string
    $w_{\alpha+1} \dots w_T$. Specifically, consider the induced
    subgraph $\cut{u}{F}$ of $\pinch{u}{F}$ given by the vertices
    $\{u\} \cup \{ v \SuchThat \depth(v) > \depth(u)\}$. By treating $u$ as a
    root vertex and suitably defining the labels $\cut{u}{\ell}$ of
    $\cut{u}{F}$ so that $\cut{u}{\ell}(v) = \ell(v) - \ell(u)$, this
    subgraph has the defining properties of a fork for
    $w_{\alpha+1} \ldots w_T$. In particular, considering that
    $\alpha$ is honest, it follows that each honest index $h > \alpha$
    has depth $\hdepth(h) > \length(u)$ and hence any vertex with label $h$ 
    is also present in $\cut{u}{F}$. 
    For a tine $t$ of $\pinch{u}{F}$, we let $\cut{u}{t}$
    denote the suffix of this tine beginning at $u$, which forms a tine
    in $\cut{u}{F}$. (If $\length(t) \leq \depth(u)$, we define
    $\cut{u}{t}$ to consist solely of the vertex $u$.)  
    Considering $\cut{u}{t_1}$ and $\cut{u}{t_2}$, 
    let $\check{t}_i, i \in \{1, 2\}$ be the longest prefix of $\cut{u}{t_i}$ 
    so that $\check{t}_i$ is labeled by a slot in $y$.
    Since the tines $\cut{u}{t_1}, \cut{u}{t_2}$ are disjoint in $\cut{u}{F}$, 
    so are $\check{t}_1,\check{t}_2$. 
    
    Recall that that $y$ is as a prefix of $w_{\alpha+1} \ldots w_T$.
    Let $h^*$ be the largest honest index in $y$. 
    Let $F_y$ denote the subtree of $\cut{u}{F}$, with the same root as $\cut{u}{F}$, 
    containing the following tines: 
    $\check{t}_1, \check{t}_2$, and 
    all tines $\cut{u}{t} \in \cut{u}{F} \setminus \{\check{t}_1, \check{t}_2\}$ so that 
    $\ell(\cut{u}{t})$ is drawn from $y$ and 
    \begin{equation}\label{eq:tines-Fy}
      \length(\cut{u}{t}) \leq \hdepth(h^*)
      \,.
    \end{equation}
    Note that the length of every honest tine 
    labeled by $y$ is at most $\hdepth(h^*)$; 
    hence, thanks to~\eqref{eq:honest-depth}, 
    $F_y$ contains all honest tines from $\cut{u}{F}$ 
    that have labels in $y$. 
    Note, in addition, that the tines $\check{t}_1$ and $\check{t}_2$ 
    are consistently labeled in $F_y$. 
    Thus $F_y$ satisfies all properties of a legal fork. 
    
    Having defined $F_y$, we claim that 
    \begin{equation}\label{eq:two-long-tines}
        \min\left(\length(\check{t}_1), \length(\check{t}_2) \right) \geq \hdepth(h^*)
        \,.
    \end{equation}
    Let $i \in \{1,2\}$.
    If $\ell(t_i) < \beta$ then $\check{t}_i = \cut{u}{t_i}$ and,
    by~\eqref{eq:honest-depth}, $\length(\check{t}_i) = \length(\cut{u}{t_i}) \geq \hdepth(h^*)$. 
    Othereise, we have $\ell(t_i) = \beta$ which means 
    $\ell(t_i)$ is an honest slot. 
    Thus $\cut{u}{t_i}$ must be an honest tine, 
    building directly on top of the viable tine $\check{t}_i$. 
    Therefore, we have $\length(\check{t}_i) \geq \hdepth(h^*)$.


    \paragraph{Constructing a balanced fork $\tilde{F}_y \Fork y$.}    
    If $\length(\check{t}_1) = \length(\check{t}_2)$, set $\tilde{F}_y = F_y$ 
    and, due to~\eqref{eq:tines-Fy} and~\eqref{eq:two-long-tines}, 
    the fork $\tilde{F}_y \Fork y$ must be balanced. 
    Otherwise, 
    let $a, b \in \{1, 2\}, a \neq b$ be two integers so that 
    $\length(\check{t}_a) > \length(\check{t}_b)$. 
    We modify $F_y$ by deleting some trailing nodes from $\check{t}_a$ 
    so that the surviving prefix---let it be denoted by $\Final{t}_a$---has the same length as $\check{t}_b$. 
    That is, we achieve 
    \[
      \length(\Final{t}_a) = \length(\check{t}_b) = \min\left(\length(\check{t}_1), \length(\check{t}_2) \right)
      \,. 
    \]
    Let $\tilde{F}_y$ be the resulting fork. 
    Equations~\eqref{eq:tines-Fy} and~\eqref{eq:two-long-tines} imply that 
    $\tilde{F}_y$ has at least two maximum-length tines (i.e., $\Final{t}_a$ and $\check{t}_b$) 
    and therefore, it is balanced.
    It remains to show that the longer tine, $\check{t}_a$, 
    has sufficiently many trailing adversarial vertices so that after deleting them, 
    we obtain 
    $\length(\Final{t}_a) = \length(\check{t}_b)$. 
    (If we had to delete an honest vertex in this process, 
    $\tilde{F}_y$ may have violated 
    property~\ref{fork:unique-honest} in the definition of a fork.)    
    Let $h_a$ be the label of the last honest vertex 
    on $\check{t}_a$. 
    Thanks to~\eqref{eq:two-long-tines}, 
    we have 
    $\length(\check{t}_a) > \length(\check{t}_b) \geq \hdepth(h^*) \geq \hdepth(h_a)$. 
    Hence all vertices in $\check{t}_a$ 
    with labels in $[h_a + 1, \ell(\check{t}_a)]$ 
    must be adversarial; 
    we can safely delete $|\length(\check{t}_a) - \length(\check{t}_b)|$ 
    of these adversarial vertices.

    \paragraph{An $x$-balanced fork $\hat{F} \ForkPrefix F$.} 
    Let us identify the root of the fork $\tilde{F}_y$ with the vertex $u$ of $F_x$ and 
    let $\hat{F}$ be the resulting graph (after ``gluing'' the root of $\tilde{F}_y$ to $u$). 
    By~\eqref{eq:tu}, it is easy to see that the fork 
    $\hat{F} \ForkPrefix F$ 
    is indeed a valid fork on the string $x y$. 
    Moreover, $\hat{F}$ is $x$-balanced since $\tilde{F}_y$ is balanced. 
    The claim in Theorem~\ref{thm:divergence-settlement} follows immediately since $|y| \geq k$.
  
    \hfill$\qed$
